\documentclass[11pt,peerreview]{ieeeaccess2}

\usepackage{verbatim} % comments

\usepackage{nicefrac}
\usepackage{amsmath,amssymb,amsfonts}
\usepackage{amsthm}

\newcommand{\propnumber}{} % initialize
\newtheorem{proposition}{Proposition\propnumber}

\newcommand{\thmnewnumber}{} % initialize
\newtheorem{theorem}{Theorem\thmnewnumber}

\newcommand{\asmnumber}{} % initialize
\newtheorem{assumption}{Assumtption\asmnumber}

\newcommand{\defnumber}{} % initialize
\newtheorem{definition}{Definition\defnumber}

\newcommand{\asmnumberC}{C.\hspace{-0.25em}} % initialize

\newcommand{\exmnumber}{} % initialize

\usepackage{adjustbox}
\usepackage{multirow}
\usepackage{booktabs,tabularx,caption}

\usepackage{graphicx}
\usepackage{textcomp}

\usepackage{eso-pic}% http://ctan.org/pkg/eso-pic
\usepackage{atbegshi,picture}
\usepackage[usestackEOL]{stackengine}

\usepackage{calligra}
\DeclareMathAlphabet{\mathcalligra}{T1}{calligra}{m}{n}

\DeclareMathAlphabet{\mathlcal}{U}{dutchcal}{m}{n}
\SetMathAlphabet{\mathlcal}{bold}{U}{dutchcal}{b}{n}

\usepackage{algorithm}
\usepackage{algpseudocode}
\usepackage{cite}

\usepackage{accents}

\usepackage{setspace}  
\usepackage{scalerel}

\usepackage{relsize}

\newcommand\abs[1]{\left| {#1} \right|}
\newcommand\close[1]{\left( {#1} \right)}
\newcommand\braket[1]{\left[ {#1} \right]}
\newcommand\curly[1]{\left \{{#1} \right\}}

\newcommand\myspace{\,\,\,}

\newcommand\fracBig[2]{\frac{\mathlarger{{#1}}}{\mathlarger{{#2}}}}

\newcommand\B[1]{\boldsymbol{{#1}}}

\newcommand\disturb{\xi}

\newcommand\NestedCondE[3]{\mathbb{E}_{\B{\mathrm{{#2}}}}\braket{\mathbb{E}[\mathrm{{#1}}|\B{\mathrm{{#2}}},\mathrm{{#3}}]|\mathrm{{#3}}={#3}}}
\newcommand\CondEXY[3]{\mathbb{E}[\mathrm{{#1}|{#3}}={#3}]}

\newcommand\NestedCondIndepE[4]{\mathbb{E}_{\B{\mathrm{{#3}}}}\braket{\mathbb{E}_{\mathrm{{#1}}}[\mathrm{{#1}}|\B{\mathrm{{#3}}},\mathrm{{#4}}] \mathbb{E}_{\mathrm{{#2}}}[\mathrm{{#2}}|\B{\mathrm{{#3}}},\mathrm{{#4}}]|\mathrm{{#4}}={#4}}}

\newcommand\NestedCondIndepENosubscript[4]{\mathbb{E}_{\B{\mathrm{{#3}}}}\braket{\mathbb{E}_{\mathrm{{#1}}}[\mathrm{{#1}}|\B{\mathrm{{#3}}},\mathrm{{#4}}] \mathbb{E}[\mathrm{{#2}}|\B{\mathrm{{#3}}},\mathrm{{#4}}]|\mathrm{{#4}}={#4}}}

\newcommand{\norm}[2]{\left\lVert#1\right\rVert_{#2}}
\newcommand{\normsq}[2]{\Big\lVert{#1}\Big\rVert_{#2}^{2}}

\makeatletter
\newcommand{\normd}[1][{1}]{
	\@ifnextchar[{\normd@i[{#1}]}{\normd@i[{#1}][{#1}]} }
\def\normd@i[#1][#2]#3{
	\scaleobj{#1}{\left\|\right.} {#3} \scaleobj{#2}{\left.\right\|}
}
\makeatother

\makeatletter
\def\Prob{\@ifnextchar[{\@with}{\@without}}
\def\@with[#1]#2{\myProbTwo{#2}{#1}}
\def\@without#1{\myProb{#1}}
\makeatother

\makeatletter
\def\ProbS{\@ifnextchar[{\@withS}{\@withoutS}}
\def\@withS[#1]#2{\myProbSTwo{#2}{#1}}
\def\@withoutS#1{\myProbS{#1}}
\makeatother

\makeatletter
\def\NonProb{\@ifnextchar[{\@withNon}{\@withoutNon}}
\def\@withNon[#1]#2{\myNonProbTwo{#2}{#1}}
\def\@withoutNon#1{\myNonProb{#1}}
\makeatother

\newcommand\myProb[1]{\mathrm{Pr}(\mathrm{y_{2}^*}\ge 0|\B{\mathrm{{#1}}={#1}_{g_i}})}

\newcommand\myProbTwo[2]{\mathrm{Pr}(\mathrm{y_{2}^*}\ge 0|\B{\mathrm{{#1}}={#1}_{g_i},\mathrm{{#2}}={#2}_i})}

\newcommand\myProbS[1]{\mathrm{Pr}(\mathrm{S}=1 |\B{\mathrm{{#1}}={#1}})}

\newcommand\myProbSTwo[2]{\mathrm{Pr}(\mathrm{S}=1|\B{\mathrm{{#1}}={#1}_{g_i},\mathrm{{#2}}={#2}_i})}

\newcommand\myNonProb[1]{\mathrm{Pr}(\mathrm{y_{2}^*}< 0|\B{\mathrm{{#1}}={#1}_{g_i}})}

\newcommand\myNonProbTwo[2]{\mathrm{Pr}(\mathrm{y_{2}^*}< 0|\B{\mathrm{{#1}}={#1}_{g},\mathrm{{#2}}={#2}_i})}

\usepackage{accents}

\usepackage{setspace}  
\usepackage{scalerel}

\usepackage{relsize}

\IEEEoverridecommandlockouts

\begin{document}
	%This article has been accepted for publication in a future issue of this journal, but has not been fully edited. Content may change prior to  final publication.\\ Citation information: DOI 10.1109/ACCESS.2019.2929571, IEEE Access 
	\history{This article has been accepted for publication in a future issue of this journal, but has not been fully edited. Content may change prior to  final publication. Citation information: DOI 10.1109/ACCESS.2019.2929571, IEEE Access. }
	%\doi{}
%\history{\tiny{\textcopyright  $ $ 2019 IEEE.  Personal use of this material is permitted.  Permission from IEEE must be obtained for all other uses, in any current or future media, including reprinting/republishing this material for advertising or promotional purposes, creating new collective works, for resale or redistribution to servers or lists, or reuse of any copyrighted component of this work in other works.}}	
%	\history{\tiny{\textcopyright$ $ 2019 IEEE.  Personal use of this material is permitted.  Permission from IEEE must be obtained for all other uses, in any current or future media, including reprinting/republishing this material for advertising or promotional purposes, creating new collective works, for resale or redistribution to servers or lists, or reuse of any copyrighted component of this work in other works.\hfill\\ DOI: 10.1109/ACCESS.2019.2929571, IEEE Access.} }
	\doi{10.1109/ACCESS.2019.2929571}
	
	%This work is licensed under a Creative Commons Attribution 4.0 License. For more information, see https://creativecommons.org/licenses/by/4.0/.
	
	\title{S\lowercase{emiparametric} W\lowercase{avelet-based} JPEG IV E\lowercase{stimator for} e\lowercase{ndogenously} t\lowercase{runcated} d\lowercase{ata}}
	\author{\uppercase{Nir Billfeld}\authorrefmark{1},
		\uppercase{Moshe Kim}\authorrefmark{2} %
	}

	\address[1]{University of Haifa, Haifa, Israel (e-mail: nbillfeld@staff.haifa.ac.il)}
	\address[2]{University of Haifa, Haifa, Israel (e-mail: kim@econ.haifa.ac.il)}
	
	%\markboth
	%{Author \headeretal: Preparation of Papers for IEEE TRANSACTIONS and JOURNALS}%
	%{Author \headeretal: Preparation of Papers for IEEE TRANSACTIONS and JOURNALS}
	
	\markboth{}{}
%	{\chead{\small This article has been accepted for publication in a future issue of this journal, but has not been fully edited Content may change \\prior to  final publication. Citation information: DOI 10.1109/ACCESS.2019.2929571, IEEE Access }}
%	{\chead{\small This article has been accepted for publication in a future issue of this journal, but has not been fully edited Content may change \\prior to  final publication. Citation information: DOI 10.1109/ACCESS.2019.2929571, IEEE Access }}

	\corresp{Corresponding author: Moshe kim (e-mail: kim@econ.haifa.ac.il).\\ 
	This work was supported by the Research Authority, The University of Haifa.\\
	The authors would like to thank seminar participants of the faculty of Mathematics and Computer Science, The Weizmann Institute of Science and Statistic departments of Tel-Aviv and Haifa universities, for very constructive comments.}
	\begin{abstract}
A new and an enriched JPEG algorithm is provided for identifying redundancies in a sequence of irregular noisy data points which also accommodates a reference-free criterion function. Our main contribution is by formulating analytically (instead of approximating) the inverse of the transpose of JPEG-wavelet transform without involving matrices which are computationally cumbersome. The algorithm is suitable for the widely-spread situations where the original data distribution is unobservable such as in cases where there is deficient representation of the entire population in the training data (in machine learning) and thus the covariate shift assumption is violated.  The proposed estimator corrects for both biases, the one generated by endogenous truncation and the one generated by endogenous covariates. Results from utilizing 2,000,000 different distribution functions verify the applicability and high accuracy of our procedure to cases in which the disturbances are neither jointly nor marginally normally distributed.	
	\end{abstract}
	
	\begin{keywords}
		JPEG, Semiparametric, biorthogonal wavelet, causality, proximal gradient-descent, Lifting scheme, Denoising, Covariate Shift, Training data, Reference-free
		%Enter key words or phrases in alphabetical order, separated by commas. For a list of suggested keywords, send a blank e-mail to keywords@ieee.org or visit \underline {http://www.ieee.org/organizations/pubs/ani\_prod/keywrd98.txt}
	\end{keywords}
	
	\titlepgskip=-15pt
	
	\maketitle

	\IEEEpeerreviewmaketitle 
	\section{Introduction}
	\linespread{0.8}
Scientists routinely try to model and extract causal relations among covariates, rather than merely their correlations.\footnote{For a specific form of causality due to treatment effect see \cite{angrist2004semiparametric} definition.} In practice however, the presence of endogenous covariates in the model challenges the causal inference due to comovement of the random disturbance with these covariates. We distinguish between population induced comovement and training data (as in machine learning) comovement without having to rely on a covariate shift assumption since the behavioral (causal) model embedded in the training data does not necessarily describing the behavior in the entire population (see discussion in \cite{billfeld2019semiparametric}). 

The common way to overcome the aforementioned, is to generate a variation in the endogenous covariate without introducing variation in the random disturbance. This idea is achieved by employing a proper instrumental variable (IV).\footnote{Note that we deal with endogenously \textit{truncated} sample selection model to differentiate from \textit{censored} sample selection models \cite{heckman1979sample,newey2009two,powell2001semiparametric}, where there exists information pertaining to the non-participants.}

Application of  a proper instrumental variable generates variation in the endogenous covariate without  introducing variation in the random disturbance and hence is orthogonal to it. Thus, IV should contribute to exogeneity and therefore has been extremely popular in empirical work.

Once we have analytically shown that the IV estimator is no longer valid in an \textit{endogenously truncated} environment, we offer a truncation-proof estimator, which is a semiparametric wavelet-based JPEG-IV denoising algorithm.\footnote{A wavelet is a bandwidth-free estimator that is based on a multi-scale representation of the data. It is a widely used denoising technique \cite{xie2002sar}.} This denoising algorithm decomposes the random disturbance into a noise and a systemic bias part, enabling the elimination of the truncation bias. The magnitude of the this bias is captured by the size of the wavelet coefficients which quantify and measure the degree of redundancy hidden in a sequence of data points. Consequently, this algorithm nests the conventional IV estimator as a special case due to the fact that in the absence of systemic endogenous truncation, the wavelet coefficients approach zero except for the intercept which describe the coarse level of the function.\footnote{Unlike Fourier transform, the wavelet estimator preservers not only the data average (coarse) behavior but also its local behavior capturing deviations (details) from the average. This fact renders our denoising suitable for irregular-spaced data which largely depend on local behavior and play an important role in the denoising.} %Although this decomposition can also be achieved by employing the Fourier transform, the wavelet estimator has an important advantage. Our proposed wavelet-based JPEG preserves both local information as well as average behavior, while Fourier as a global estimator captures only average behavior. This fact renders our denoising suitable for irregular-spaced data largely depending on local behavior  and plays an important role in denoising.

Our main contribution to the biorthogonal wavelet estimator is by formulating analytically (instead of approximating) the inverse of the transpose of wavelet transform without involving matrices which are computationally cumbersome \cite{voronin2015compression}\footnote{``The implemented routine for the inverse transpose transform is approximate."
	%``We cannot obtain the inverse-transpose matrix $W^T$ by transposing the inverse of $W$. This is because $W$ is a very large $n\times n$ matrix and is very costly to build for large n. Hence, we instead use a routine for applying $W$ and $W^T$ to vectors...implemented routine for the inverse transpose transform is approximate."
	\cite{voronin2015compression}, p.285.} as well as the management of irregular-spaced data.\footnote{In the orthogonal wavelets design various interpolation methods are used to alleviate these irregularities \cite{hall1997interpolation} and specific methodologies can be used to extend the Haar wavelet transform to the unequally spaced case \cite{sardy1999wavelet}.} Additionally, our proposed methodology enables the combination of several penalty functions in the estimation procedure which are resolution-dependent.\footnote{It is known that soft thresholding provides smoother results relative to the hard thresholding because it is continuous. The latter, however, provides better edge preservation in comparison with the former.}

Wavelets are useful in denoising data. %signal and image denoising. %\footnote{Alternative estimators such as kernel and splines do not resolve local structures well enough.
%}$^{,}$
%\footnote{Denoising is an algorithm employed to recover the original signal (or function) from a noisy signal. The various denoising procedures are closely related to the estimation of a nonparametric regression with an additive disturbance of unknown distribution function. However, the most important aspect in recovering a signal from a noisy signal is finding the optimal thresholding (or tuning) parameter which controls the degree of denoising. This is useful to avoid over-fitting as well as under-fitting.} Denoising is affected both by the function that applies constraints on the coefficients in the estimation procedure, as well as by the tuning parameter. The former is referred to as a penalty function, while the latter determines the degree of penalization. Therefore, we must draw our attention to the choice of these factors to optimize the denoising. %\footnote{Choosing a larger (smaller) thresholding parameter than the optimal implies over (under) smoothing. }  
Several image quality assessment (IQA) measures have been introduced to choose the optimal level of denoising. %image quality assessment (IQA) measures can be employed to measure the quality of the denoised signal.%, for image quality assessment (IQA) or evaluation of image quality (EIQ) which are employed to evaluate the quality of the compressed or denoised image. 
These approaches can be classified to full-reference (FR) in cases the original image (noise-free) is observed; reduced-reference (RR) in cases where there is a partial information about the reference; reference-free (RF) in cases where the original image is not accessible \cite{carnec2003full,farah2014full}.  
As we deal with truncated distributions, the source (the complete non-truncated distribution) is intrinsically unobservable and thus, we cannot assess the success of the denoising by comparing it to the original non-truncated distribution.  Therefore, we  select both the thresholding (tuning) parameter as well as the penalty function using a reference-free criterion function.

The proposed JPEG IV is biorthogonal, thus preserving both the symmetry (the original shape of the data) and compact support (small number of coefficients) properties of the data.\footnote{Our JPEG IV is a biorthogonal wavelet as it requires two sets of vectors, which are the \textit{dual basis} and the \textit{series expansion} sets, to obtain a denoised representation of the data. The elements in the former set are orthogonal to the corresponding elements in the latter set. See \cite{cohen1992biorthogonal} for a formal definition of biorthogonality.} Importantly, the proposed methodology is easy to compute by precluding the need to find an optimal bandwidth as conventionally done.\footnote{Kernel estimation involves computational burden due to the necessity of finding the optimal bandwidth \cite{ichimura1993semiparametric}. Unlike the nonparametric case, in the semiparametric context there is no ``protocol" for finding the optimal bandwidth, as the traditional bandwidth choice methods might lead to bias estimates due to improper bandwidth choice \cite{lewbel2007nonparametric}.}  These properties make it suitable for denoising by alleviating both the problem of coefficient expansion as well as border discontinuities \cite{usevitch2001tutorial}.      %formed by using a dual basis set
%A biorthogonal wavelet is a pair of family of dual wavelets derived from two mother wavelets,
%A biorthogonal wavelet is formed by two mother functions $\psi$ and $\tilde{\psi}$ and two scale (father) functions $\phi$ and $\tilde{\phi}$, unlike the orthogonal wavelet which is formed by a single mother and single father.\footnote{The former is a generalization of the latter, in which the mother wavelets do not form an orthogonal bases. See \cite{cohen1992biorthogonal} 
%for a formal definition of biorthogonality.}
%The key difference between orthogonal and biorthogonal wavelets is the fact that in the biorthogonal design the data are decomposed by one wavelet  and scale function (the forward transform) and are reconstructed (the inverse transform) by another, while in the orthogonal design the decomposition and reconstruction coincide \cite{cohen1992biorthogonal, karoui1994families}.
The proposed algorithm corrects for both sources of bias: the endogeneity of covariates as well as the endogenous self-selection biases.% is nested within our proposed JPEG IV estimator.%The selection bias term is estimated by denoising. 

We run Monte Carlo simulations to measure the magnitude of the potential bias in the parameters' estimates under endogenous truncation, obtained by employing a conventional IV to eliminate the endogeneity bias. Our empirical implementation shows that even under mild correlation between the random disturbances, the resulting bias in the estimated  parameter of the endogenous covariate in the substantive equation can amount to almost tenfold the true parameter value. %the mean of the estimated parameter of the endogenous covariate in the substantive equation can reach ten-fold the true parameter.  %\color{blue} the mean of the estimated parameter of the endogenous covariate in the substantive equation can range from one-tenth to one-fifth of the true parameter. \color{black}
%resulting bias
%in the estimated  parameter of the endogenous covariate in the substantive equation can amount to almost $10$ times smaller than the true parameter value for 500 observations and can amount to $5$ times smaller than the true parameter value in a sample of 10,000 observations. 
Further, for sake of generality of the offered estimator, we subject it to various distributions in which the disturbances are neither jointly nor marginally normally distributed. These disturbances are constructed as realizations of non-symmetric and non-unimodal distribution functions.\footnote{Unlike the practice in some other studies applying only normally distributed disturbances.}
 
The rest of this paper is organized as follows. The methodology is presented is section \ref{Section:Methodology}. Section \ref{Section:Preliminatries} prepares the ground for the biorthogonal wavelet. Section \ref{Section:JPEG:denoising} presents our proposed JPEG algorithm. In section \ref{Section:Simulation} we employ Monte Carlo simulations to validate our estimator performance. Section \ref{Section:Conclusion} concludes.

\section{Methodology}\label{Section:Methodology}

As discussed above, the IV is based on the following basic requirements: it is correlated with the endogenous covariate, as well as orthogonal to the random disturbance. Additionally, it must satisfy the exclusion restriction, such that in the presence of the endogenous covariate, the IV must be excluded from the regression.  The IV is allowed to affect  the dependent variable only through its effect on the endogenous covariate. However, the orthogonality condition is rarely satisfied in the presence of endogenous truncation, which is very frequently the nature of data used in empirical research, and therefore the IV will not provide a solution for the endogeneity problem. In what follows, we demonstrate the shortcoming of the conventional IV estimator, as well as potential bias generated in an environment of endogenously truncated data. %a consistent estimator of the effect 

Suppose that there is a population random variable $\mathbf{\mathlcal{w}}=(\mathrm{z};\mathrm{x_{1}},\B{\mathrm{x_{-1}}};\B{\mathrm{w}})$ and that there is an independent and identically distributed sample $\curly{z_i,x_{1i},\B{x_{-1_i}}, \B{w_i}}_{i=1}^N$ drawn from this population, referred to as the complete data set consisting of $N$ observations.\footnote{Capital letters indicate random variables; lower case letters indicate realizations of these random variables.} The instrumental variable is $\mathrm{z}$, the endogenous variable is $\mathrm{x_{1}}$ and the exogenous random variables are $(\B{\mathrm{x_{-1}};\mathrm{w}})$, and where $\B{\mathrm{w}}\in \mathbb{R}^l$ is a covariate vector.

%Denoting an instrumental variables set by $\B{z_i}$, 

Let $\disturb_{1i}$, $\disturb_{2i}$ and $\mathrm{v}_i$ be jointly dependent random disturbances with the respective marginal distribution functions $F_{\disturb_1}$, $F_{\disturb_2}$ and $F_{\mathrm{v}}$. Their joint distribution function is $F_{\disturb_1, \disturb_2,\mathrm{v}}$. The model is semiparametric, as neither the marginals nor the joint distribution function are required to be specified by the researcher. 

The underlying model is composed of two parts.  The first part consists of a selection equation, while the second part consists of the substantive (of interest) equation.

%Let $\disturb_{1i}$ and $\disturb_{2i}$ be the two jointly dependent random disturbances. The underlying model consists of the latent (population) regressions \eqref{Baseline:Model} and \eqref{Baseline:Model:Selection} to follow:

%where $\B{x}_i=[x_{1i},\B{x}_{{-1}_i}^T]^T$ is a covariate vector consisting of a $p\times 1$ vector of exogenous covariates $\B{x}_{-1i}$ and an endogenous covariate $x_{1i}$ implying that $\mathbb{E}[\disturb_{1i}|x_{1i}]\ne 0$. 

The population (non-truncated) selection equation is defined as:
\begin{IEEEeqnarray}{lCr}
	\label{Nontruncated:Selection:Equation}
	y_{2i}^*=\B{w}_i^T\B{\gamma} + \disturb_{2i}
\end{IEEEeqnarray}

where $\B{\gamma}\in\mathbb{R}^l$ and $\B{w_i}\in\mathbb{R}^l$ are the selection equation's coefficients and covariates vector,  respectively. The selection equation's random disturbance is denoted by $\disturb_{2i}$.

The substantive equation and the endogenous variable equation are defined as a system of equations:
\begin{IEEEeqnarray}{lCr}
	\label{Baseline:Model}
	\scaleobj{0.75}{\color{white}\begin{cases}\color{black}\braket{\begin{matrix}
					y_{1i}^*\\ x_{1i}^*
			\end{matrix}}=\braket{\begin{matrix}
					\B{x}_i^T\\ [\B{z}_i^T,\B{x}_{-1_i}^T]
			\end{matrix}}\braket{\begin{matrix}
					\B{\beta}\\ \B{\delta}
			\end{matrix}}+\braket{\begin{matrix}
					\disturb_{1i}\\ \mathrm{v}_{i}
			\end{matrix}} \hspace{1em}\begin{matrix}\text{the substantive equation  \color{white}variable\color{black}}\\\text{the endogenous variable equation}\end{matrix}
		\end{cases}  \color{black}}%\text{the substantive system of equations}\\ %\text{\shortstack
\end{IEEEeqnarray}

where $\B{\beta}\in\mathbb{R}^{p_1}$ and $\B{\delta}\in\mathbb{R}^{p_2}$ are covariates vectors, $x_{1i}$ is an endogenous variable included in vector $\B{x}_i\in\mathbb{R}^{p_1}$, and the exogenous variables are denoted by $\B{x}_{-1_i}^T$. The substantive equation's random disturbances are denoted by $\disturb_{1i}$ and $\mathrm{v}_{1i}$.

However the variables $y_{1i}^*, y_{2i}^*, x_{1i}^*$ are latent in the truncated environment and their respective observed realizations are denoted by $y_{1i}, y_{2i}, x_{1i}$,  defined in  \eqref{selection:equation} and \eqref{Substantive:Equation} to follow.

The variable $y_{2i}^*$ is latent, while $y_{2i}$ is observed and defined as:
\begin{IEEEeqnarray}{lCr}
	\label{selection:equation}
	\scaleobj{0.80}{y_{2i} = \begin{cases}
			1 & \text{if } y_{2i}^*\ge 0\\
			\text{Unobserved} & \text{if } y_{2i}^*< 0\\
		\end{cases}, \hspace{1em}\text{the selection equation}}
\end{IEEEeqnarray}

\begin{IEEEeqnarray}{lCr}
	\scaleobj{0.80}{\label{Substantive:Equation}
		\braket{\begin{matrix}
				y_{1i}\\ 
				x_{1i}\\
		\end{matrix}}= \begin{cases}
			\braket{\begin{matrix}
					y_{1i}^*\\
					x_{1i}^*\\
			\end{matrix}} & \text{if } y_{2i}^*\ge 0\\
			\text{Unobserved} & \text{if } y_{2i}^*< 0\\
		\end{cases}, \hspace{1em}\text{the substantive equations}}
\end{IEEEeqnarray}

%\newcommand\NestedCondE[3]{\mathbb{E}_{\B{\mathrm{{#2}}}}\braket{\mathbb{E}_{\mathrm{{#1}}}[\mathrm{{#1}}|\B{\mathrm{{#2}}},\mathrm{{#3}}]|\mathrm{{#3}}={#3}}}

%We denote an indicator variable $\mathrm{s}_i$ defined as:
%\begin{IEEEeqnarray}{lCr}
%\label{Population:Selection:Var}
%\mathrm{s}_i=\begin{cases}
%1 & y_{2i}^*\ge 0 \\
%0 & y_{2i}^*< 0 \\
%\end{cases}
%\end{IEEEeqnarray}
In the next section we reformulate the substantive equation as a partially linear single index model. %attend to this issue.
\subsection{Semiparametric selectivity bias correction}
%Bias correction in truncated environment by removing the contamination in the substantive equation's random disturbance
The key difference between censored and truncated sample selection models is that in the former the entire covariate set (including the non-participants) and the selection variable are fully observed. In the latter, the entire data are truncated.  Nevertheless, in both cases, the substantive equation can be represented as a partially linear regression, in which the dependent variable is observed only for the participants, as we are about to show. 
%Thus, the substantive equation can be modeled under endogenous truncation by utilizing the approach prevailed in the censored case, which is reformulating it as a partially linear regression we are about to present.   \color{black}%This reformulation enables to control for the selectivity bias. 
Following \cite{robinson1988root}, the conditional expectation of the substantive equation in semiparametric (censored)\footnote{His approach is a generalization of the well-known inverse-mills ratio estimator  introduced by \cite{heckman1979sample} for the substantive equation's bias term $\mathbb{E}\braket{\disturb_{1i}|\disturb_{2i} > -\B{w}_i^T\B{\gamma}}$ in the case of a censored sample selection model. Note the difference between censored data and truncated data, which is the case we deal with.} sample selection models is some generally unknown function $\mathcal{M}_1(.)$ (to be estimated) of the selection equation's covariates variables $\B{w_i}$:
\begin{IEEEeqnarray}{lCr}
	\label{Robinson:Equation}
	\mathbb{E}\braket{\disturb_{1i}|\disturb_{2i} > -\B{w}_i^T\B{\gamma}}=\mathcal{M}_1(\B{w}_i^T\B{\gamma})
\end{IEEEeqnarray}

such that $\B{\gamma}$ is the selection equation's coefficient vector. Since $y_{1i}$ is observed only if $i$ is a participant, the substantive equation's dependent variable obtains the following functional form:
\begin{IEEEeqnarray}{lCr}
	\label{Y:Bias:Term}
	y_{1i}=\B{x}_i^T\B{\beta}\hspace{0.1em} + \underbrace{\mathcal{M}_1(\B{w}_i^T\B{\gamma})}_\text{the bias term}\hspace{1em}+ \hspace{0.1em}\underbrace{\tilde{\epsilon}_{1i}}_{\text{white noise}}
\end{IEEEeqnarray}

The regression equation in \eqref{Y:Bias:Term} is referred to as a semiparametric partially linear regression (SP-NLS), in which the non-linear part is the bias term function. This regression can be estimated semiparametrically in cases of a truncated sample selection model using a non-linear least squares procedure as suggested by \cite{ichimura1993semiparametric}.

Both \cite{ichimura1993semiparametric} and \cite{robinson1988root} models involve a kernel function estimation. %(a methodology that will be described in section \ref{Section:Kernel} to follow). 
However, kernel estimates' accuracy is sensitive to the bandwidth selected. This entails a potential problem of finding the optimal bandwidth resulting in computational complexity.\footnote{There is an open question whether there is a way to choose a bandwidth
	sequence that is optimal for the estimation of the parameters \cite{ichimura1991semiparametric}.} 
%One of the drawbacks of the above-described methodology is that the kernel estimates' accuracy depends on the bandwidth selection, which is known problem.
Due to the lack of applicability of the traditional bandwidth selection methods in the semiparametric context, informal methods are being used, that may lead to a non-ignorable bias in the estimates \cite{lewbel2007simple}.\footnote{``The well
	known bandwidth selection rules used in non-parametric estimation, such as cross validation,
	are not generally applicable to semiparametric settings.'' \cite[p.~191]{lewbel2007simple}} %Thus, in practice, one needs to use a bandwidth that is ``slightly'' smaller than the optimal bandwidth obtained using the cross-validation procedure \cite{lewbel2007simple}.} %However, this informal method for bandwidth choice may lead to a non-ignorable bias in the estimates \cite{lewbel2007simple}.} %\footnote{There is no ``scientific protocol'' enabling the determination of a proper bandwidth (in semiparametric estimation). It is known, though, that incorrect bandwidth choice leads to bias estimates \cite{lewbel2007simple}. } 
In order to avoid the problems involved with kernel estimation, our methodology relies on a (thresholding-propagated) nonlinear wavelet-based JPEG IV estimator to approximate the bias term (in \eqref{Y:Bias:Term}). 

%why decomposition? 
The substantive equation depicted in \eqref{Y:Bias:Term} deals with endogenous truncation bias, assuming that the random disturbance and the covariates are not jointly dependent. However, in cases where this random disturbance is jointly dependent with one (or more) of the covariates there will emerge two bias terms: the first one is propagated by the endogenous truncation and the second one is propagated by the endogenous covariate.
Next we present a decomposition Theorem \ref{Theorem:Decomp}, which enables reformulating the substantive equations as a partially linear single index model in the presence of an endogenous covariate. 
\subsection{Decomposition of the substantive equations}
\begin{theorem}\label{Theorem:Decomp}
	Let the underlying model be as depicted in \eqref{selection:equation} and \eqref{Substantive:Equation}. Denote the  random disturbances $\varepsilon_{i}$ and $\epsilon_{1i}$ which are constructed as:  $\varepsilon_{i}=y_{1i}^*-\mathbb{E}[y_{1i}^*|\B{x}_i]$ and $\epsilon_{1i}=y_{1i}-\mathbb{E}[y_{1i}^*|y_{2i}=1]$, respectively. 
	The following requirements must hold:\\ \nonumber \\ (i) $\mathbb{E}[y_{1i}^*|y_{2i}=1]=\mathbb{E}[\B{x}_i^T\B{\beta}\hspace{0.2em}|\hspace{0.2em}\B{x}_i]+\mathbb{E}[\disturb_{1i}\hspace{0.2em}|\hspace{0.2em}\B{x}_i]+\mathbb{E}[\varepsilon_{i}|y_{2i}=1]$ $\forall i\in\curly{1,...,N}$;\\\\ \nonumber (ii) $y_{1i}=\B{x}_i^T\B{\beta}+\epsilon_{1i}^*$,\hspace{0.5em} $\mathbb{E}[\epsilon_{1i}\big|\hspace{0.5em}y_{2i}=1]=0$,\\ $\hspace{0.5em}\epsilon_{1i}^*\equiv\epsilon_{1i}+\mathbb{E}[\disturb_{1i}\hspace{0.2em}|\hspace{0.2em}\B{x}_i]+ \mathcal{M}(\B{w_i}^T\B{\gamma})$,  $\hspace{0.5em}\forall i\in\curly{i|y_{2i}=1}$. \nonumber
\end{theorem}

\begin{proof}%[\ref{Theorem:Decomp}]
	By construction $\varepsilon_{i}=y_{1i}^*-\mathbb{E}[y_{1i}^*|\B{x}_i]$, it follows that
	\begin{IEEEeqnarray}{lCr}
		\label{implicit:expect:part:one}
		y_{1i}^*= \mathbb{E}[\B{x}_i^T\B{\beta}\hspace{0.2em}|\hspace{0.2em}\B{x}_i]+\mathbb{E}[\disturb_{1i}\hspace{0.2em}|\hspace{0.2em}\B{x}_i]+\varepsilon_{i}
	\end{IEEEeqnarray}

	%	Rearranging \eqref{implicit:expect:part:one} and using the fact that $\mathbb{E}[\B{x}_i^T\B{\beta}\hspace{0.2em}|\hspace{0.2em}\B{x}_i]=\B{x}_i^T\B{\beta}$ we get:
	%	\begin{IEEEeqnarray}{lCr}
	\label{explicit:expect:part:one}
	%	y_{1i} = \B{x}_i^T\B{\beta}+\mathbb{E}[\disturb_{1i}\hspace{0.2em}|\hspace{0.2em}\B{x}_i] + \varepsilon_{i}
	%	\end{IEEEeqnarray}

	%	By law of iterated expectation it follows that
	%	\begin{IEEEeqnarray}{lCr}
	
	%	\mathbb{E}[\varepsilon_{i}\hspace{0.2em}|\hspace{0.2em}\B{x}_i]= \mathbb{E}[\disturb_{1i}\hspace{0.2em}|\hspace{0.2em}\B{x}_i]-\mathbb{E}\curly{\mathbb{E}[\disturb_{1i}\hspace{0.2em}|\hspace{0.2em}\B{x}_i]\hspace{0.2em}|\hspace{0.2em}\B{x}_i}=\mathbb{E}[\disturb_{1i}\hspace{0.2em}|\hspace{0.2em}\B{x}]-\mathbb{E}[\disturb_{1i}\hspace{0.2em}|\hspace{0.2em}\B{x}_i]=0.
	%	\end{IEEEeqnarray}

	%Due to the fact that $\epsilon_{1i}=\varepsilon_{i}-\mathbb{E}[\varepsilon_{i}|\mathrm{s}=1]$ it follows that:
	%	\begin{IEEEeqnarray}{lCr}
	
	%\mathbb{E}[\epsilon_{1i}]=\mathbb{E}\curly{\varepsilon_{i}-\mathbb{E}[\varepsilon_{i}|\mathrm{s}=1]}=\mathbb{E}[\varepsilon_{i}]-\mathbb{E}[\varepsilon_{i}|\mathrm{s}=1]
	%	\end{IEEEeqnarray}
	
	Using \eqref{implicit:expect:part:one} we get:% $\mathbb{E}[\varepsilon_{i}|y_{2i}=1]=\mathbb{E}\curly{y_{1i}^* - \mathbb{E}[\B{x}_i^T\B{\beta}\hspace{0.2em}|\hspace{0.2em}\B{x}_i]-\mathbb{E}[\disturb_{1i}\hspace{0.2em}|\hspace{0.2em}\B{x}_i]|y_{2i}=1}$, implying that
	\begin{IEEEeqnarray}{lCr}
		\mathbb{E}[y_{1i}^*|y_{2i}=1]=\mathbb{E}[\varepsilon_{i}|y_{2i}=1]\\+\mathbb{E}\curly{ \mathbb{E}[\B{x}_i^T\B{\beta}\hspace{0.2em}|\hspace{0.2em}\B{x}_i]|y_{2i}=1}+\mathbb{E}\curly{ \mathbb{E}[\disturb_{1i}\hspace{0.2em}|\hspace{0.2em}\B{x}_i]|y_{2i}=1}\nonumber
	\end{IEEEeqnarray}
	
	which is simplified to:
	\begin{IEEEeqnarray}{lCr}
		\label{Expected:Subs:Truncated}
		\mathbb{E}[y_{1i}^*|y_{2i}=1]=\mathbb{E}[\B{x}_i^T\B{\beta}\hspace{0.2em}|\hspace{0.2em}\B{x}_i]+\mathbb{E}[\disturb_{1i}\hspace{0.2em}|\hspace{0.2em}\B{x}_i]\\+\mathbb{E}[\varepsilon_{i}|y_{2i}=1]\nonumber
	\end{IEEEeqnarray}

	In order to obtain the substantive equation in the truncated environment, we construct $\epsilon_{1i}=y_{1i}-\mathbb{E}[y_{1i}^*|y_{2i}=1]$ where $\mathbb{E}[y_{1i}^*|y_{2i}=1]$ is obtained from \eqref{Expected:Subs:Truncated}.\footnote{By construction of $y_{1i}$, the equality $\mathbb{E}[y_{1i}|y_{2i}=1]=\mathbb{E}[y_{1i}^*|y_{2i}=1]$ must be satisfied. It implies that $\mathbb{E}[\epsilon_{1i}|y_{2i}=1]=\mathbb{E}[y_{1i}|y_{2i}=1]-\mathbb{E}\curly{\mathbb{E}[y_{1i}^*|y_{2i}=1]|y_{2i}=1}=\mathbb{E}[y_{1i}|y_{2i}=1]-\mathbb{E}[y_{1i}^*|y_{2i}=1]=0$. }  Following \cite{ichimura1991semiparametric}, the conditional expectation of $\varepsilon_{i}$, given participation is expressed by some unknown function $\mathcal{M}_1(\cdot)$  as $\mathbb{E}[\varepsilon_{i}|y_{2i}=1]=\mathcal{M}_1(\B{w}_i^T\B{\gamma})$. Thus, we obtain:
	%    \begin{IEEEeqnarray}{lCr}
	\label{implicit:expect:part:two}
	%   	y_{1i}=\mathbb{E}[\varepsilon_{i}|\mathrm{s}=1]+\B{x}_i^T\B{\beta}+\mathbb{E}[\disturb_{1i}\hspace{0.2em}|\hspace{0.2em}\B{x}_i]+\epsilon_{1i}
	%    \end{IEEEeqnarray}
	
	%	\begin{IEEEeqnarray}{lCr}
	\label{implicit:expect:part:two}
	%	\epsilon_{1i}=\varepsilon_{i}-\mathbb{E}[\B{x}_i^T\B{\beta}|y_{2i}=1]-\mathbb{E}[\mathbb{E}[\disturb_{1i}\hspace{0.2em}|\hspace{0.2em}\B{x}_i]|y_{2i}=1]-\mathbb{E}[\disturb_{1i}|y_{2i}=1]
	%	\end{IEEEeqnarray}
	
	%	Using the law of total expectation it follows that $\mathbb{E}[\mathbb{E}[\disturb_{1i}\hspace{0.2em}|\hspace{0.2em}\B{x}_i]|y_{2i}=1]=\mathbb{E}[\disturb_{1i}\hspace{0.2em}|\hspace{0.2em}\B{x}_i]$. Rearranging \eqref{implicit:expect:part:two} we get:
	%	\begin{IEEEeqnarray}{lCr}
	\label{explicit:expect:part:two}
	%	y_{1i} = \B{x}_i^T\B{\beta}+\mathbb{E}[\disturb_{1i}\hspace{0.2em}|\hspace{0.2em}\B{x}_i] +\mathbb{E}[\varepsilon_{i}|y_{2i}=1] + \epsilon_{1i}
	%	\end{IEEEeqnarray}
	
	\begin{IEEEeqnarray}{lCr}
		\label{final:expect:part:two}
		\scaleobj{0.85}{y_{1i} = \underbrace{\B{x}_i^T\B{\beta}}_{\text{\shortstack{substantive\\ covariates}}}\hspace{0.2em}+\overbrace{\hspace{0.2em}\underbrace{\mathcal{M}_1(\B{w}_i^T\B{\gamma})}_{\text{\shortstack{selection bias\\ term}}} \hspace{0.2em}+\hspace{0.2em}\underbrace{\mathbb{E}[\disturb_{1i}\hspace{0.2em}|\hspace{0.2em}\B{x}_i]}_{\text{\shortstack{endogeneity bias\\ term}}}\hspace{0.5em} + \underbrace{\epsilon_{1i}}_{\text{white noise}}}^{\epsilon_{1i}^*}}
	\end{IEEEeqnarray}

	For sake of brevity we present equation \eqref{final:expect:part:two}, which is a decomposition of the substantive equation into its components, such as the substantive equation's covariates, selection bias term, endogeneity bias term and a stochastic white noise term.
	It is easy to see that the conventional IV can not be sufficient in eliminating the endogeneity bias $\mathbb{E}[\disturb_{1i}\hspace{0.2em}|\hspace{0.2em}\B{x}_i]$ in \eqref{final:expect:part:two}, since under truncation the endogeneity bias term is actually $\mathbb{E}[\disturb_{1i}\hspace{0.2em}|\hspace{0.2em}\B{x}_i,y_{2i}=1]$. 
	
	Similarly, we construct $\epsilon_{2i}=x_{1i}-\mathbb{E}[x_{1i}^*|y_{2i}=1]$ where $\mathbb{E}[x_{1i}^*|y_{2i}=1]$ satisfies:
	\begin{IEEEeqnarray}{lCr}
		\label{Expected:Subs:Truncated:X}
		\scaleobj{0.8}{\mathbb{E}[x_{1i}^*|y_{2i}=1]=\mathbb{E}[\mathrm{v}_i|y_{2i}=1]+\mathbb{E}[[\B{z}_i^T,\B{x}_{-1_i}^T]\B{\delta}\hspace{0.2em}|\hspace{0.2em}y_{2i}=1]}
	\end{IEEEeqnarray}
	
	to get:
	\begin{IEEEeqnarray}{lCr}
		\label{Expected:Subs:Truncated:eps:X}
		\scaleobj{0.8}{\epsilon_{2i}=x_{1i}-\mathbb{E}[\mathrm{v}_i|y_{2i}=1]-\mathbb{E}[[\B{z}_i^T,\B{x}_{-1_i}^T]\B{\delta}\hspace{0.2em}|\hspace{0.2em}y_{2i}=1]}
	\end{IEEEeqnarray}

	We express $\mathbb{E}[\mathrm{v}_i|y_{2i}=1]$ in \eqref{Expected:Subs:Truncated:eps:X} as $\mathbb{E}[\mathrm{v}_i|y_{2i}=1]=\mathcal{M}_2(\B{w}_i^T\B{\gamma})$ where $\mathcal{M}_2(\cdot)$ is some unknown function  and obtain (see Theorem \ref{Theorem:IndirectDependence:V:Eps2} to follow):
	\begin{IEEEeqnarray}{lCr}
		\label{final:expect:endog}
		\scaleobj{0.8}{x_{1i} = \underbrace{[\B{z}_i^T,\B{x}_{-1_i}^T]\B{\delta}}_{\text{\shortstack{substantive\\ covariates}}}\hspace{0.2em}+\hspace{0.2em}\overbrace{\underbrace{\mathcal{M}_2(\B{w}_i^T\B{\gamma})}_{\text{\shortstack{selection bias\\ term}}} \hspace{0.2em}+\hspace{0.5em}\underbrace{\epsilon_{2i}}_{\text{white noise}}}^{\epsilon_{2i}^*}}
	\end{IEEEeqnarray}

	It is easy to see  the joint dependence of $\epsilon_{2i}^*$ and $\epsilon_{1i}^*$ through the selection bias terms in \eqref{final:expect:part:two} and \eqref{final:expect:endog}.% Additionally, under truncation the endogeneity bias term is actually $\mathbb{E}[\disturb_{1i}\hspace{0.2em}|\hspace{0.2em}\B{x}_i,\mathrm{s}=1]$.
	
	%It easy to show that $\mathbb{E}[\epsilon_{1i}|\B{x}_i]$ is not a function of $\B{x}_i$
	%\begin{IEEEeqnarray}{lCr}
	
	%\mathbb{E}[\epsilon_{1i}|\B{x}]=\mathbb{E}\curly{\varepsilon_{i}-\mathbb{E}[\varepsilon_{i}|y_{2i}=1]\hspace{0.2em}|\hspace{0.2em}\B{x}}=\mathbb{E}\braket{\varepsilon_{i}|\B{x}}-\mathbb{E}\curly{\mathbb{E}[\varepsilon_{i}|y_{2i}=1]\hspace{0.2em}|\hspace{0.2em}\B{x}}=\underbrace{\mathbb{E}[\varepsilon_{i}|\B{x}]}_{0}-\underbrace{\mathbb{E}[\varepsilon_{i}|y_{2i}=1]}_{\text{\shortstack{depends\\ on $y_{2i}$}}}
	%\end{IEEEeqnarray}

	%It easy to show that 
	%\begin{IEEEeqnarray}{lCr}
	
	%\mathbb{E}[\epsilon_{1i}|y_{2i}=1]=\mathbb{E}\curly{\varepsilon_{i}-\mathbb{E}[\varepsilon_{i}|y_{2i}=1]\hspace{0.2em}|\hspace{0.2em}y_{2i}=1}=\mathbb{E}\braket{\varepsilon_{i}|y_{2i}=1}-\underbrace{\mathbb{E}\curly{\mathbb{E}[\varepsilon_{i}|y_{2i}=1]\hspace{0.2em}|\hspace{0.2em}y_{2i}=1}}_{\mathbb{E}\braket{\varepsilon_{i}|y_{2i}=1}}=0
	%\end{IEEEeqnarray}

	%By law of iterated expectation it follows that $\mathbb{E}\braket{\mathbb{E}[\disturb_{1i}\hspace{0.2em}|\hspace{0.2em}\B{x}_i]\hspace{0.2em}\big|\hspace{0.2em}\mathrm{s}=1}=\mathbb{E}[\disturb_{1i}\hspace{0.2em}|\hspace{0.2em}\B{x}_i]$: 
\end{proof}

Next we formulate the relationship between the covariates and dependent variables in the equations to be estimated, in the presence of an endogenous covariate in the substantive equation under truncation.
\subsection{Truncated sample selection model with an endogenous covariate}
In cases where the substantive equation's dependent variable is a function of an endogenous covariate $x_{1i}$, both $x_{1i}$ as well as $y_{1i}$ (as in \eqref{Substantive:Equation}) are truncated, we face a truncated sample selection model with an endogenous covariate. 

Thus, the semiparametric partially linear index model in a truncated environment consists of the following system of equations:
\begin{IEEEeqnarray}{lCr}
	\label{Truncated:Substatnive:System}
	\scaleobj{0.8}{\braket{\begin{matrix}
				y_{1i}\\
				x_{1i}
		\end{matrix}}=
		\begin{cases}
			\B{x}_i^T\B{\beta}\hspace{3.2em} + \mathcal{M}_1(\B{w}_i^T\B{\gamma}) \hspace{1em}+\hspace{1em}\overbrace{\mathbb{E}[\disturb_{1i}\hspace{0.2em}|\hspace{0.em}\B{x}_i]\hspace{0.1em}+ \hspace{0.05em}\underbrace{\epsilon_{1i}}_{\text{white noise}}}^{\epsilon_{1i}^{**}}\\
			[\B{z}_i^T,\B{x}_{-1_i}^T]\B{\delta}\hspace{0.5em} + \mathcal{M}_2(\B{w}_i^T\B{\gamma}) \hspace{1em}+ \hspace{1em}\underbrace{\epsilon_{2i}}_{\text{white noise}}
	\end{cases}}
\end{IEEEeqnarray}

where  $\epsilon_{1i}^{**}$ and $\epsilon_{2i}$ are two jointly dependent random disturbances,\footnote{There is dependence of these two random disturbances due to the dependence between $\mathrm{v_i}$ and $\mathrm{\disturb_{1i}}$ (as in \eqref{Baseline:Model}) in the complete (non-truncated) data.} and by construction are independent of the random variables vector $\B{\mathrm{w}}$.\footnote{Not to be confused with its realization $\B{w_i}$.} The intrinsic endogeneity in the model is captured by the joint dependence of $\epsilon_{1i}^{**}$ and the covariates.\footnote{The intrinsic model's endogeneity is related to the joint dependence of the random disturbance and the covariates \textit{in the population}, unlike a conditional joint dependence  of the random disturbance and the covariates given participation in the sample.} %As we concentrate on the presence of endogenous covariates in the substantive equation, we assume that in the population the random disturbances $\disturb_{1i}$ and $\mathrm{v}_{1i}$ are independent of $\B{\mathrm{w}}$. 
The presence of the function $\mathcal{M}_2(.)$ implies that we allow for a dependence between $\mathrm{v_i}$ (the endogenous part of $x_i$) and the selection equation's random disturbance $\disturb_{2i}$ (in \eqref{Nontruncated:Selection:Equation}, the complete, non-truncated, sample selection equation).% In such cases, $x_i$ is endogenous to both random disturbances $\epsilon_{1i}$ and $\epsilon_{2i}$.

%The instrumental variable $z_i$ is constructed as follows: 
%\begin{IEEEeqnarray}{lCr}
\label{instrument}
%z_i=\mathcal{G}(\B{w_i}) + \varphi_i
%\end{IEEEeqnarray}

%with $\varphi_i$ an i.i.d white noise random disturbance. It is important to notice that $x_i$ is correlated with $z_i$ according to \eqref{Substantive:Equation}. By construction, $z_i$ and any other variable (except for $\varphi_i$) are conditionally independent given $\B{w_i}$.

%\color{red} % why this is important

Our primary interest is to show that the instrumental variable and the random disturbance might be correlated in a truncated environment as will be depicted in Theorem \ref{Theorem:Lack:of:Orthogonality} to follow.
By doing this, we denote a truncated environment using the indicator (selection variable) $\mathrm{s}=I(\disturb_{2i}>-\B{\mathrm{w}}^T\B{\mathrm{\gamma}})$ and postulate the following assumptions:
\begin{assumption}\label{Joint:Dependence:Z}
	The instrumental variable  $\mathrm{z}$ is jointly distributed with all covariates in the data: $\mathbb{E}[\mathrm{z}|\B{\mathrm{w}=w}]=\mathcal{G}(\B{\mathrm{w}})$ where $\mathcal{G}(\cdot)$ is some function of $\B{\mathrm{w}}$.
\end{assumption}

\begin{assumption}\label{Stochastic:Independence:Requirement}
	Conditioning the instrumental variable $\mathrm{z}$ both on random variable $\B{\mathrm{w}}$ and a stochastic function of $\B{\mathrm{w}}$ denoted by $\mathcal{F}(\B{\mathrm{w}},\mathrm{\varepsilon})$  (given that the stochastic component $\mathrm{\varepsilon}$ is an i.i.d white noise which is independent of $\mathrm{z}$), would be the same as conditioning it only on $\B{\mathrm{w}}$. Formally: $\mathbb{E}[\mathrm{z}|\B{\mathrm{w}=w},\mathcal{F}(\B{\mathrm{w}},\mathrm{\varepsilon})]=\mathbb{E}[\mathrm{z}|\B{\mathrm{w}=w}]$.% This assumption implies the joint dependence of $\mathrm{z}$ and $\mathrm{s}$ which is generated by a variation in $\B{\mathrm{w}}$.
	%\begin{IEEEeqnarray}{lCr}
	\label{Stochastic:Independence:Requirement}
	%\mathbb{E}[\mathrm{z}|\B{\mathrm{w}=w},\mathcal{F}(\B{\mathrm{w}},\mathrm{\varepsilon})]=\mathcal{G}(\B{\mathrm{w}})
	%\end{IEEEeqnarray}
	
	%where $\mathcal{G}(\cdot)$ is some function of $\mathrm{w}$.
\end{assumption} 

These two assumptions implies that the conditional expectation of the instrumental variable, given the selection variable, is a function of the random variable vector $\B{\mathrm{w}}$, as the following proposition argues: %Their joint dependency is generated by a variation in the random variable vector $\B{\mathrm{w}}$. are required in the following proposition
\begin{proposition}\label{Proposition:Dependence:Z:Selection}
	Given assumptions \ref{Joint:Dependence:Z} and \ref{Stochastic:Independence:Requirement}, the conditional expectation of the instrumental variable given the selection variable is a function of $\B{\mathrm{w}}$, rather $\mathbb{E}[\mathrm{z}|\mathrm{s}=s]=\int_{\B{w}}\mathcal{G}(\B{w})f_{\B{\mathrm{w}}|\mathrm{s}=s}(\B{w}|\mathrm{s}=s)d\B{w}$ $\forall s\in\curly{0,1}$, where $f_{\B{\mathrm{w}}|\mathrm{s}=1}(\B{w}|\mathrm{s}=1)$ and $f_{\B{\mathrm{w}}|\mathrm{s}=0}(\B{w}|\mathrm{s}=0)$ are the conditional density functions of vector $\B{\mathrm{w}}$ given participation and non-participation, respectively.
\end{proposition}

\begin{proof}
	It easy to see that $\B{\mathrm{w}}$ mediates between  $\mathrm{z}$ and $\mathrm{s}$ using the Tower property of conditional expectation \cite{williams1991probability}:\footnote{The Tower property is referred interchangeability to the law of iterated expectations. For formal proof see \cite{williams1991probability}.}
	%\begin{IEEEeqnarray}{lCr}
	
	%\TowerNew{z}{w}{s}
	%\end{IEEEeqnarray}
	
	\begin{IEEEeqnarray}{lCr}
		\CondEXY{z}{w}{s}=\NestedCondE{z}{w}{s}, \hspace{1em} s\in\curly{0,1}\nonumber
	\end{IEEEeqnarray}

	The indicator variable $\mathrm{s}$ is a stochastic function of $\B{\mathrm{w}}$, thus, it follows from assumption \ref{Stochastic:Independence:Requirement} that
	\begin{IEEEeqnarray}{lCr}
		\NestedCondE{z}{w}{s}=\mathbb{E}_{\B{\mathrm{w}}}[\mathbb{E}\braket{\mathrm{z}|\B{\mathrm{w}}}|\mathrm{s}=s], \hspace{1em} s\in\curly{0,1}\nonumber
	\end{IEEEeqnarray}
	
	Following assumption \ref{Joint:Dependence:Z} we get:
	\begin{IEEEeqnarray}{lCr}
		\scaleobj{0.85}{\CondEXY{z}{w}{s}=\mathbb{E}_{\B{\mathrm{w}}}[\mathbb{E}\braket{\mathrm{z}|\B{\mathrm{w}}}|\mathrm{s}=s]}\\ \scaleobj{0.8}{=\mathbb{E}_{\B{\mathrm{w}}}[\mathrm{\mathcal{G}}|\mathrm{s}=s]=\int_{\B{w}}\mathcal{G}(\B{w})f_{\B{\mathrm{w}}|\mathrm{s}=s}(\B{w}|\mathrm{s}=s)d\B{w}, \hspace{1em} s\in\curly{0,1}.}\nonumber
	\end{IEEEeqnarray}

	%The second equality follows from assumption \ref{Stochastic:Independence:Requirement} and relies on the fact that the indicator variable $\mathrm{s}$ is a stochastic function of $\B{\mathrm{w}}$. 
	%The second equality stems from assumption \ref{Stochastic:Independence:Requirement}, 
	
	%The indicator variable $\mathrm{s}$ is a stochastic function of $\B{\mathrm{w}}$, thus, it follows from assumption \ref{Stochastic:Independence:Requirement} that
	%$\mathbb{E}[\mathrm{z}|\mathrm{s},\B{\mathrm{w}}]=\mathbb{E}[\mathrm{z}|\B{\mathrm{w}}]$. Consequently, using assumption   \ref{Joint:Dependence:Z} $\mathbb{E}[\mathbb{E}\braket{\mathrm{z}|\B{\mathrm{w}}}|\mathrm{s}=s]=\mathbb{E}[\mathcal{G}(\B{\mathrm{w}})|\mathrm{s}=s]$ we get:
	%$\mathbb{E}[\mathrm{z}|\mathrm{s}=s]=\mathbb{E}[\mathcal{G}(\B{\mathrm{w}})|\mathrm{s}=s]=\int_{\B{w}}\mathcal{G}(\B{w})f_{\B{\mathrm{w}}|\mathrm{s}=s}(\B{w}|\mathrm{s}=s)d\B{w}$ $\forall s\in\curly{0,1}$.
\end{proof}
%This result is obtained 

%Following assumption \ref{Joint:Dependence:Z} we get
%$\mathbb{E}[\mathrm{z}|\B{\mathrm{w}=w}]=\int$,
%where $\mathcal{G}(\cdot)$ is some function of $\mathrm{w}$. Following assumption \ref{Stochastic:Independence:Requirement}:

%We denote a truncated environment using the indicator (selection variable) $\mathrm{s}=I(\disturb_{2i}>-\B{\mathrm{w}}^T\B{\mathrm{\gamma}})$ and show that .

%By assumptions \ref{Joint:Dependence:Z} and \ref{Stochastic:Independence:Requirement}, using the Tower property we get:
%$\TowerNew{z}{w}{s}$.\footnote{For formal proof see \cite{williams1991probability}.}

%by showing that the conditional expectation of the random variables product $\mathrm{z\cdot\disturb_{1}}$ given participation (which is a function of a covariates vector $\B{\mathrm{w}}$) is affected by the comovement of each one of these random variables with respect to $\B{\mathrm{w}}$. 

%We initially simplify the expectation of these random variables product,
%utilizing the Tower property \cite{williams1991probability} of conditional expectation that is: 
%let $\mathrm{z}$ and $\B{\mathrm{w}}$ be continuous random variables, and let $\mathrm{s}$ be a discrete variable, indicating  participation. These three random variables $\curly{\mathrm{z},\B{\mathrm{w}},\mathrm{s}}$ must satisfy $\Tower{z}{w}{s}$.\footnote{For formal proof see \cite{williams1991probability}.}% (see Appendix for a formal proof).

%We assume that the instrumental variable is jointly distributed with all covariates in the data.  
In Theorem \ref{Theorem:Lack:of:Orthogonality} to follow we use proposition \ref{Proposition:Dependence:Z:Selection} and present our primary argument: in truncated sample selection models, the orthogonality condition of the instrumental variable with respect to the random disturbance might be violated. This violation stems from a dependency between the instrumental variables and the selection equation's covariates.

\begin{theorem}[Lack of orthogonality]\label{Theorem:Lack:of:Orthogonality}	
	Let  $\mathrm{\disturb_1}$ and  $\mathrm{\disturb_2}$ be two jointly distributed random disturbances, and let $\mathrm{z}$ be a valid instrumental variable satisfying $\mathbb{E}[\mathrm{z\cdot\disturb_1}]=0$. Denote a random variables vector $\B{\mathrm{w}}\in\mathbb{R}^l$, a parameters vector $\B{\gamma}\in\mathbb{R}^l$ and a truncated environment using the indicator variable $\mathrm{s}=I(\disturb_{2}>-\B{\mathrm{w'\gamma})}$.
	%(ii) 	 $\mathbb{E}[\mathrm{z}|\B{\mathrm{w}=w},\mathrm{s}=s]=\mathbb{E}[\mathrm{z}|\B{\mathrm{w}=w}]=\mathcal{G}(\B{w})$ ;	
	Suppose that the following conditions are satisfied:\\ (i) assumptions \ref{Joint:Dependence:Z} and \ref{Stochastic:Independence:Requirement} hold; \\(ii) $\mathbb{E}[\mathrm{\disturb_1}|\mathrm{s}=s,\B{\mathrm{w}=w}]=\mathcal{M}(\B{w}^T\B{\gamma})$;
	(iii) $\mathrm{z}$ and $\mathrm{\disturb_1}$ are conditionally independent given $\B{\mathrm{w}}$ and $\mathrm{s}$; (iv) $\mathcal{G}$ and $\mathcal{M}$ are linearly dependent in the truncated environment (given $\mathrm{s}$).\footnote{The conditional linear dependence between $\mathcal{G}$ and $\mathcal{M}$ given the indicator (selection) variable implies that $\mathbb{E}\braket{\mathcal{G}\mathcal{M}|\mathrm{s}=s}\ne \mathbb{E}\braket{\mathcal{G}|\mathrm{s}=s}\mathbb{E}\braket{\mathcal{M}|\mathrm{s}=s}$. Since  $\mathcal{G}$ and $\mathcal{M}$ are both functions of the random variable $\B{\mathrm{w}}$, this inequality implies that %$\mathbb{E}\braket{\mathcal{G}\mathcal{M}|\mathrm{s}=s}=\int\mathcal{G}(\B{w})\mathcal{M}(\B{w}^T\B{\gamma})f_{\B{\mathrm{w}}|\mathrm{s}=s}(\B{w})d\B{w}$ and $\mathbb{E}\braket{\mathcal{G}|\mathrm{s}=s}=\int\mathcal{G}(\B{w})f_{\B{\mathrm{w}}|\mathrm{s}=s}(\B{w})d\B{w}$ , since 
		$\int \mathcal{G}(\B{w})\mathcal{M}(\B{w}^T\B{\gamma})f_{\B{\mathrm{w}}|\mathrm{s}=s}(\B{w})d\B{w}\ne \int \mathcal{G}(\B{w})f_{\B{\mathrm{w}}|\mathrm{s}=s}(\B{w})d\B{w} \int \mathcal{M}(\B{w}^T\B{\gamma})f_{\B{\mathrm{w}}|\mathrm{s}=s}(\B{w})d\B{w}\nonumber$.} 
	
	Under conditions (i)-(iv) above, $\mathrm{z}$ is not orthogonal to the random disturbance $\mathrm{\disturb_1}$ given $\mathrm{s}$.
\end{theorem}

%\pagebreak
\begin{proof}
	Using the Tower property, the following must hold: % depicted in proposition \eqref{Theorem:Tower}
	\begin{IEEEeqnarray}{lCr}
		\scaleobj{0.8}{\CondEXY{z\disturb_1}{w}{s}=\NestedCondE{z\disturb_1}{w}{s}}
		\scaleobj{0.8}{=\underbrace{\NestedCondIndepE{z}{\disturb_1}{w}{s}}_{\substack{\text{by conditional independence of }\\ \mathrm{z} \text{ and } \mathrm{\disturb_1} \text{ given } \B{\mathrm{w}} \text{ and } \mathrm{s}}}}\nonumber\\
		\scaleobj{0.8}{=\mathbb{E}_{\B{\mathrm{w}}}[\mathrm{\mathcal{G}\mathcal{M}}|\mathrm{s}=s]=\int_{\B{w}}\mathcal{G}(\B{w})\mathcal{M}(\B{w}^T\B{\gamma})f_{\B{\mathrm{w}}|\mathrm{s}=s}(\B{w}|\mathrm{s}=s)d\B{w}}\nonumber
	\end{IEEEeqnarray}

	Similarly (using proposition \ref{Proposition:Dependence:Z:Selection}),
	\begin{IEEEeqnarray}{lCr}
		\scaleobj{0.8}{\CondEXY{z}{w}{s}=\int_{\B{w}}\mathcal{G}(\B{w})f_{\B{\mathrm{w}}|\mathrm{s}=s}(\B{w}|\mathrm{s}=s)d\B{w}}\nonumber
	\end{IEEEeqnarray}
	
	and
	\begin{IEEEeqnarray}{lCr}
		\scaleobj{0.8}{\CondEXY{\disturb_1}{w}{s}=\NestedCondE{\disturb_1}{w}{s}}\\
		\scaleobj{0.8}{=\mathbb{E}_{\B{\mathrm{w}}}[\mathrm{\mathcal{M}}|\mathrm{s}=s]=\int_{\B{w}}\mathcal{M}(\B{w}^T\B{\gamma})f_{\B{\mathrm{w}}|\mathrm{s}=s}(\B{w}|\mathrm{s}=s)d\B{w}}\nonumber
	\end{IEEEeqnarray}

	As $\mathcal{G}$ and $\mathcal{M}$ are conditionally linearly dependent random variables in the truncated environment (given $\mathrm{s}$), implies:
	\begin{IEEEeqnarray}{lCr}
		\scaleobj{0.8}{\int \mathcal{G}(\B{w})\mathcal{M}(\B{w}^T\B{\gamma})f_{\B{\mathrm{w}}|\mathrm{s}=s}(\B{w})d\B{w}}
		\\ \scaleobj{0.8}{\ne \int \mathcal{G}(\B{w})f_{\B{\mathrm{w}}|\mathrm{s}=s}(\B{w})d\B{w} \int \mathcal{M}(\B{w}^T\B{\gamma})f_{\B{\mathrm{w}}|\mathrm{s}=s}(\B{w})d\B{w}}\nonumber
	\end{IEEEeqnarray}
	
	and consequently:
	\begin{IEEEeqnarray}{lCr}
		\scaleobj{0.8}{\CondEXY{z\disturb_1}{w}{s}\ne\CondEXY{z}{w}{s}\CondEXY{\disturb_1}{w}{s} => \mathbb{COV}[\mathrm{z},\mathrm{\disturb_1}|\mathrm{s}=s]\ne 0}\nonumber
	\end{IEEEeqnarray}
	
	Therefore, $\mathrm{z}$ is not orthogonal to $\mathrm{\disturb_1}$ given $\mathrm{s}$ (in the truncated environment).
	%	Let define the participants' substantive regression equation as \cite{robinson1988root}: 
	%	\begin{IEEEeqnarray}{lCr}
	\label{Robinson:equation}
	%	 y_{1i}^*=\B{x_i\beta} + \underbrace{\mathcal{M}(\B{w}_i^T\B{\gamma}) + \epsilon_i}_{\disturb_1^*}
	%	\end{IEEEeqnarray}
	
	%	where $\B{x}_i$ and $\B{w_i}$ are each a realization of $\B{\mathrm{x}}\in\mathbb{R}^p$ and $\B{\mathrm{w}}\in\mathbb{R}^l$, respectively. $\mathcal{M}$ is some unknown function for capturing the bias term (which is the conditional expectation of substantive equation's dependent variable) $\mathbb{E}[y_{1i}|y_{2i}\ge 0]$. While $\epsilon_i$ is an i.i.d white noise randomly distributed disturbance.
	
	%	Using \eqref{Robinson:equation}, we obtain:	
	%	\begin{IEEEeqnarray}{lCr}
	
	%	\mathbb{E}[(\mathcal{G}(\B{w_i}) + \varphi_i)\close{\mathcal{M}(\B{w}_i^T\B{\gamma}) + \epsilon_i}]=\underbrace{\mathbb{E}[\mathcal{G}(\B{w_i})^{*\prime}\mathcal{M}(\B{w}_i^T\B{\gamma})]}_{\ne 0}+\underbrace{\mathbb{E}[z^{*\prime}\epsilon_i}_0]
	%	\end{IEEEeqnarray}
	
	%	where $z^*\equiv \braket{z|\disturb_2 > -\B{\mathrm{w}'\gamma}}$. 
\end{proof}

%The orthogonality condition is violated since the random disturbance $\epsilon_{1i}^{**}$ in  \eqref{Truncated:Substatnive:System}  is contaminated with $\mathcal{M}_1(\cdot)$, a function of the selection equation's covariates which are correlated with the instrumental variable. %Due to the joint dependence of the instrumental variable and the covariates are referred to as the contamination factor. % a set of covariates which are jointly dependent on the instrumental variable the covariate $\mathcal{M}_1(\cdot)$, referred to as the contamination factor which might be correlated with the instrumental variable. 
However, the orthogonality condition can be satisfied by removing the contamination factor, which is the covariate generating the comovement between the random disturbance  and the instrumental variable, as shown in the following Theorem \ref{Theorem:Contamination}.

\begin{theorem}[Bias removal]\label{Theorem:Contamination}
	Let  $\mathrm{\disturb_1}$ and  $\mathrm{\disturb_2}$ be two jointly distributed random disturbances, and let $\mathrm{z}$ be a valid instrumental variable satisfying $\mathbb{E}[\mathrm{z\cdot\disturb_1}]=0$. Denote a random variables vector $\B{\mathrm{w}}\in\mathbb{R}^l$, a parameters vector $\B{\gamma}\in\mathbb{R}^l$ and a truncated environment using the indicator variable $\mathrm{s}=I(\disturb_{2}>-\B{\mathrm{w'\gamma})}$.
	%(ii) 	 $\mathbb{E}[\mathrm{z}|\B{\mathrm{w}=w},\mathrm{s}=s]=\mathbb{E}[\mathrm{z}|\B{\mathrm{w}=w}]=\mathcal{G}(\B{w})$ ;	
	Suppose that the following conditions are satisfied: (i) assumptions \ref{Joint:Dependence:Z} and \ref{Stochastic:Independence:Requirement} hold; \\(ii) $\mathbb{E}[\mathrm{\disturb_1}|\mathrm{s}=s,\B{\mathrm{w}=w}]=\mathcal{M}(\B{w}^T\B{\gamma})$.
	
	Under conditions (i) and (ii) above,	
	removing the contamination factor (the bias term) from the residual in the truncated environment
	leads to orthogonality of the instrumental variable to the substantive equation's disturbance, such that: $\CondEXY{z\braket{\disturb_{1}-\mathcal{M}(\B{\mathrm{w}}^T\B{\mathrm{\gamma}})}}{w}{s} = 0$.
\end{theorem}

%\color{blue}
\begin{proof}
	Express $\CondEXY{z\braket{\disturb_{1}-\mathcal{M}(\B{\mathrm{w}}^T\B{\mathrm{\gamma}})}}{w}{s}$ as a difference of two conditional expectations:% $\CondEXY{\disturb_1-\mathcal{M}(\B{\mathrm{w}}^T\B{\mathrm{\gamma}})}{w}{s}=0$,} we obtain:
	\begin{IEEEeqnarray}{lCr}
		\scaleobj{0.8}{\CondEXY{z\braket{\disturb_1-\mathcal{M}(\B{\mathrm{w}}^T\B{\mathrm{\gamma}})}}{w}{s}	=
			\CondEXY{z\disturb_{1}}{w}{s}-\CondEXY{z\mathcal{M}(\B{\mathrm{w}}^T\B{\mathrm{\gamma}})}{w}{s}}\nonumber
	\end{IEEEeqnarray}
	
	Using the Tower property, to get:  % depicted in \eqref{Theorem:Tower}
	\begin{IEEEeqnarray}{lCr}
		\scaleobj{0.8}{\CondEXY{z\mathcal{M}(\B{\mathrm{w}}^T\B{\mathrm{\gamma}})}{w}{s}=
			\NestedCondE{z\mathcal{M}(\B{\mathrm{w}}^T\B{\mathrm{\gamma}})}{w}{s}}\nonumber\\
		=\scaleobj{0.8}{\underbrace{\NestedCondIndepENosubscript{z}{\mathcal{M}(\B{\mathrm{w}}^T\B{\mathrm{\gamma}})}{w}{s}}_{\substack{\text{by conditional independence of }\\ \mathrm{z} \text{ and } \mathcal{M}(\B{\mathrm{w}}^T\B{\mathrm{\gamma}}) \text{ given } \B{\mathrm{w}} \text{ and } \mathrm{s}}} = \CondEXY{\mathcal{G}(\B{\mathrm{w}})\mathcal{M}(\B{\mathrm{w}}^T\B{\mathrm{\gamma}})}{w}{s}} 	\nonumber
	\end{IEEEeqnarray}
	
	As $\scaleobj{0.8}{\CondEXY{z\disturb_{1}}{w}{s}=\CondEXY{\mathcal{G}(\B{\mathrm{w}})\mathcal{M}(\B{\mathrm{w}}^T\B{\mathrm{\gamma}})}{w}{s}}$  
	(proof of Theorem \ref{Theorem:Lack:of:Orthogonality}), which implies that $\scaleobj{0.8}{\CondEXY{z\braket{\disturb_{1}-\mathcal{M}(\B{\mathrm{w}}^T\B{\mathrm{\gamma}})}}{w}{s} = 0.}$
	
	%\pagebreak
	
	Moreover,
	\begin{IEEEeqnarray}{lCr}
		\scaleobj{0.8}{\mathbb{COV}\braket{\mathrm{z},\disturb_{1}-\mathcal{M}(\B{\mathrm{w}}^T\B{\mathrm{\gamma}})|\mathrm{s}=s}	=\underbrace{\CondEXY{\mathrm{z}\braket{\disturb_{1}-\mathcal{M}(\B{\mathrm{w}}^T\B{\mathrm{\gamma}})}}{w}{s}}_{0}}\\ \scaleobj{0.8}{-\CondEXY{z}{w}{s} \underbrace{\CondEXY{\disturb_{1}-\mathcal{M}(\B{\mathrm{w}}^T\B{\mathrm{\gamma}})}{w}{s}}_{0}=0.}  \nonumber
	\end{IEEEeqnarray}
	
\end{proof}
%\color{black}

Therefore, a valid instrumental variable $\mathrm{z}$ is orthogonal to the truncated distribution  (non-contaminated) disturbance $\epsilon_{1i}^{**}$ in \eqref{Truncated:Substatnive:System}, even though $\mathrm{z}$ and $\B{\mathrm{w}}$ are dependent.

The joint dependence of  $(\disturb_1,\disturb_2,\mathrm{v})$ implies the violation of zero mean expectation (under truncation) in the $x_{1i}$ regression equation \eqref{Substantive:Equation}, such that  $\mathbb{E}[\mathrm{v}|\disturb_{2}>-\B{w'\gamma}]=\mathcal{M}_2(\B{w}^T\B{\gamma})
\ne \mathbb{E}[\mathrm{v}]=0$. That is, the conditional expectation of $\mathrm{v}$ given an endogenous truncation is a function of the covariate vector $\B{w}$, while in the population it does not depend on $\B{w}$ and has a zero mean expectation.  This violation is a precondition for the endogeneity of  $\curly{\B{x_{-1i},z_i}}$ with respect to $\mathrm{v_i}$ given participation in the regression of $x_{1i}$.\footnote{As been discussed in \cite{heckman1979sample}, the fact that the conditional disturbance (given participation) in the substantive equation of  $x_{1i}$ is a function of the selection equation's covariates, leads to a potential correlation between the disturbance and the substantive equation's covariates. This correlation implies the endogeneity of the substantive equation's covariates $\curly{\B{x_{-1i},z_i}}$ with respect to its random disturbance $\mathrm{v_i}$ given participation.} The following theorem indicates that such violation is also obtained in cases where %$\mathrm{v}$ and $\disturb_2$ are conditionally independent given $\disturb_1$ implying that 
the comovement of $\mathrm{v}$ and $\disturb_2$ is entirely related to a variation in $\disturb_1$. 

%bias-term consisting of the selection equation's covariates, the .
%\footnote{As can be shown in \eqref{Truncated:Substatnive:System}, the endogenous covariate $x_{1i}$ in the truncated environment (given participation) is regressed on the instrumental variable, the substantive equation's exogenous covariates and a non-linear function of the selection's equation covariates.} %This non-linear function captures

% the non-constant part in is necessary due to the fact that the disturbance $\mathrm{v}|\mathrm{s}=s$ is not constant.

% $\mathbb{E}[\mathrm{v}|\mathrm{s}=s]\ne 0$, which is a violation of $OLS$ assumption in the truncated environment. 
%for generating  a zero-mean random disturbance $\epsilon_{1i}$ since generally 
%Next we provide a Theorem that unconditional independence between $\mathrm{v}$ and $\disturb_2$, the conditional expec

\begin{theorem}[Conditional independence]\label{Theorem:IndirectDependence:V:Eps2}	
	Let  $\mathrm{\disturb_1}$ and $\mathrm{\disturb_2}$ be two jointly distributed random disturbances of the substantive and selection equations, respectively. Let $\mathrm{v}$ be a random variable which depends on $\disturb_1$ such that $\mathrm{v}$ and $\disturb_2$ are conditionally independent given $\disturb_1$.
	%the endogenous part of $\mathrm{x_1}$, the endogenous covariate in the substantive equation, such that $\mathrm{v}$ and $\disturb_1$ are dependent (endogeneity), while $v_i$ and $\disturb_2$ are independent.  % jointly distributed random disturbances, 
	%such that ($\disturb_1$, $\mathrm{v}$ and ($\disturb_1$ and $\disturb_2$) are bivariate dependent random variable. $\mathrm{v}$ is the endogenous part of the endogenous covariate. 
	Denote a random variables vector $\B{\mathrm{w}}\in\mathbb{R}^l$ independent of $(\disturb_1,\disturb_2,\mathrm{v})$ with a realization $\B{w}$, a parameters vector $\B{\gamma}\in\mathbb{R}^l$ and a truncated environment using the indicator variable $\mathrm{s}=I(\disturb_{2}>-\B{w'\gamma})$.% The realization of $\B{\mathrm{w}}$ and $\mathrm{s}$ are denoted by $\B{w}$ and $s$, respectively.
	
	Assume the following conditions are satisfied: (i) the conditional expectation of the random disturbance given participation is $\mathbb{E}[\mathrm{\disturb_1}|\disturb_{2}>-\B{w'\gamma}]=\mathcal{M}_1(\B{w}^T\B{\gamma})$ \cite{robinson1988root}; (ii) 	 $\mathbb{E}[\mathrm{v}|\disturb_1,\disturb_{2}>-\B{w'\gamma}]=\mathbb{E}[\mathrm{v}|\disturb_1]=\mathcal{H}(\disturb_1)$, (endogeneity); (iii) $\mathcal{H}(.)$, a monotonic mapping $\mathbb{R}\mapsto \mathbb{R}$. 
	
	Under conditions (i)-(iii) above, $\mathbb{E}[\mathrm{v}|\disturb_{2}>-\B{w'\gamma}]\ne\mathbb{E}[\mathrm{v}]$ regardless of the conditional independence of $\mathrm{v}$ and $\disturb_2$ given $\disturb_1$. 
	%$\mathrm{v}$ depends on $\mathrm{s}$ regardless of the unconditional independence between $\mathrm{v}$ and $\disturb_2$. 
	%(which is  a function of $\disturb_2$) are jointly dependent random variables given $s$ (participation), regardless of their unconditional independence. 
\end{theorem}

\begin{proof}
	Applying Tower property  to $\mathbb{E}[\mathrm{v}|\mathrm{s}=1]$:
	\begin{IEEEeqnarray}{lCr}
		\scaleobj{0.8}{\mathbb{E}[\mathrm{v}|\disturb_{2}>-\B{w'\gamma}]=\mathbb{E}[\mathrm{v}|\mathrm{s}=1]=\mathbb{E}_{\disturb_1}\curly{\mathbb{E}[\mathrm{v}|\disturb_1,\mathrm{s}]|\mathrm{s}=1}} \nonumber \\ \scaleobj{0.8}{=\mathbb{E}[\mathcal{H}(\disturb_1)|\mathrm{s}=1]=\mathcal{M}_2(\B{w}^T\B{\gamma})\ne\mathbb{E}[\mathrm{v}].}\nonumber
	\end{IEEEeqnarray}

	It can be shown that $\disturb_1$  mediates between $\mathrm{v}$ and $\mathrm{s}$ (participation), in that it generates a comovement between the random variables $\mathrm{v}$ and $\mathrm{s}$.
	The last equality relies on the fact that the random variable $\mathcal{H}(\disturb_1)$ is a monotonic mapping of $\disturb_1$, implying dependence on $\mathrm{s}$ due to the dependency between $\disturb_1$ and $\mathrm{s}$.	
\end{proof}

%We note that in the presence of truncated data,  one cannot resort to \cite{hausman1978specification}'s exogeneity test as the instrumental variable is no longer orthogonal to the substantive equation's disturbance, being the basic assumption of Hausman's procedure.

Next we show that the conventional IV estimator is inconsistent in the presence of a truncated environment in which the expectation of the instrumental variable and the random disturbance are functions of the selection equation's covariates vector $\B{\mathrm{w}}$. The proof in section \ref{Proof:asymptotic} to follow, relies on a linear dependence assumption between these two functions of $\B{\mathrm{w}}$. The rationale for the linear dependence is due to the fact that the random disturbance's ($\disturb_1$) conditional expectation generally satisfies monotonicity with respect to the index variable $\B{\mathrm{w}'\gamma}$. Therefore, it is enough to assume that, on average, $\mathrm{z}$ is affected monotonically by the index variable $\B{\mathrm{w}'\gamma}$ to generate a linear dependence between  $\mathrm{z}$ and the conditional expectation of $\disturb_1$ given participation.\footnote{Both functions are dependent through $\B{\mathrm{w}}$ by construction, generally leading to some degree of linear dependence.}
\subsection{The conventional IV estimator's asymptotic bias}\label{Proof:asymptotic}
The IV estimator's asymptotic bias is:
\begin{IEEEeqnarray}{lCr}
	\scaleobj{0.8}{\B{\mathrm{\widehat{\beta}_{iv}}}=(\B{\mathrm{z}}^T\B{\mathrm{x}})^{-1}\B{\mathrm{z}}^T\B{\mathrm{y}}_1=
		(\B{\mathrm{z}}^T\B{\mathrm{x}})^{-1}\B{\mathrm{z}}^T(\B{\mathrm{x\beta}}+\mathcal{M}_1(\B{\mathrm{w}}^T\B{\mathrm{\gamma}})+\epsilon_{1i}^{**})}
\end{IEEEeqnarray}

\vspace{-2em}
\begin{IEEEeqnarray}{lCr}
	\scaleobj{0.8}{\B{\mathrm{\widehat{\beta}_{iv}}}=
		(\B{\mathrm{z}}^T\B{\mathrm{x}})^{-1}\B{\mathrm{z}}^T(\B{\mathrm{x\beta}})+(\B{\mathrm{z}}^T\B{\mathrm{x}})^{-1}\B{\mathrm{z}}^T\mathcal{M}_1(\B{\mathrm{w}}^T\B{\mathrm{\gamma}})+(\B{\mathrm{z}}^T\B{\mathrm{x}})^{-1}\B{\mathrm{z}}^T\epsilon_{1i}^{**}}\nonumber\\
	\scaleobj{0.8}{\B{\mathrm{\widehat{\beta}_{iv}}}=\B{\mathrm{\beta}}+(\B{\mathrm{z}}^T\B{\mathrm{x}})^{-1}\B{\mathrm{z}}^T\mathcal{M}_1(\B{\mathrm{w}}^T\B{\mathrm{\gamma}})+(\B{\mathrm{z}}^T\B{\mathrm{x}})^{-1}\B{\mathrm{z}}^T\epsilon_{1i}^{**}}\nonumber\\
	\scaleobj{0.8}{\underset{N\to \infty}{\mathrm{plim}}\braket{\B{\mathrm{\widehat{\beta}_{iv}}}}=
		\B{\mathrm{\beta}}+\underbrace{\underset{N\to \infty}{\mathrm{plim}}\braket{(N^{-1}\B{\mathrm{z}}^T\B{\mathrm{x}})^{-1}}\underset{N\to \infty}{\mathrm{plim}}\braket{N^{-1}\B{\mathrm{z}}^T\mathcal{M}_1(\B{\mathrm{w}}^T\B{\mathrm{\gamma}})}}_\text{Asymptotic bias}}\nonumber
\end{IEEEeqnarray}

Given any correlation between $\B{\mathrm{z}}$ and $\mathcal{M}_1(\B{\mathrm{w}}^T\B{\mathrm{\gamma}})$,  $\underset{N\to \infty}{\mathrm{plim}}\braket{\B{\mathrm{z}}^T\mathcal{M}_1(\B{\mathrm{w}}^T\B{\mathrm{\gamma}})}\not\to 0$. Thus, the $\B{\mathrm{\widehat{\beta}_{iv}}}$ estimator is an inconsistent estimator for $\B{\mathrm{\beta}}$. 
%\B{\mathrm{\beta}}+
\newcommand\EE[1]{\mathbb{E}\braket{{#1}}}

%\begin{IEEEeqnarray}{lCr}

%\EE{\mathrm{z_i\disturb_{1i}}|\mathrm{\disturb_{2i}}\ge-\B{\mathrm{w_i'\gamma}}}
%\end{IEEEeqnarray}

%Denote a random variables vector $\B{\mathrm{w}}$ of size $p\times 1$ and a convolution of its elements by $\B{\mathrm{w}'\gamma}$, where $\B{\mathrm{\gamma}}$ is a coefficients vector of size $p\times 1$. Denote an instrument variable $\mathrm{z_i}$ satisfying $\mathrm{z_i}\perp \mathrm{\disturb_{1i}}$.
%Assuming $\mathrm{z_i} \nBigCI  \B{\mathrm{w}}^T\B{\mathrm{\gamma}}$ such that $\mathrm{z_i} \nBigCI {\lambda(\B{\mathrm{-w'\gamma}})}$, implying that $\mathrm{z_i} \nBigCI  \EE{\mathrm{\disturb_{1i}}|\mathrm{\disturb_{2i}}\ge-\B{\mathrm{w}}^T\B{\mathrm{\gamma}}}$. Thus, $\mathrm{z_i}$ and  $\mathrm{\disturb_{1i}}|\mathrm{\disturb_{2i}}\ge-\B{\mathrm{w}}^T\B{\mathrm{\gamma}}$ are two dependent random variables.
%One obtains a conditional independence between the instrument and the disturbance $\mathrm{\disturb_1}$ in the truncated sample given $\B{\mathrm{w}=w_i}$ if the following condition is satisfied:
%\begin{IEEEeqnarray}{lCr}
\label{Coditional:Independence}
%\close{\mathrm{z_i}|\B{\mathrm{w}=w_i}}  \BigCI   %\close{\braket{\mathrm{\mathrm{\disturb_{1i}}|\mathrm{\disturb_{2i}}\ge-\B{\mathrm{w}}^T\B{\mathrm{\gamma}}}}|\B{\mathrm{w}=w_i}}
%\end{IEEEeqnarray}

%where the key assumption is that the source of dependence between $\mathrm{z_i}$ and the conditional disturbance $\mathrm{\disturb_{1i}}|\mathrm{\disturb_{2i}}\ge-\B{\mathrm{w}}^T\B{\mathrm{\gamma}}$ is the random variable $\B{\mathrm{w}}$, because they are both stochastic functions of $\B{\mathrm{w}}$. Once, conditioning both sides of \eqref{Coditional:Independence} on a specific value of $\B{\mathrm{w}}$, there is no dependence between these two variables.

%\color{blue}
%Explain what you are going to show next
%\color{black}
Next we discuss the two types of joint dependence which are present in our model. This is done is order to facilitate the understanding of our proposed procedure, which is intended to correct for the bias propagated by each type of joint dependence.% The first one is related to the population (non-truncated) distribution function, while the second one is related to the truncated sample distribution function. %two types of endogeneity which are present in the model and explain how they are related to the joint dependence of the covariates and the random disturbances in the population and in the truncated sample. 
%This understanding is essential in order to the procedure correcting for both of sources endogeneity% distribution and the sample (truncated) distribution functions. The second one is related to the truncated distribution function.%: (i)  unconditional joint dependence between the covariates and the substantive equation's random disturbance (given the selection variable) and (ii) conditional joint dependence and truncation are related to unconditional and conditional endogeneity (given the selection variable) in the model.% which can be alleviated by our procedure.% the joint dependence of the covariates and the substantive equations' random disturbances in the truncated data.% our proposed instrumental variable estimator to correct for the selectivity bias propagated by endogenously truncated sample. %Their procedure relies on the assumption of exogeneity of the regressors in the substantive equation with respect to the disturbance. We utilize a similar approach by estimating a partial linear regression. The only difference is due to the necessity to correct for two sources of bias: endogeneity of $x_i$ and truncation bias. Our suggested correction procedure requires the exogeneity of the instrumental variables to correct for these two types of bias sources.

\section{Preliminaries}\label{Section:Preliminatries}

The objective is to eliminate the selection bias term captured by $\mathcal{M}_1(\cdot)$ in \eqref{Truncated:Substatnive:System}. As we don't want to impose a specific distribution function on the random disturbances, the aforementioned elimination should be performed in a nonparametric manner. This can be achieved using a semiparametric estimation method, which is distribution-free. However, the bias term might be a discontinuous function with different levels of smoothness that must be considered. These issues can be alleviated using multi-resolution analysis by employing the wavelet estimator \cite{haar1910theorie}. Wavelet is a bandwidth-free estimator, that is based on the idea of multi-scale representation of the data \cite{delouille2006second}\footnote{Due to its multi-scale property, we can distinguish between the important information, the function's average behavior, from the noise. The coarse scales (lower resolution-levels) usually convey important information, while at fine scales there is usually more noise.} and is used as a denoising technique by
simple thresholding, which is based
on the concept of sparsity.\footnote{Sparsity implies that the majority of wavelet
	coefficients are small, and can be replaced by zero \cite{vanraes2002stabilised}.}

The applicability of the classical wavelet estimator is problematic in several important aspects. First, it limits the sample size to be represented as $2^J$, with $J$ a non-negative integer, and the observations to be equispaced, which challenges the estimation in case of irregular-spaced data.\footnote{The observations location in space or time must be of equal distance.} Second, the classical wavelet estimator imposes the parametric assumption that the disturbances are independent identically distributed normal variables \cite{silverman1999wavelets}. 
Lastly, there are the problems of coefficient expansion and border discontinuities.\footnote{The standard orthogonal wavelet transform has the 
	shortcoming in that it requires a large number of coefficients (coefficient expansion) to represent the original data \cite{usevitch2001tutorial}.} 

In order to overcome these limitations, second generation wavelets have been introduced \cite{daubechies1998factoring} which define wavelets in terms of lifting-steps instead of matrices to reduce computational complexity.\footnote{The lifting-steps are consecutive operations of prediction (scaled-moving average) and update (scaled-first difference) to obtain the wavelet coefficients.} An earlier attempt to deal with irregular-spaced data using second generation wavelets is presented in 
\cite{delouille2006second} by postulating a prior distribution function for the wavelet coefficients.\footnote{\cite{delouille2006second} adopt the parametric Bayesian denoising approach introduced by \cite{johnstone2005empirical, johnstone2004needles} to obtain the wavelet coefficients assuming the coefficients are distributed according to a continuous mixture of a normal by a Beta density.} Alternative approaches extend Haar wavelet transform to accommodate for irregular data \cite{delouille2004smooth}. 

%The semiparametric estimation method we employ is the biorthogonal wavelet estimator due to its many advantageous.

%Our main contribution to wavelet CDF 9/7 is by formulating analytically instead of approximating the inverse of the transpose of wavelet
Both first as well as second generation wavelet estimation methods involve three steps: coefficient estimation (forward transform); (ii) denoising by using element-wise thresholding (coefficients selection) and (iii) reconstruction of the data without the noise (inverse transform). It is important to notice that the sequential nature of the estimation that relies on element-wise thresholding is applicable for limited types of wavelets, referred to as orthogonal wavelets which consist of the above described limitations. The main shortcoming of orthogonal wavelets is that the compact support and the symmetry properties which are useful in denoising are conflicting.\footnote{Unlike biorthogonality, orthogonality and symmetry are conflicting properties for design
	of compactly supported nontrivial wavelets (see Theorem 8.1.4 in \cite{daubechies1992ten}).} % \cite{wei1998new} %One of the most useful non-orthogonal wavelets is introduced to overcome these limitations and is referred to as Cohen JPEG CDF 9/7, %the wavelets regressors to be orthonogal 
To preserve both these properties,  the biorthogonal wavelet-based JPEG is used %which posses both symmetry and compact support is JPEG 
\cite{chern1999interpolating, wei1998new}.\footnote{The JPEG algorithm used here is termed `wavelet CDF 9/7'.} %The primary use of which is compressing and denoising pictures, and as such is suitable for our purposes.
% transform The main drawback of the lifting step, especially in irregular-spaced data is the absence of an analytic expression of the 
%{_{j,k}}

%Next we define biorthogonality. 
%\subsubsection{Biorthogonality definition}
In what follows we briefly explain the concept of biorthogonality. Denote a set of functions $\curly{\varphi_{_{k}}(t)}$ which spans a vector space $\mathcal{F}$, referred to as the expansion set. % and similarly the set of functions $\curly{\tilde{\varphi}_{_{k}}(t)}$, which is referred to as the dual basis of $\curly{\varphi_{_{k}}(t)}$. 
By construction, any function $g(t)\in \mathcal{F}$ can be expressed by using a series expansion, such that $g(t)=\sum_{k}\eta_{_{k}}\varphi_{_{k}}(t)$, %follows:
%\begin{IEEEeqnarray}{lCr}
\label{Coeff:Expansion}
%g(t)=\sum_{j,k}\delta_{_{j,k}}\psi_{_{j,k}}
%\end{IEEEeqnarray}

where $\eta_{_{k}}$ and $\varphi_{_{k}}$ are the expansion coefficients and expansion functions, respectively. The set $\curly{\varphi_{_{k}}(t)}$ is biorthogonal to the set $\curly{\tilde{\varphi}_{_{k}}(t)}$ if $\langle \varphi_{_{k}},\tilde{\varphi}_{_{k'}}\rangle = \mathfrak{d}(k-k')$ $\forall k$ and $k'$, with $\langle \cdot \rangle$ being the $L_2$ inner product and the function $\mathfrak{d}(\cdot)$ is the Kronecker delta.\footnote{$\mathfrak{d}(k)=\begin{cases}
	1 & \text{if } \abs{k}=0\\
	0 & \text{if } \abs{k}>0
	\end{cases}$. The set $\varphi_{_{k}}(t)$ is orthogonal if $\langle \varphi_{_{k}},\varphi_{_{k'}}\rangle = 0$ $\forall k\ne k'$.} These two sets form a biorthogonal system, in which $\curly{\tilde{\varphi}_{_{k}}(t)}$ is referred to as the dual basis of $\curly{\varphi_{_{k}}(t)}$. %The function $\tilde{\psi}_{_{j,k}}$ is referred to as the dual of $\psi_{_{j,k}}$.	
%\begin{definition}[Biorthogonality]
%The sets $\psi_{_{j,k}}(t)$ and $\tilde{\psi}_{_{j,k}}(t)$ are biorthogonal if $\langle \psi_{_{j,k}},\tilde{\psi}_{_{j',k'}}\rangle = 0$ $\forall j\ne j\text{ or } k\ne k'$. The function $\tilde{\psi}_{_{j,k}}$ is referred to as the dual of $\psi_{_{j,k}}$.	
%\end{definition}
%Given that set $\curly{\varphi_{_{k}}(t)}$ is biorthogonal to the set $\curly{\tilde{\varphi}_{_{k}}(t)}$
Thus, we get the following unique representation:% by utilizing an inner product on $g(t)$:% an analytic representation of the coefficients:% and by 
\begin{IEEEeqnarray}{lCr}
	\label{coeff:rep} 
	\eta_{_{k}}=\langle g(t),\tilde{\varphi}_{_{k}}(t)\rangle
\end{IEEEeqnarray}

Substituting each $\eta_k$ coefficient with its analytic expression in \eqref{coeff:rep}, to obtain:
\begin{IEEEeqnarray}{lCr}
	\label{func:Expansion}
	g(t)=\sum_{k}\langle g(t),\tilde{\varphi}_{_{k}}(t)\rangle\varphi_{_{k}}(t)
\end{IEEEeqnarray}

%In the biorthogonal design there are two different stages. 
% as reflected by the inner product in \eqref{func:Expansion}
Obviously, in the present case of biorthogonality, the coefficients in \eqref{coeff:rep} are obtained by using the dual basis and the function is reconstructed in \eqref{func:Expansion} by using another basis which is the expansion set. In cases where $\curly{\tilde{\varphi}_{_{k}}(t)}=\curly{\varphi_{_{k}}(t)}$ we have an orthogonal basis $\curly{\varphi_{_{k}}(t)}$, which is referred to as self-dual.  Therefore, biorthogonality is a generalization of orthogonality that allows for a larger class of expansions.
\newcommand{\Penalty}[1]{P_{\lambda_j,\gamma_j}({#1})}
\newcommand{\PenatlyItem}[1]{\rho_{\lambda_j,\gamma_j}(\abs{{#1}})}

Recall that our objective is to estimate the bias term for an unknown functional form, captured by $\mathcal{M}_1(\cdot)$  in \eqref{Truncated:Substatnive:System},  the conditional expectation of $\varepsilon_{i}$, given participation defined as $\mathbb{E}[\varepsilon_{i}|y_{2i}=1]=\mathcal{M}_1(\B{w}_i^T\B{\gamma})$.\footnote{For brevity, we present $\mathcal{M}_1(\cdot)$ only. Identical treatment is applied to $\mathcal{M}_2(\cdot)$. } In what follows, we attend to the estimation of $\mathcal{M}_1(\cdot)$ using the wavelet estimator. % denoising technique.
%For brevity, \eqref{Truncated:Substatnive:System} is simplified using the scalars $u_i=y_{1i}-\B{x}_i^T\B{\beta}$ and $t_i= \B{w}_i^T\B{\gamma}$ to get:
%Although the function of interest is $\mathcal{M}_1(\cdot)$, we suppose that only noisy data about $\mathlcal{A}\mathcal{M}_1(\cdot)$ are observed, where $\mathlcal{A}$ is some linear operator. 
%In order to obtain the estimated We address to the main problem that might arises in the estimation procedure due to ill-posed instability  
%Let $L_2(\mathbb{R})$ be the space of  square integrable and real-valued functions on $\mathbb{R}$. We define a one-to-one linear operator $\mathlcal{A}$ from $L_2(\mathbb{R})$ to $L_2(\mathbb{R})$ by. . For this purpose, we denote

%The relationship between the observed data and the function to be estimated can be formulated by utilizing $\mathlcal{A}$ a one-to-one linear operator from $L_2(\mathbb{R})$ to $L_2(\mathbb{R})$, 
We use the concept of \textit{a frame} in \eqref{Definition:Frame} to define \textit{Riezs basis} in \eqref{Definition:Riesz:basis}. \textit{Riezs basis} is a building block in the definition of biorthogonal wavelets in \eqref{Definition:BioWavelet} to follow.  %this basis is a special case of a frame.

Let $\mathbb{H}$ be a separable Hilbert space with inner product $\langle\cdot,\cdot\rangle$ and a norm $\normsq{\cdot}{2}$. We denote a sequence  $\mathcal{F}=\curly{f_k, k\in \Lambda}\subset\mathbb{H}$, in which $\Lambda \subset \mathbb{Z}$.

We use the following \textit{frame} and \textit{Riesz basis} definitions \cite{zalik1999riesz}:
\begin{definition}[Frame]\label{Definition:Frame}
	$\mathcal{F}$ is called a frame if there are constants $0<A\le B$ such that $\forall f\in\mathbb{H}$ $A\normsq{f}{2}\le\sum_{k\in\Lambda} \abs{\langle f, f_k\rangle}^2\le B\normsq{f}{2}$.%\footnote{  is called a \textit{Riesz basis} if its span is dense in $\mathbb{H}$. }
\end{definition}

\begin{definition}[Riesz basis]\label{Definition:Riesz:basis}
	A sequence $\mathcal{F}$ is a \textit{Riesz basis} if and only if it is a frame having the additional property that upon the removal of any element from the sequence, it ceases to be a frame. %\footnote{  is called a \textit{Riesz basis} if its span is dense in $\mathbb{H}$. }
\end{definition}

Let $L_2(\mathbb{R})$ be the space of square integrable and real-valued functions on $\mathbb{R}$.  We use the following biorthogonal wavelet definition  \cite{dicken1996wavelet}:
\begin{definition}[Biorthogonal wavelet]\label{Definition:BioWavelet}
	A pair of functions $\varphi_{_{j,k}},\tilde{\varphi}_{_{j,k}}\in L_2(\mathbb{R})$ is a pair of biorthogonal wavelets if the sets $\curly{\varphi_{_{j,k}}|j,k\in \mathbb{Z}}$ and $\curly{\tilde{\varphi}_{_{j,k}}|j,k\in \mathbb{Z}}$ form the Riesz basis for $L_2(\mathbb{R})$ and if any function $g\in L_2(\mathbb{R})$ has the representation:
	%\vspace{-1em} 
	%\begin{IEEEeqnarray}{lCr}
	\label{BioWavelet}
	$g=\sum_{j\in \mathbb{Z}}\sum_{k\in \mathbb{Z}}\langle g,\tilde{\varphi}_{_{j,k}} \rangle \varphi_{_{j,k}}$.% \nonumber
	%\end{IEEEeqnarray}
	
\end{definition}

%The challenge in nonparametric estimation is solving the  %ill-posed problem, address.
It is worth noticing that both existence as well as uniqueness of the series representation are satisfied in definition \ref{Definition:BioWavelet}. However, our proposed nonparametric estimator might be unstable, rendering the estimation problem ill-posed, which is one of the challenges in nonparametric estimation of unknown functions.\footnote{An estimator violating at least one of the requirements: existence, uniqueness and stability is referred to as ill-posed.} This ill-posed problem can be alleviated by employing regularization on the wavelet series expansion coefficients \cite{abramovich1998wavelet,horowitz2014adaptive}.

%The main problem 
%For generality, we allow for a linear transform of the function to be observed in the data \cite{abramovich1998wavelet}.
%To address the issue of ill-posed We desire to explain the main difficulty in the estimation, and how it can be resolved and in order to justify the proposed algorithm. %We address to the ill-posed problem in the estimation of $\mathcal{M}_1$.
%\cite{abramovich1998wavelet} argue that the ill-posed problem might arise in a more general case in which a linear transform of the unknown function is observed in the data. 

%For this purpose, we define a one-to-one linear operator $\mathlcal{A}$ from $L_2(\mathbb{R})$ to $L_2(\mathbb{R})$, where $L_2(\mathbb{R})$ is the space of square integrable and real-valued functions on $\mathbb{R}$. 
Let $u_i=y_{1i}-\B{x}_i^T\B{\beta}$ be the $i$'th element in vector $\B{u}$ of size $n\times 1$, which satisfies: 
\begin{IEEEeqnarray}{lCr}
	\label{Vector:U}
	u_i=\mathcal{M}_1(t_i)+\epsilon_{_{1i}}
\end{IEEEeqnarray}

where $\curly{t_i}_{i=1}^{n}$ is a sequence in which the $i$'th element satisfies  $t_i=\B{w}_i^T\B{\gamma}$ and $\epsilon_{1i}$ is the white noise described in \eqref{Truncated:Substatnive:System}.\footnote{A more general formulation solves the ill-posed problem by employing regularization in cases where a linear transform of the unknown function replaces the original function \cite{abramovich1998wavelet}.}

We use $\B{\Phi}_{_I}$ and $\B{\Phi}_{_F}\equiv \B{\Phi}_{_I}^{-1}$ to denote the inverse and forward transformation matrices, respectively, each of size $n\times n$ \cite{he2011computer} in an orthogonal wavelet.\footnote{The forward transform is referred to as the Discrete Wavelet Transform (DWT).} We note that using orthogonal wavelets one obtains the close-form solution to the wavelet coefficients as follows: 
\begin{IEEEeqnarray}{lCr}
	\label{Denoising}
	\widehat{\mathcal{M}_1} = \B{\Phi}_{_I}\rho_{_{\lambda}}(\B{\Phi}_{_F}\B{u})
\end{IEEEeqnarray}

where $\widehat{\mathcal{M}_1}(\cdot)$ is the estimate of the unknown function $\mathcal{M}_1(\cdot)$, and $\rho_{_{\lambda}}(\cdot)$  represents the element-wise thresholding generated by some penalty function, in which the tuning parameter is represented by $\lambda$. The procedure in \eqref{Denoising} to obtain $\widehat{\mathcal{M}_1}(\cdot)$ by employing a thresholding operator is referred to as denoising.   
%where $\B{\Phi}^T=\B{\Phi}^{-1}$.

We depart from the denoising procedure in \eqref{Denoising} by employing biorthogonal wavelets, as we are interested in the applicability of the general case where the penelization is not an element-wise due to correlations among wavelet regressors.\footnote{It has been shown that the performance of wavelet estimator can be improved
	when the dependencies among coefficients were taken
	into account \cite{tomassi2015wavelet}.} In such cases, there is no such a close-form solution, which necessitates the regularized least squares optimization method to follow.%, the penalization is not element-wise, as suggested in our procedure to follow. 

Let $\B{\Psi}_{_I}$ and $\B{\Psi}_{_F}\equiv(\B{\Psi}_{_I}^T\B{\Psi}_{_I})^{-1}\B{\Psi}_{_I}^T$ denote the inverse and forward transformation matrices, respectively, each of size $n\times n$ of the wavelet-based JPEG, which is a biorthogonal wavelet.\footnote{For a definition of biorthogonal wavelets see \cite{daubechies1998factoring}.}$^{,}$\footnote{%Let denote the wavelet inverse transform matrix by $\B{X}$ and a signal $\B{Y}$. The first generation wavelets are orthogonal, implying that the wavelet coefficients are obtained as $\B{\delta}=\B{X}^T\B{Y}$ and allow for perfect reconstruction  of the signal ($\B{Y}=\B{X}\B{\delta}$).	
	Unlike biorthogonality, orthogonality implies that the wavelet regressors are mutually uncorrelated and that the inverse transform is the transpose of the forward transform. This simplifies the computation as the wavelet coefficients are obtained analytically (closed-form) using element-wise thresholding operators (e.g., hard and soft thresholding operators). However, we opted for the biorthogonality wavelet to exploit the correlation structure of the regressors. Biorthogonal wavelets preserve the perfect reconstruction property (by employing dual-filters) as well, but is more flexible in that the inverse of $\B{X}$ is not required to be its transpose. Consequently, the thresholding is applied to the entire coefficient vector.  }  
%the estimate  is the estimates of $\mathcal{M}_1(\cdot)$

Let $\rho_{\lambda,\gamma} (\cdot)$ be the minimax concave penalty (MCP) function \cite{zhang2010nearly}, defined as \cite{breheny2011coordinate}:\footnote{The penalty function in \eqref{MCP} represents a family of penalty functions as a generalization of the soft thresholding (if $\gamma\to\infty$) and hard thresholding (if $\gamma\to 1^{+}$) \cite{breheny2011coordinate}.}
\begin{IEEEeqnarray}{lCr}
	\label{MCP}
	\rho_{\lambda,\gamma}(\theta) = \begin{cases}
		\lambda \theta - \frac{\theta^2}{2\gamma} & \text{if } \theta\le \gamma \lambda, \\
		\frac{1}{2}\lambda^2\gamma & \text{if } \theta> \gamma \lambda 
	\end{cases}
\end{IEEEeqnarray}

where $\theta\in(-\infty,\infty)$ is the parameter to be penalized, $\lambda>0$ and $\gamma\in (1,\infty)$.

We define resolution-dependent regularized least squares at resolution levels $1,...,J$:
\begin{IEEEeqnarray}{lCr}
	\label{Optimization}
	\widehat{\B{\delta}}=\underset{\B{\delta}}{\arg\min}\frac{1}{2n}\normsq{\B{u}-\B{\Psi}_{_I}\B{\delta}}{2}+\sum_{j=1}^{J}\Penalty{\B{\delta}_j}
\end{IEEEeqnarray}

%, \hspace{1.5em} \widehat{\mathcal{M}}_1(t_i)=(\B{\Psi}_{_I}\B{\widehat{\delta}}_j)_i 
%\mathcal{G}_{\lambda_j}^k(\B{\delta}_j)
where $\B{\delta}=[\B{\delta}_{1}^T,\B{\delta}_{2}^T,...,\B{\delta}_{J}^T]^T$ is the wavelet coefficient vector of size $n\times 1$ and $\B{\delta}_j$ is of size $n_j\times 1$. $\norm{\B{\cdotp}}{2}$ is the usual $\ell_2$ (Euclidean) norm, defined as $\norm{\B{b}}{2}=\close{\sum_{i=1}^{n}\abs{b_{i}}^2}^{1/2}$.%and $\widehat{\mathrm{c_i}}$ is the estimated class membership of the $i$'th observation.
The penalty function is $\Penalty{\B{\delta}_j}=\sum_{k=1}^{n}\PenatlyItem{\delta_{j,k}}$.

It is evident that when the $\B{\delta}_j\to 0$, the bias propagated by the endogenous truncation approaches zero and thus, our algorithm is reduced to the conventional IV estimator. 

The univariate solution of a regularized least squares problem using the penalty function in \eqref{MCP} is denoted by $S_{\alpha}(\cdot)$ and defined as:\footnote{In order to utilize the min-max concave (MCP) penalty function in \eqref{MCP}, we depart from the regularized least squares algorithm in \cite{yang2015sparse}, as it is limited to its special case of the LASSO penalty function. We introduce  $\alpha$ as an approximation to the Hessian of the least squares problem in order to obtain an element-wise thresholding. This amounts to a dimensional reduction technique for reducing computational complexity. For the special case of $\alpha=1$, see \cite{breheny2011coordinate}. }
\begin{IEEEeqnarray}{lCr}
	\label{S:Thresholding:Solution}
	\scaleobj{0.8}{S_{\alpha}(\tilde{\delta}; \lambda, \gamma) = \begin{cases}
			\frac{1}{1-1/(\alpha\gamma)}\text{sign}(\tilde{\delta})\text{max}( \abs{\tilde{\delta}}-\frac{\lambda}{\alpha})  & \text{if } \abs{\tilde{\delta}}\le \gamma \lambda, \\
			\tilde{\delta}  & \text{if } \abs{\tilde{\delta}} >\gamma \lambda
	\end{cases}}
\end{IEEEeqnarray}

where $\alpha\in(0,\infty)$. It is worth noting that if $\gamma\to \infty$ the solution is soft-thresholding introduced by \cite{donoho1994ideal}; in case that $\alpha\gamma\to 1^{+}$ the solution is hard-thresholding (see proof in the Appendix \ref{AppendixA}).

To reduce computational complexity, the optimization problem in \eqref{Optimization} is reformulated as: 
\begin{IEEEeqnarray}{lCr}
	\scaleobj{0.8}{\B{\delta}^{(iter+1)}=\underset{\B{\delta}}{\arg\min}\frac{1}{2n}((\B{\delta}-\B{\delta}^{(iter)})^T\B{\Psi}_{_I}^T(\B{u}-\B{\Psi}_{_I}\B{\delta}^{(iter)}))}\\ \scaleobj{0.8}{+\alpha/2\normsq{\B{\delta}-\B{\delta}^{(iter)}}{2}+\sum_{j=1}^J\Penalty{\B{\delta}_j}}\nonumber 
\end{IEEEeqnarray}

where $\Psi_{_I}^T$ is the transpose of matrix $\Psi_{_I}$, $\alpha \B{I}$ is an approximation of the Hessian and $\B{I}$ is the identity matrix of size $n\times n$. The number of iterations is denoted by the integer $iter$.

For brevity, we divide the argument to be minimized by $\alpha$ and complete the squares using the expressions in $\normsq{\cdot}{2}$ \cite{yang2015sparse} to get:
\begin{IEEEeqnarray}{lCr}
	\label{Objective:function}
	\scaleobj{0.75}{\B{\delta}^{(iter+1)}=\underset{\B{\delta}}{\arg\min}\frac{1}{2}\normsq{\B{\delta}-\B{\delta}^{(iter)}+1/(\alpha n)\B{\Psi}_{_I}^T(\B{u}-\B{\Psi}_{_I}\B{\delta}^{(iter)})}{2}}\\ \scaleobj{0.75}{+\frac{1}{\alpha}\sum_{j=1}^J\Penalty{\B{\delta}_j} }\nonumber
\end{IEEEeqnarray}

The iterative procedure performs MCP-thresholding on a proximal gradient-descent update for $k=1,...,n$ (see Algorithm  \ref{PenalizedGradientDecent} in the Appendix):
\begin{IEEEeqnarray}{lCr}
	\label{Optimization:Gradient:decent}
	\scaleobj{0.8}{\delta_{j,k}^{(iter+1)}=S_{\alpha}\close{\delta_{j,k}^{(iter)}+1/(\alpha n)\B{\psi}_{_{I_k}}^T(\B{u}-\B{\Psi}_{_I}\B{\delta}^{(iter)}); \lambda_j, \gamma_j}}
\end{IEEEeqnarray}

where $\delta_{j,k}^{(iter)}$ is the $k$'th coefficient in vector $\B{\delta}_{j}^{(iter)}$ and $\B{\psi}_{_{I_k}}$ is $k$'th column in $\B{\Psi}_{_I}$ (the inverse wavelet transform). The notation $\B{\psi}_{_{I_k}}^T$ implies the transpose of  $\B{\psi}_{_{I_k}}$.
We use \eqref{Optimization:Gradient:decent} to update the wavelet coefficients iteratively until the update is negligible, such that the following convergence criterion is satisfied:
\begin{IEEEeqnarray}{lCr}
	\scaleobj{0.8}{\normsq{\B{\delta}^{(iter+1)}-\B{\delta}^{(iter)}}{2}/\normsq{\B{\delta}^{(iter)}}{2} < \tau}
\end{IEEEeqnarray}

where $\tau$ is the tolerance which is a positive real number that we arbitrarily set to $10^{-16}$.

%Its product with $T(\B{u}_j-\B{\Psi}_{_I}\B{\delta}_{j,k}^{(iter)})$ is solved by employing a lifting scheme instead of a cumbersome matrices multiplication.

The optimization method in \eqref{Optimization:Gradient:decent} involves matrices multiplication which is computationally infeasible for large data sets. To alleviate this computational complexity we develop a lifting scheme to be employed in order to perform simultaneously the transposed-inverse of the wavelet transform, consisting of lifting steps (see Algorithms \ref{Algrorithm:ForwardTransform}-\ref{Algrorithm:TransInvTransform} to follow). Conventionally, a lifting step can be either a prediction, that is a procedure generating a smoothed version of the data (the scaled coefficients), or an update that is the procedure to generate the remainder (the detail coefficients) between the data and its smoothed version. For the present case we define a new operator because the existing lifting steps do not provide an analytic representation of the transposed-inverse, as discussed in \cite{voronin2015compression}.  %These reminders are referred to as the wavelet `detail coefficients'. % which perform a scaled-moving average and scaled first-difference on the data, respectively. % prediction (a scaled moving-average) and update (a scaled first-difference) coefficients as follows: employed.   

In the next section we discuss the main idea behind lifting steps, in order to obtain analytically the transposed-inverse transform. First we describe the lifting steps in a regular-spaced data given a sample size of $2^{J}$ for a non-negative integer $J$. Then in equations \eqref{Irregular:Lifting}-\eqref{Optimal:Weights:Sample:Size} to follow, we alleviate these two restrictions by formulating our proposed algorithm. 
\subsubsection{Lifting steps to obtain the wavelet coefficients}
Let $w=\close{w_1,...w_n}$ be a discrete sequence of data consisting of $n$ real numbers, such that the sequence is referred to as \textit{dyatic} iff $n=2^J$ for some integer $J\ge 0$. The sequence can be expressed uniquely in terms of detail (difference) and summation coefficients denoted by $\curly{d_{_{J-1},k}}_{k=1}^{n/2}$ and $\curly{c_{_{J-1},k}}_{k=1}^{n/2}$, respectively. The former capture the variation in the sequence at different scales and locations and the latter are a smooth representation of the original sequence.

The multi-scale representation of a function $g\in L_2(\mathbb{R})$ is  obtained as follows:
\begin{IEEEeqnarray}{lCr}
	\label{Wavelet:regression}
	g(t)=\sum_{k\in \mathbb{Z}} c_{_{0,k}}\phi_{_{0,k}}(t) + \sum_{j\in \mathbb{Z}}\sum_{k\in \mathbb{Z}} d_{_{j,k}}\varphi_{_{j,k}}(t)
\end{IEEEeqnarray}

The first set of terms, $\phi_{_{0,k}}$, represents the average level of function $g$ and the second set of terms $\varphi_{_{j,k}}$ represents its details by accumulating information at a set of scales $j\in\mathbb{Z}$.

Let $\curly{0,...,J-1}$ denotes a set of scales (resolution levels).
We define $d_{_{J-1},k}$ and $c_{_{J-1},k}$ as follows \cite{nason2010wavelet}:
\begin{IEEEeqnarray}{lCr}
	\label{Sequence:finest}
	\begin{matrix}
		d_{_{J-1},k} = w_{2k}-w_{2k-1}, \hspace{2em} k=1,...,2^{J-1}. \\
		c_{_{J-1},k} = w_{2k}+w_{2k-1}, \hspace{2em} k=1,...,2^{J-1}.
	\end{matrix}
\end{IEEEeqnarray}

The key idea is that a lower detail coefficient $d_{_{J-1},k}$ implies that $w_{2k}$ is very close to $w_{2k-1}$ and visa versa, as such a smoother function is represented by a small sequence of detail coefficients. 

In order to represent the sequence in a coarser-scale (using a lower resolution), we define the coefficients:
\begin{IEEEeqnarray}{lCr}
	\label{Sequence:coarser}
	\begin{matrix}
		\scaleobj{0.8}{d_{_{J-2},k} = c_{_{J-1},2k}-c_{_{J-1}, 2k-1}, \hspace{2em} k=1,...,2^{J-2}.} \\
		\scaleobj{0.8}{c_{_{J-2},k} = c_{_{J-1},2k}+c_{_{J-1}, 2k-1}, \hspace{2em} k=1,...,2^{J-2}.}
	\end{matrix}
\end{IEEEeqnarray}

By repeating the procedure in \eqref{Sequence:coarser} we obtain detailed and smoothed coefficients for lower resolutions.
The multiscale algorithm stops when the $c_{_{0},1}$ coefficient is produced. %defined by the index $j$ ranging from $0$ to infinity. 
%The two sequences in \eqref{Sequence:finest} are the To obtain the next coarsest detail  %To obtain the next detail coefficients:
%Multi-resolution analysis is a device to extract information (captured by details) from $w$ at different scales and locations. For doing so,

Next we discuss how to select optimally the thresholding (tuning) parameter in \eqref{Optimization} for each resolution-level. 
\subsubsection{Optimal thresholding by a reference-free criterion function}\label{Section:Optimal:Thresholding}

Since we deal with truncated distributions, the source (the complete non-truncated distribution) is intrinsically unobservable and thus, we cannot assess the success of denoising by comparing it to the original non-truncated distribution. Therefore, we utilize the ``two-fold cross-validation" in \cite{nason1996wavelet} which is a reference-free criterion function assesing the quality of the function estimated by denoising:\footnote{The methodology implemented in \cite{nason1996wavelet} chooses one threshold that is applicable to all resolution levels in the wavelet transform. In the present case, however, we select a threshold for each level in order to implement multi-resolution analysis increasing our proposed estimator's accuracy. }
\begin{IEEEeqnarray}{lCr}
	\label{Optimal:Thresholding}
	\scaleobj{0.8}{\lambda_{j,\gamma_{_j}}=\underset{\lambda_{j,\gamma_{_j}}}{\arg\min}\curly{\frac{1}{2}\normsq{\widehat{f}_{\lambda_{j,\gamma_{_j}}}^o-\B{u}_j^e}{2}+\frac{1}{2}\normsq{\widehat{f}_{\lambda_{j,\gamma_{_j}}}^e-\B{u}_j^o}{2}}}
\end{IEEEeqnarray}

where $\lambda_{j,\gamma_{_j}}$ is the tuning (thresholding) parameter being used in \eqref{Optimization:Gradient:decent}, in which $\gamma_{_j}$ is a specific penalty function. The odd sample  and even sample are denoted by $\B{u}_j^o$ and  $\B{u}_j^e$, respectively, and their corresponding estimates are $\widehat{f}_{\lambda_{j,m}}^o$ and $\widehat{f}_{\lambda_{j,m}}^e$. These estimates are obtained by employing the iterative procedure in \eqref{Optimization:Gradient:decent}. %utilizing either the odd observations or the even observations.

\color{black} As previously discussed, our proposed truncation-proof IV estimator requires controlling for the bias terms $\mathcal{M}_1(\cdot)$ and $\mathcal{M}_2(\cdot)$. For generality and applicability purposes of the proposed estimator, we adopt a semiparametric approach which is not subjected to distributional assumptions and consequently, does not require specifying the functional form of these unknown functions. %the joint distribution function of the random disturbances $\disturb_{1i}$ and $\disturb_{2i}$. Thus, our estimation strategy is semiparametric.

The wavelet-based JPEG semiparametric estimator is chosen for its many advantages. It enables a multi-resolution representation of the noisy data points, implying that the data points are characterized both globally as well as locally.\footnote{Global representation is a weighed average (smoothing) of the data, while local representation consists of more detailed information regarding first differences between neighboring data points. } Such a multi-resolution decomposition facilitates distinguishing between the noise and the systemic part. The systemic part is the functional relationship between the covariates and the dependent variables in the regression equations in \eqref{Truncated:Substatnive:System}. 

An additional advantage of incorporating the aforementioned newly introduced JPEG estimator is that it accommodates for various data set forms of different types of irregularities, such as non-equispaced design that will be described in section \ref{Section:Adusting:Location} to follow. These irregularities are alleviated by introducing the locations in space of the various data points as an additional covariate that is unique for each resolution-level. An additional merit of our approach is enabling a group-wise denoising rather than the traditional element-wise JPEG denoising in the cases of image processing. Group-wise denoising plays an important role in data denoising, as it takes into account potential dependencies among the various data points. Thus, our contribution to the JPEG algorithm are controlling for irregularities, group-wise thresholding on the entire data and utilizing a reference-free criterion function to choose the optimal thresholding. \color{black}

In next section we describe the JPEG algorithm which is introduced to estimate each of the bias terms $\mathcal{M}_1(\cdot)$ and $\mathcal{M}_2(\cdot)$. Although our proposed denoising procedure can be applicable to both even as well as odd sample sizes (as will be demonstrated in Algorithm \ref{Algrorithm:ForwardTransform} to follow), for ease of presentation and without loss of generality, the denoising procedure is formulated as a function of a data set consisting of $2n$ observations. %Our proposed algorithm is amenable for denoising any sample size observations (as will be discussed in \ref{section:Irregular:Forward:N_observations}).

\section{The wavelet-based JPEG denoising}\label{Section:JPEG:denoising}
Let $\curly{(u_i, t_i)}_{i=1}^{2n}$ be a pairwise sequence of $2n$ data points as described in \eqref{Vector:U}, such that $t_i<t_j$ $\forall i<j$. The sequence $\curly{u_i}_{i=1}^{2n}$ indicates the noisy data points (or colors of pixels in image processing) and their respective locations in space are represented by $\curly{t_i}_{i=1}^{2n}$.  The JPEG algorithm is a procedure generating a multi-resolution denoised representation of the sequence $\curly{u_i}_{i=1}^{2n}$, which is denoted by the sequence $\curly{\widehat{u}_i}_{i=1}^{2n}$.
The purpose of the present section is three-fold: first, to describe the JPEG algorithm to be employed in order to obtain a noise-free representation of the noisy data; second, to extend the JPEG algorithm to be compatible with irregularities in the data;\footnote{We define $\Delta_i\equiv t_{i}-t_{i-1}$ $\forall$ $2\le i\le n$, such that equispaced (regular) data is a sequence of data points satisfying $\Delta_i=\Delta_j$ $\forall i$ and $j$. Other cases are referred to as non-equispaced (irregular) spaced data.} and thirdly, to incorporate a reference-free criterion to evaluate the denoising procedure accuracy.

Applying the conventional JPEG algorithm on a vector of data points is equivalent to employing three different procedures on the noisy data: (i) the JPEG forward transform $\mathcal{T}_{_{F}}:\mathbb{R}^{2n\times 2}\to \mathbb{R}^{2n\times 1}$ to obtain the wavelet-based JPEG coefficients (as will be shown in \eqref{Forward:Operator} to follow); (ii) coefficients selection $\mathcal{T}_{_{S}}:\mathbb{R}^{2n\times 2}\to \mathbb{R}^{2n\times 1}$ by applying a thresholding procedure (as will be shown in \eqref{Transposed:Inverse:Operator:Matrix}) and (iii) the JPEG inverse transform $\mathcal{T}_{_{I}}:\mathbb{R}^{2n\times 2}\to \mathbb{R}^{2n\times 1}$, which recovers the noise-free data by utilizing the selected coefficients (as will be shown in \eqref{Inverse:Operator} to follow). 

\color{black} In the ensuing section \color {black} we introduce auxiliary matrices to be used in each of the JPEG transforms, which are essential to construct the covariate matrix in the wavelet-based JPEG regression (in \eqref{JPEG:Wavelet:regression} to follow).
%For brevity, we use vector notations $\B{u}=[u_1,...,u_n]^T$, $\B{t}=[t_1,...,t_n]^T$ and $\B{\delta}\in\mathbb{R}^n$.
%Differently from the wavelet nonparametric regression in \eqref{Wavelet:regression}, these three procedures are implemented simultaneously which allows for a non element-wise denoising. This can be achieved by introducing the transposed of the wavelet inverse transform. 
\subsection{Auxiliary matrices for the JPEG alogrithm}
The implementation of the JPEG algorithm necessitates the construction of $\mathcal{T}_{_{F}}$ and $\mathcal{T}_{_{I}}$ operators. For this purpose,  we construct auxiliary matrices $\mathcal{A}_{_{2n}}$, $\mathcal{S}_{_{2n}}$, $\curly{\mathcal{H}^{\B{(t)}}_{_{2n,\ell}}}_{\ell=1}^{4}$ which are the shifting, rescaling and smoothing operator matrices, respectively, each of size $2n\times 2n$. %These operators enable a closed-form representation of the wavelet forward transform: 

Let $\mathcal{S}_{_{2n}}$ and $\mathcal{S}_{_{2n}}^{-1}$ be the rescaling and inverse-rescaling matrices, respectively each of size $2n\times 2n$. Its elements are defined for $m=0,...,n$ as:
\begin{IEEEeqnarray}{lCr}
	\scaleobj{0.65}{(\mathcal{S}_{_{2n}})_{i,j}=\begin{cases}
			1/\varphi & \text{if } i = j = 2m\\
			\varphi & \text{if } i = j = 2m+1\\
			0 & \text{if } i \ne j
		\end{cases}, \hspace{0.25em}(\mathcal{S}_{_{2n}}^{-1})_{i,j}=\begin{cases}
			\varphi & \text{if } i = j = 2m\\
			1/\varphi & \text{if } i = j = 2m+1 \\
			0 & \text{if } i \ne j
	\end{cases}}
\end{IEEEeqnarray}

The rescaling operator $\B{\tilde{v}}=\mathcal{S}_{_{2n}}\B{v}$ takes a vector $\B{v}$ of size $2n\times 1$ and return a rescaled vector $\B{\tilde{v}}$ of the same size, such that even and odd elements of the original vector are multiplied by the scalars $1/\varphi$ and $\varphi$, respectively.% $[\B{v}_{\text{odd}}^T, \B{v}_{\text{even}}^T]^T$

Let $\mathcal{A}_{_{2n}}$ be a shifting operator matrix of size $2n\times 2n$, its elements are defined for $m=0,...,n$ as:
\begin{IEEEeqnarray}{lCr}
	\scaleobj{0.8}{(\mathcal{A}_{_{2n}})_{i, j} =\begin{cases}
			1 & \text{if } (i > n,\hspace{0.5em} j=2m) \text { or } (i \le n,\hspace{0.5em} j=2m+1)\\
			0 & \text{otherwise}. 
	\end{cases}  }
\end{IEEEeqnarray}

The $\B{\tilde{v}}=\mathcal{A}_{_{2n}}\B{v}$ operator takes a vector $\B{v}=[v_{_{1}},...,v_{_{2n}}]^T$ of size $2n\times 1$ and return the vector $\B{\tilde{v}}=[\B{v}_{\text{odd}}^T, \B{v}_{\text{even}}^T]^T$. The vectors $\B{v}_{\text{odd}}=[v_{_{1}},...,v_{_{2n-3}},v_{_{2n-1}}]^T$ and $\B{v}_{\text{even}}=[v_{_{2}},...,v_{_{2n-2}},v_{_{2n}}]^T$ consist of the odd and even elements of $\B{v}$, respectively. %[v_{1}, v_{3},...,v_{2n-1},v_{2}, v_{4},...,v_{2n}]^T$. split the series $\curly{u_i}_{i=1}^n$ into odd and even elements and stores the odd elements at the moves the odd elements of the matrix,

Unlike the conventional JPEG, we allow for data irregularities by controlling for the data set location in space. For doing so, we denote a sequence of matrices $\curly{\mathcal{H}^{\B{(t)}}_{_{2n,\ell}}}_{\ell=1}^4$, such that the elements of matrix $\mathcal{H}^{\B{(t)}}_{_{2n,\ell}}$ $\forall \ell\in\curly{1,2,3,4}$ of size $2n\times 2n$ are defined for $m = 1,...,n$ as:%,\hspace{0.5em} j=1,...,2n$ as:
\begin{IEEEeqnarray}{lCr}
	\scaleobj{0.8}{\underset{\ell\in\curly{1,3}}{(\mathcal{H}^{\B{(t)}}_{_{2n,\ell}})_{i,j}}=\begin{cases}
			\substack{2\pi_{\ell}\omega^{\B{(t)}}_{_{\ell,i-1}}} & \text{if } \substack{i=2m \text{ and } j=i-1} \\
			\substack{1} & \text{if } \substack{j=i} \\
			\substack{2\pi_{\ell}(1-\omega^{\B{(t)}}_{_{\ell,i-1}})} & \text{if } \substack{i=2m \text{ and } j=i+1} \\
			0 & \text{otherwise }  
	\end{cases}}\\ \scaleobj{0.8}{\underset{\ell\in\curly{2,4}}{(\mathcal{H}^{\B{(t)}}_{_{2n,\ell}})_{i,j}}=\begin{cases}
			\substack{2\pi_{\ell}\omega^{\B{(t)}}_{_{\ell,i-1}}} & \text{if } \substack{i=2m+1 \text{ and } j=i-1} \\
			\substack{1} & \text{if }  \substack{j=i}  \\
			\substack{2\pi_{\ell}(1-\omega^{\B{(t)}}_{_{\ell,i-1}})} & \text{if } \substack{i=2m+1 \text{ and } j=i+1} \\
			0 & \text{otherwise }.   
	\end{cases}}
\end{IEEEeqnarray}

%where $m = 1,...,n,\hspace{0.5em} j=1,...,2n$.
%For brevity, the JPEG forward and inverse tranforms are defined as the linear transformations $\mathcal{T}_{_{F}}(\B{u};\B{t})=\B{\Psi}_{_F}(\B{t})\B{u}$ and
%$\mathcal{T}_{_{I}}(\B{\delta};\B{t})=\B{\Psi}_{_I}(\B{t})\B{\delta}$, respectively. The matrices $\B{\Psi}_{_F}(\B{t})$ and $\B{\Psi}_{_I}(\B{t})$ are each of size $2n\times 2n$, and are 
\newcommand\mathf[1]{\B{\mathfrak{{#1}}}}
\newcommand\diagBlockMat[2]{\braket{\begin{matrix}{#1}^{(\B{t})}_{_{\mathf{m}(j)\times \mathf{m}(j)}} & \B{0}_{_{\mathf{m}(j)\times (\mathf{m}(J)-\mathf{m}(j))}} \\ \B{0}_{_{(\mathf{m}(J)-\mathf{m}(j))\times \mathf{m}(j)}} & {#2}_{_{(\mathf{m}(J)-\mathf{m}(j))\times (\mathf{m}(J)-\mathf{m}(j))}}\end{matrix}}}
%---
\newcommand\diagBlockMatA[2]{\braket{\begin{matrix}{#1}_{_{\mathf{m}(j)\times \mathf{m}(j)}} & \B{0}_{_{\mathf{m}(j)\times (\mathf{m}(J)-\mathf{m}(j))}} \\ \B{0}_{_{(\mathf{m}(J)-\mathf{m}(j))\times \mathf{m}(j)}} & {#2}_{_{(\mathf{m}(J)-\mathf{m}(j))\times (\mathf{m}(J)-\mathf{m}(j))}}\end{matrix}}}
where each of sequences $\curly{\omega^{\B{(t)}}_{_{\ell,l}}} \hspace{0.5em}\forall\ell \in\curly{1,2,3,4}$ are the interpolation weights (defined in \eqref{Interpolation:Weights:irregular} to follow) to control for the location in space of data points (enabling irregular non-equispaced data to be used) and $\pi_1,\pi_2,\pi_3,\pi_4$ are scalar constants described in \cite{schelkens2009jpeg}, which are referred to as the filter coefficients of the wavelet-based JPEG. In the special case in which $\omega^{\B{(t)}}_{_{\ell,l}}=0.5$ $\forall l$ and $\ell \in\curly{1,2,3,4}$ the algorithm is reduced to the regular-spaced wavelet-based JPEG. 

We define the linear interpolation weights:
\begin{IEEEeqnarray}{lCr}
	\label{Interpolation:Weights:irregular}
	\scaleobj{0.8}{\omega^{\B{(t)}}_{_{\ell,2i}} = \begin{cases}
			{\fracBig{{{{\tilde t}_{_{2i + 1}}} - {{\tilde t}_{_{2i}}}}}{{{{\tilde t}_{_{2i + 1}}} - {{\tilde t}_{_{2i - 1}}}}}} & \text{if } \ell \in\curly{1,3}\\
			{\fracBig{{{{\tilde t}_{_{2i + 2}}} - {{\tilde t}_{_{2i + 1}}}}}{{{{\tilde t}_{_{2i + 2}}} - {{\tilde t}_{_{2i}}}}}} & \text{if } \ell \in\curly{2,4}
	\end{cases}}\\ \scaleobj{0.8}{ \tilde {t}_{_l} = \begin{cases}
			t_{_2} & \text{if } l=2m, l < 2\\
			t_{_l} & \text{if } 1\le l\le 2n\\
			t_{_{2n-1}} & \text{if } l=2m+1, l > 2n-1
	\end{cases}}
\end{IEEEeqnarray}

For tractability, we formulate the JPEG coefficients estimation problem as a linear regression estimation, which necessitates obtaining a closed-form expressions of the JPEG forward and inverse transforms. These closed-form expression are required to characterize the JPEG covariate matrix to be used in the wavelet-based JPEG regression. \color{black} In the following section \color {black} we express analytically each of the forward and inverse transforms using matrix notation as a function of the auxiliary matrices presented above.
\subsection{The various JPEG transforms in matrix notation}\label{Section:Adusting:Location}
\color{black}Employing our proposed JPEG algorithm on a data set involves representation of data set in multiple resolution levels, a property which referred to as a multi-resolution analysis. Let $J$ be the highest resolution level, which requires the same number of data points as in the noisy data set. The data set representation in $j$'th resolution-level $\forall j<J$ is a transformation of the data set representation in the finer (higher) resolution-level $j+1$. Consequently, the JPEG noise-free representation $\forall j<J$ can be formulated recursively. However, the implication of this formulation is that the location in space of the data points in any given resolution-level is also determined recursively. This fact stems from depicting the noisy data set $\B{u}=[u_{_{1}},...,u_{_{2n}}]^T$ and its location in space $\B{t}=[t_{_{1}},...,t_{_{2n}}]^T$ as a pairwise sequence. For ease of notation we construct the adjusted space location operator $\forall j<J$: \color{black}% $\curly{\B{t}_{_j}}$:% shifting in space operator:
%Applying multiple transformations on the data sets implies that the location of the data set is 
%challenges in the presence of non-equispaced (irregular) spaced-data. 
%In order to allow for irregular data set in the presence of a multi-resolution analysis.
% the transformation that is applied to the data set must be applied also to the location in space of the data point must be taken into account. 
%We construct a space location operator, which is a linear transform of the vector $\B{t}$ of size $2n\times 1$, consisting of the noisy data set locations in space, defined for any resolution level $j<J$ as:
\begin{IEEEeqnarray}{lCr}
	\label{Adjusted:space:location}
	\scaleobj{0.7}{\B{t}_{_{j}}\equiv\close{\prod_{h=j+1}^{J} \tilde{\mathcal{A}}_{_{h}}}\B{t}, \hspace{0.25em} \tilde{\mathcal{A}}_{_{j}}=\diagBlockMatA{\mathcal{A}}{\B{I}}}
\end{IEEEeqnarray}

where $\mathf{m}(j)\equiv \lceil2n/2^{^{J-j}}\rceil$, $J=\lceil\log_2\close{2n}\rceil$ is the number of resolution-levels and $j$ is a specific resolution level.\footnote{The operator's notation $\lceil\cdot \rceil$ represents the ceiling of a real number.}
%that irregularities in the original data affect all the resolutions-level, which challenge the JPEG algorithm implementation in the presence of such irregularities.

This recursive formulation takes the the noisy data points locations in space as control variables, which are essential for alleviating irregularities in the noisy data. %, starting from the highest resolution level $J$.% data representation in the finer resolution level is a new data set.

Using the adjusted space location sequence $\curly{\B{t}_{_j}}$ in \eqref{Adjusted:space:location}, we define matrix $\B{\Psi_{_F}^{\B{(t)}}}$ (to be used in \eqref{Forward:Operator} to follow) for $J\in\curly{1,...,\lceil\log_2\close{2n}\rceil}$ resolution levels as:
\newcommand\JPEGT[1]{(\B{t}_{_{{#1}}})}
\begin{IEEEeqnarray}{lCr}
	\label{Forward:Operator:Matrix}
	\scaleobj{0.8}{\B{\Psi_{_F}^{\B{(t)}}}\equiv \close{\prod_{j=1}^{J-1} \tilde{\B{\Phi}}^{\JPEGT{j}}_{_{j}}} \B{\Phi}^{(\B{t})}_{\mathf{m}(J)\times \mathf{m}(J)}}\\ \scaleobj{0.8}{\tilde{\B{\Phi}}^{(\B{t})}_{_{j}}=\diagBlockMat{\B{\Phi}}{\B{I}}}
\end{IEEEeqnarray}

where $\B{\Phi}^{(\B{t})}_{m\times m}\equiv\mathcal{A}_{{_m}}\mathcal{S}_{{_m}} \mathcal{H}^{\B{(t)}}_{_{m,4}} \mathcal{H}^{\B{(t)}}_{_{m,3}} \mathcal{H}^{\B{(t)}}_{_{m,2}} \mathcal{H}^{\B{(t)}}_{_{m,1}}$ and $\B{I}_{m\times m}$ is the identity  matrix of size ${m\times m}$. It worth noticing that $\B{\Phi}^{\B{(t)}}_{m\times m}$ is a product of invertible matrices and consequently, its inverse is characterized as:
\begin{IEEEeqnarray}{lCr}
	\scaleobj{0.7}{\close{\B{\Phi}^{\B{(t)}}_{m\times m}}^{-1}=\close{\mathcal{H}^{\B{(t)}}_{_{m,1}}}^{-1} \close{\mathcal{H}^{\B{(t)}}_{_{m,2}}}^{-1} \close{\mathcal{H}^{\B{(t)}}_{_{m,3}}}^{-1} \close{\mathcal{H}^{\B{(t)}}_{_{m,4}}}^{-1}\mathcal{S}_{{_m}}^{-1}\mathcal{A}_{{_m}}^{-1}}
\end{IEEEeqnarray}

%However, the location in space of the data set depends on the resolution level, a fact that challenges the denoising algorithm. 

The multi-resolution JPEG forward transform operator $\mathcal{T}_{_{F}}(\B{u},\B{t})$ is the linear transform:
\begin{IEEEeqnarray}{lCr}
	\label{Forward:Operator}
	\scaleobj{0.8}{\B{\delta} =  \B{\Psi_{_F}^{(\B{t})}} \B{u}, \hspace{1em} \mathcal{T}_{_{F}}(\B{u},\B{t})=\B{\Psi_{_F}^{(\B{t})}} \B{u} }
\end{IEEEeqnarray}

where $\B{u}$, $\B{t}$ and $\B{\delta}$ are the noisy data, the location in space and the JPEG coefficient vectors, respectively, each of size $2n\times 1$.

Similarly, the multi-resolution JPEG inverse transform operator $\mathcal{T}_{_{I}}(\B{\delta},\B{t})$ is the linear transform:
\begin{IEEEeqnarray}{lCr}
	\label{Inverse:Operator}
	\scaleobj{0.8}{\B{u} =  \B{\Psi_{_I}^{\B{(t)}}} \B{\delta}, \hspace{1em} \mathcal{T}_{_{I}}(\B{\delta},\B{t}) = \B{\Psi_{_I}^{\B{(t)}}} \B{\delta}}
\end{IEEEeqnarray}

where $\B{u}$, $\B{t}$ and $\B{\delta}$ are the noisy data, the location in space and the JPEG coefficient vectors, respectively, each of size $2n\times 1$. 

Matrix $\B{\Psi_{_I}^{\B{(t)}}}$ in \eqref{Inverse:Operator} is defined for $J\in\curly{1,...,\lceil\log_2\close{2n}\rceil}$ resolution levels as:
\begin{IEEEeqnarray}{lCr}
	\label{Inverse:Operator:Matrix}
	\scaleobj{0.8}{\B{\Psi_{_I}^{\B{(t)}}}\equiv\close{\B{\Phi}^{(\B{t})}_{\mathf{m}(J)\times \mathf{m}(J)}}^{-1}\close{\prod_{j=J-1}^{1} \close{\tilde{\B{\Phi}}^{\JPEGT{j}}_{_{j}}}^{-1}}}
\end{IEEEeqnarray}
which $\B{\Psi_{_I}^{\B{(t)}}}$ is the analytic inverse transform operator.
%where $\B{\Phi}_{m\times m}^{-1}=\mathcal{H}_{_{m,1}}^{-1} \mathcal{H}_{_{m,2}}^{-1} \mathcal{H}_{_{m,3}}^{-1} \mathcal{H}_{_{m,4}}^{-1}\mathcal{S}_{{_m}}^{-1}\mathcal{A}_{{_m}}^{-1}$.

%The matrix $\B{\Phi}_{m\times m}$ is invertible $\forall m$ due to its characterization as a product of invertible matrices, and its inverse is denoted by $\B{\Phi}_{m\times m}^{-1}$. Thus, the inverse of the block diagonal $\tilde{\B{\Phi}}_{_{j}}$ is denoted by $\tilde{\B{\Phi}}_{_{j}}^{-1}$ described in \eqref{Forward:Operator:Matrix} is denoted as:

%\subsection{The irregular forward transform}
%The forward JPEG transform can be represented using a mapping $\mathcal{T}_{_{F}}:(\B{v},\B{t})\to \B{\delta}$, where  $\B{v}$, $\B{t}$ and $\B{\delta}$ are $n\times 1$ vectors. This mapping can be defined as a linear transformation:

%which implies as a linear transformation as follows: 

%For brevity, these procedures are represented by the operators:  $\mathcal{T}_{_{F}}:(\B{v},\B{t})\to \B{\delta}$ 

%and noisy data As we are treating with data instead of pixels,
We introduce the wavelet-based JPEG nonparametric regression given a non-equispaced irregular data set:
\begin{IEEEeqnarray}{lCr}
	\label{JPEG:Wavelet:regression}
	\scaleobj{0.8}{u(t)=\close{\B{\Psi_{_I}^{\B{(t)}}} \B{\delta}}_{_{t}}=\sum_{k\in \mathbb{Z}} c_{_{0,k}}\phi_{_{0,k}}(t) + \sum_{j\in \mathbb{Z}}\sum_{k\in \mathbb{Z}} d_{_{j,k}}\varphi_{_{j,k}}(t)}
\end{IEEEeqnarray}

where $\B{\delta}$ consists of the sequences of coefficients $\curly{c_{_{0,k}}}$ and $\curly{d_{_{j,k}}}$, which capture the function's average behavior and details, respectively.  The matrix $\B{\Psi_{_I}^{\B{(t)}}}$ consists of the sequences of covariates $\curly{\phi_{_{0,k}}(t)}$ and $\curly{\phi_{_{j,k}}(t)}$. Unlike the present case which employs a group-wise denoising on the entire data, in the conventional JPEG the coefficients selection operator, $\mathcal{T}_{_{S}}(
\B{\delta},\B{t})$, is constructed to be used element-wise (for each resolution-level $j$), e.g, $\mathcal{T}_{_{S}}(
\delta_{ij},\B{t})=S_{\alpha}(\delta_{ij}; \lambda_{j,\gamma_{_j}}, \gamma_{_j})$
using $S_{\alpha}(\cdot)$ in \eqref{S:Thresholding:Solution}, where $\lambda_{j,\gamma_{_j}}$ and $\gamma_{_j}$ are defined in \eqref{Optimal:Thresholding} using $\alpha=1$, such that $\curly{\delta_{ij}}_{i=1}^{\mathf{m}(j)}$ is a subset of vector $\B{\delta}$ consisting of the $j$'th resolution-level coefficients.

In the following section, we discuss about the JPEG group-wise coefficients selection to perform denoising.
\subsubsection{JPEG group-wise coefficients selection for irregular data denoising}
Lastly, we formulate the transpose of the inverse wavelet transform, in order to employ the group-wise denoising procedure depicted in \eqref{Optimization:Gradient:decent}:
%Next we show the various algorithms to be employed for estimating the parameters in the substantive equations.
%First, we get the forward wavelet transform:
\begin{IEEEeqnarray}{lCr}
	\label{Transposed:Inverse:Operator:Matrix}
	\scaleobj{0.67}{\tilde{\B{u}}=\close{\B{\Psi}_I^{\B{(t)}}}^T\B{u},\hspace{0.25em} \close{\B{\Psi}_I^{\B{(t)}}}^T\equiv\close{\prod_{j=1}^{J-1} \braket{\close{\tilde{\B{\Phi}}^{\JPEGT{j}}_{_{j}}}^{-1}} }^T \braket{\close{\B{\Phi}^{(\B{t})}_{\mathf{m}(J)\times \mathf{m}(J)}}^{-1}}^T}
\end{IEEEeqnarray}

where $\scaleobj{0.8}{\braket{\close{\B{\Phi}^{\B{(t)}}_{m\times m}}^{-1}}^T}$ is constructed as: 
\begin{IEEEeqnarray}{lCr}
	\label{H:Inverse:Transform}
	\scaleobj{0.8}{\braket{\close{\B{\Phi}^{\B{(t)}}_{m\times m}}^{-1}}^T=\braket{\mathcal{A}_{{_m}}^{-1}}^T\braket{\mathcal{S}_{{_m}}^{-1}}^T 
		\braket{\close{\mathcal{H}^{\B{(t)}}_{_{m,4}}}^{-1}}^T}\\ \scaleobj{0.8}{\times \braket{\close{\mathcal{H}^{\B{(t)}}_{_{m,3}}}^{-1}}^T \braket{\close{\mathcal{H}^{\B{(t)}}_{_{m,2}}}^{-1}}^T \braket{\close{\mathcal{H}^{\B{(t)}}_{_{m,1}}}^{-1}}^T } \nonumber
\end{IEEEeqnarray}

%$\close{\B{\Phi}_{m\times m}^{-1}}^T=(\mathcal{A}_{{_m}}^{-1})^T(\mathcal{S}_{{_m}}^{-1})^T (\mathcal{H}_{{_{m,4}}}^{-1})^T (\mathcal{H}_{{_{m,3}}}^{-1})^T (\mathcal{H}_{{_{m,2}}}^{-1})^T (\mathcal{H}_{{_{m,1}}}^{-1})^T$.

The JPEG estimator, denoted by $\widehat{\B{\delta}}$, is obtained by minimizing the objective function in \eqref{Objective:function} using the iterative procedure in \eqref{Optimization:Gradient:decent} given the chosen thresholding level. The latter is determined by minimizing the reference-free criterion function depicted in section \ref{Section:Optimal:Thresholding}.

The construction of the wavelet-based JPEG transforms matrix involves computational complexity, a problem which can be alleviated by employing a faster algorithm, referred to as `a lifting step'. \color{black} In the succeeding section \color {black} we discuss the algorithm to obtain a denoised representation of non equispaced data design by using a procedure that does not necessitate matrix operation to reduce computational complexity.
%\subsection{Non equispaced data design}

\subsection{The irregular forward transform}\label{section:Irregular:Forward:N_observations}
The irregular forward transform in \eqref{Forward:Operator} is the procedure to obtain the wavelet coefficients, $\B{\delta}$, as follows:
\newcommand\LiftingStepTransINV[5]{\close{\varpi_{_{#1,{l}}}^L{#2}_{_{#3}}^{({#4})} + \varpi_{_{#1,{l}}}^H{#2}_{_{{#5}}}^{({#4})}} }
\newcommand\LiftingStep[5]{\close{\omega_{_{#1,{#5}}}{#2}_{_{#3}}^{({#4})} + (1-\omega_{_{#1,{#5}}}){#2}_{_{{#5}}}^{({#4})}} }
%\varpi_{\ell,i}^L

%\subsection{Algorithms}\label{Appendix:Algorithms}
\begin{algorithm}[H]
	%\algsetup{linenosize=\tiny}
	\tiny
	\caption{The wavelet-based JPEG Forward Transform}\label{Algrorithm:ForwardTransform}\label{Irregular:Lifting}. 
	\begin{algorithmic}[1]
		\Procedure{ForwardTransform}{u, Grid, Level}
		%\Procedure{Euclid}
		\Require (i) Two real-number vectors $\textit{u}$ and Grid of size $n\times 1$; (ii) Level$\in\curly{1,...,\log_2(n)}$;	
		\Ensure $\text{Output}\gets$ a real-number coefficient vector $\delta$ of size $n\times 1$;					
		\State \emph{The JPEG filter coefficients}:
		\State $\pi_1 \gets -1.5861343420693648$; 
		\State $\pi_2 \gets -0.0529801185718856$; 
		\State $\pi_3 \gets \color{white}+\color{black}0.8829110755411875$; 
		\State $\pi_4 \gets \color{white}+\color{black}0.4435068520511142$
		\State $\varphi\color{white}_{1}\color{black} \gets \color{white}+\color{black}1.1496043988602418$			
		\State \emph{Start}:
		\State $\textit{n} \gets \text{length of }\textit{u}$
		\State $\textit{m} \gets \text{ceiling of } \close{\textit{n}/2}$
		\State $\textit{Q} \gets \text{ceiling of } \close{\textit{m}/2}$	
		\State $d \gets \text{copy the odd elements of } \textit{u}$
		\State $s \gets \text{copy the even elements of } \textit{u}$
		
		\State \emph{top}:
		\State NewGrid $\gets$ Grid
		
		\State 
        \emph{The parameter vector $\phi_{\delta}$ is used to generate a zero wavelet coefficient for the generated observation:}
        \State $\phi_{\delta}[1]=-2\times(\pi_1\times\pi_2\times\pi_3)/(1+2\times\pi_2\times\pi_3)$
        \State $\phi_{\delta}[2]=2\times(\pi_2\times\pi_3)/(1+2\times\pi_2\times\pi_3)$
        \State $\phi_{\delta}[3]=2\times(\pi_1+\pi_3+3\times\pi_1\times\pi_2\times\pi_3)/(1+2\times\pi_2\times\pi_3)$	
		\If {$\textit{n} \text{ is an odd number}$} 
		\State $\textit{s}[m] \gets \textit{d}[m-1]*\phi_{\delta}[1]+\textit{s}[m-1]*\phi_{\delta}[2]+\textit{d}[m]*\phi_{\delta}[3]$
		\State NewGrid[n+1] $\gets$ NewGrid[n] + NewGrid[n]-NewGrid[n-1]
		%\State $\textit{n} \gets \textit{n}+1$ 
		\EndIf
		\State OddGrid $\gets \text{copy the odd elements of}$ NewGrid
		\State EvenGrid $\gets \text{copy the even elements of}$ NewGrid
		\State $s \gets \textit{s}$ $+$ \Call{Filter}{d, 0, $\pi_1$, \text{NewGrid}}
		\State $d \gets \textit{d}$ $+$ \Call{Filter}{s, 1, $\pi_2$, \text{NewGrid}}
		\State $s \gets \textit{s}$ $+$ \Call{Filter}{d, 0, $\pi_3$, \text{NewGrid}}
		\State $d \gets \textit{d}$ $+$ \Call{Filter}{s, 1, $\pi_4$, \text{NewGrid}}	
		%	\State $u[1:m] \gets \textit{d}[1:m]$
		%	\State $u[(m+1):n] \gets \textit{s}[1:(m-1)]$
		
		%\For {i=0 To m-1}
		\State u[1:m] $\gets$ d[1:m]
		\State u[(m+1):n] $\gets$ s[1:(n-m)]	
		\State Grid[1:m] $\gets$ OddGrid
		\State Grid[(m+1):n] $\gets$ EvenGrid
		
		\State \emph{Scaling the wavelet coefficients}:
		\State $d \gets \delta[1:Q]*\varphi$
		\State $s \gets \delta[(Q+1):m]/\varphi$	
		%\State Grid[m+1+i] $\gets$ OldGrid[2*i+2]	
		%\EndFor
		\If {Level $>$ 1}
		%	\State $(new.u, new.Grid)\gets \text{ForwardTransform}\close{u[1:m], Grid[1:m], Level-1}$
		\State $\curly{\begin{matrix}
			\text{u[1:m]} \\
			\text{Grid[1:m]}
			\end{matrix}}$ $\gets$ \Call{ForwardTransform}{u[1:m], Grid[1:m], Level-1}
		%		\State $j \gets j-1$.
		%		\State $i \gets i-1$.
		%		\State \textbf{goto} \emph{loop}.
		%		\State \textbf{close};
		\EndIf
		%\State $i \gets i+\max(\textit{delta}_1(\textit{string}(i)),\textit{delta}_2(j))$.
		%\State \textbf{goto} \emph{top}.
		\State\Return{$\close{\text{u},\text{Grid}}$}
		\EndProcedure
	\end{algorithmic}
\end{algorithm}

%where each of sequences $\curly{\omega_{_{1,l}}},\curly{\omega_{_{2,l}}},\curly{\omega_{_{3,l}}},\curly{\omega_{_{4,l}}}$ are the interpolation weights (obtained by using \eqref{Interpolation:Weights} in Appendix \ref{Appendix:Lifting}) and $\pi_1,\pi_2,\pi_3,\pi_4$ are scalar constants described in \cite{schelkens2009jpeg}, which are referred to as the filter coefficients of the wavelet-based JPEG. 
\pagebreak
The Filter function in algorithm \ref{Algrorithm:ForwardTransform} defined as follows:
\begin{algorithm}[H]
		%\algsetup{linenosize=\tiny}
	\tiny
	\caption{The wavelet Forward and Inverse Filter}\label{Filter}
	\begin{algorithmic}[1]
		\Procedure{Filter}{Series, Even, $\pi$, Grid}
		\Require (i) Two real-number vectors: Series of size $n\times 1$ and Grid of size $2n\times 1$; (ii) two scalars: $\text{Even}\in\curly{0,1}$ and $\pi\in\mathbb{R}$.	
		\Ensure $\text{Output}\gets$ Filter, a real-number vector of size $n\times 1$ consisting of the predicted series;			
		\State n $\gets$ length of Series
		\State  O $\gets$ copy the odd elements of Grid 
		\State E $\gets$ copy the even elements of Grid 			
		%        \State OSeries $\gets\curly{\text{OddGrid}, \text{OddGrid}}$
		\If {Even $=$ 0}
		%	\State $(new.u, new.Grid)\gets \text{ForwardTransform}\close{u[1:m], Grid[1:m], Level-1}$
		\State Low $\gets$ O[1:n-1]
		\State High $\gets$ O[2:n]
		\State weights $\gets$ (High-E[1:(n-1)])/(High-Low)
		\State $\omega_l\gets$ [weights, 0.5]
		\State $\omega_h\gets$ $1-\omega_l$
		\State S $\gets$ [Series[1], Series]		
		\State Filter $\gets$ $2\pi$ $\omega_l$*S[1:n] + $2\pi$ $\omega_h$*S[2:(n+1)]
		%		\State $j \gets j-1$.
		%		\State $i \gets i-1$.
		%		\State \textbf{goto} \emph{loop}.
		%		\State \textbf{close};
		\Else
		\State Low $\gets$ E[1:(n-1)]
		\State High $\gets$ E[2:n]
		\State weights $\gets$ (High-O[2:n])/(High-Low)
		\State $\omega_l\gets$ [0.5, weights]
		\State $\omega_h\gets$ $1-\omega_l$
		\State S $\gets$ [Series, Series[n]]		
		\State Filter $\gets$ $2\pi$ $\omega_l$*S[1:n] + $2\pi$ $\omega_h$*S[2:(n+1)]		
		\EndIf
		%\State $i \gets i+\max(\textit{delta}_1(\textit{string}(i)),\textit{delta}_2(j))$.
		%\State \textbf{goto} \emph{top}.
		\State\Return{$\close{\text{Filter}}$}
		\EndProcedure
	\end{algorithmic}
\end{algorithm}

\subsection{The irregular inverse transform}
The irregular inverse transform is the procedure to reconstruct the vector $\B{u}$ in \eqref{Inverse:Operator}, as follows:
\begin{algorithm}[H]
	%\algsetup{linenosize=\tiny}
\tiny
	\caption{The wavelet-based JPEG Inverse Transform}\label{Algrorithm:InvTransform}	\label{Irregular:Inverse:Lifting}
	\begin{algorithmic}[1]
		\Procedure{InverseTransform}{$\delta$, Grid, Level}
		%\Procedure{Euclid}
		\Require (i) Two real-number vectors $\delta$ and Grid of size $n\times 1$; (ii) Level$\in\curly{1,...,\log_2(n)}$;	
		\Ensure $\text{Output}\gets$ a real-number vector $\textit{u}$ of size $n\times 1$;
		\State \emph{The JPEG filter coefficients}:
		\State $\pi_1 \gets -1.5861343420693648$; $\pi_2 \gets -0.0529801185718856$; 
		\State $\pi_3 \gets \color{white}+\color{black}0.8829110755411875$; $\pi_4 \gets \color{white}+\color{black}0.4435068520511142$
		\State $\varphi\color{white}_{1}\color{black} \gets \color{white}+\color{black}1.1496043988602418$			
		\State \emph{Start}:
		\State $\textit{n} \gets \text{length of } \delta$
		\State $\textit{m} \gets \text{ceiling of } \textit{n}/2^{(\text{Level}-1)}$
		\State $\textit{Q} \gets \text{ceiling of } \textit{m}/2$		
		
		\State \emph{Rescaling the wavelet coefficients}:			
		\State $d \gets \delta[1:Q]/\varphi$
		\State $s \gets \delta[(Q+1):m]*\varphi$
		%		\State OddGrid $\gets$ Grid[1:Q]
		%        \State EvenGrid $\gets$ Grid[(Q+1):m]
		\State OldGrid $\gets$ Grid
		\For {i=1 To m}
		\State index $\gets (2(i-1)+1)*(i \le Q)+2(i-Q)*(i > Q)$		
		\State Grid[index] $\gets$ OldGrid[i]
		\EndFor
		%		\State NewGrid[2:2:m] $\gets$ Grid[(Q+1):m]
		
		\State \emph{top}:
		%		\State NewGrid $\gets$ Grid		
		\State NewGrid $\gets$ Grid[1:m]
		\If {$\textit{m} \text{ is an odd number}$} 
		\State $\textit{s}[Q] \gets 0$   
		\State NewGrid[m+1] $\gets$ Grid[m] - Grid[m-1] + Grid[m]		
		
		%\State $\textit{n} \gets \textit{n}+1$ 
		\EndIf
		
		%		\State OddGrid $\gets \text{copy the odd elements of}$ NewGrid
		%		\State EvenGrid $\gets \text{copy the even elements of}$ NewGrid
		\State $d \gets \textit{d}$ $-$ \Call{Filter}{s, 1, $\pi_4$, \text{NewGrid}}
		\State $s \gets \textit{s}$ $-$ \Call{Filter}{d, 0, $\pi_3$, \text{NewGrid}}
		\State $d \gets \textit{d}$ $-$ \Call{Filter}{s, 1, $\pi_2$, \text{NewGrid}}
		\State $s \gets \textit{s}$ $-$ \Call{Filter}{d, 0, $\pi_1$, \text{NewGrid}}
		
		\For {i=1 To m}
		\State index $\gets (2(i-1)+1)*(i \le Q)+2(i-Q)*(i > Q)$		
		\State $\delta$[index] $\gets$ d[i]$*(i \le Q) + $s[i-Q]$*(i > Q)$
		\EndFor
		
		\State u $\gets \delta$		
		%	\State $u[1:m] \gets \textit{d}[1:m]$
		%	\State $u[(m+1):n] \gets \textit{s}[1:(m-1)]$
		
		%\State Grid[m+1+i] $\gets$ OldGrid[2*i+2]	
		%\EndFor
		\If {Level $>$ 1}
		%	\State $(new.u, new.Grid)\gets \text{ForwardTransform}\close{u[1:m], Grid[1:m], Level-1}$
		\State $\curly{\begin{matrix}
			\text{u} \\
			\text{Grid}
			\end{matrix}}$ $\gets$ \Call{InverseTransform}{$\delta$, Grid, Level-1}
		%		\State $j \gets j-1$.
		%		\State $i \gets i-1$.
		%		\State \textbf{goto} \emph{loop}.
		%		\State \textbf{close};
		\EndIf
		%\State $i \gets i+\max(\textit{delta}_1(\textit{string}(i)),\textit{delta}_2(j))$.
		%\State \textbf{goto} \emph{top}.
		\State\Return{$\close{\text{u},\text{Grid}}$}
		\EndProcedure
	\end{algorithmic}
\end{algorithm}

\subsection{The irregular transpose of the inverse transform}
The irregular transpose of the inverse transform in \eqref{Transposed:Inverse:Operator:Matrix} enables to obtain the vector $\B{\tilde{\B{u}}}$, as follows (see transposed-inverse filter function, TransInvFilter, in algorithm \ref{euclid}):
%\subsection{Algorithms}
\begin{algorithm}[H]
	%\algsetup{linenosize=\tiny}
\tiny	
	\caption{The wavelet-based JPEG Transposed-Inverse Transform}\label{Algrorithm:TransInvTransform}\label{Irregular:TransInvLifting}
	\begin{algorithmic}[1]
		\Procedure{TransInverseTransform}{\textit{u}, Grid, Level}
		\Require (i) Two real-number vectors \textit{u} and Grid of size $m\times 1$; (ii) Level$\in\curly{1,...,\log_2(m)}$
		\Ensure $\text{Output}\gets\Psi_{_I}^T\textit{u}$
		%\Input{Two real-number vectors u and Grid and a scalar integer Level}
		%\OUTPUT{$\gcd(a,b)$}		
		%\Procedure{Euclid}
		%		\State $\textit{n} \gets \text{length of }\textit{u}$
		%		\State $\textit{m} \gets \text{ceiling of } \close{\textit{n}/2}$
		%\State $J \gets count-1$
		%\State $\textit{n}_j \gets \text{length of }\textit{u}$
		\State \emph{The JPEG filter coefficients}:
		\State $\pi_1 \gets -1.5861343420693648$; $\pi_2 \gets -0.0529801185718856$; 
		\State $\pi_3 \gets \color{white}+\color{black}0.8829110755411875$; $\pi_4 \gets \color{white}+\color{black}0.4435068520511142$
		\State $\varphi\color{white}_{1}\color{black} \gets \color{white}+\color{black}1.1496043988602418$			
		\State \emph{Start}:
		\State $\text{Output}\gets\textit{u}$
		\State $\textit{m} \gets \text{length of } \textit{u}$ %\text{ceiling of } n/2^J$ 
		\State $\textit{Q} \gets \text{ceiling of } \close{\textit{m}/2}$		
		\State $d \gets \text{copy the odd elements of } \textit{u}$
		\State $s \gets \text{copy the even elements of } \textit{u}$
		
		\State OddGrid $\gets$ copy the odd elements of Grid 
		\State EvenGrid $\gets$ copy the even elements of Grid 		
		
		\State \emph{top}:
		\State NewGrid $\gets$ Grid
		
		\If {$\textit{m} \text{ is an odd number}$} 
		\State $\textit{s}[Q] \gets 0$ 
		\State NewGrid[m+1] $\gets$ NewGrid[m] + NewGrid[m]-NewGrid[m-1]
		%\State $\textit{n} \gets \textit{n}+1$ 
		\EndIf
		%	\State OddGrid $\gets \text{copy the odd elements of}$ NewGrid
		%	\State EvenGrid $\gets \text{copy the even elements of}$ NewGrid
		\State $d \gets \textit{d}$ $-$ \Call{TransInvFilter}{$s$, 0, $\pi_1$, \text{NewGrid}}		
		\State $s \gets \textit{s}$ $-$ \Call{TransInvFilter}{$d$, 1, $\pi_2$, \text{NewGrid}}	
		\State $d \gets \textit{d}$ $-$ \Call{TransInvFilter}{$s$, 0, $\pi_3$, \text{NewGrid}}			
		\State $s \gets \textit{s}$ $-$ \Call{TransInvFilter}{$d$, 1, $\pi_4$, \text{NewGrid}}		
		%	\State $u[1:m] \gets \textit{d}[1:m]$
		%	\State $u[(m+1):n] \gets \textit{s}[1:(m-1)]$
		\State \emph{Rescaling the wavelet coefficients}:
		\State $d \gets d[1:Q]/\varphi$
		\State $s \gets s[(Q+1):m]*\varphi$
		
		%\For {i=0 To m-1}
		\State $\text{Output}$[1:Q] $\gets$ d[1:Q]
		\State $\text{Output}$[(Q+1):m] $\gets$ s[1:(m-Q)]	
		\State Grid[1:Q] $\gets$ OddGrid
		\State Grid[(Q+1):m] $\gets$ EvenGrid	
		%\State Grid[m+1+i] $\gets$ OldGrid[2*i+2]	
		%\EndFor
		\If {Level $>$ 1}
		%	\State $(new.u, new.Grid)\gets \text{ForwardTransform}\close{u[1:m], Grid[1:m], Level-1}$
		\State 
		\text{Output[1:Q]} $\gets$ \Call{TransInverseTransform}{Output[1:Q], Grid[1:Q], Level-1}
		%		\State $j \gets j-1$.
		%		\State $i \gets i-1$.
		%		\State \textbf{goto} \emph{loop}.
		%		\State \textbf{close};
		\EndIf
		%\State $i \gets i+\max(\textit{delta}_1(\textit{string}(i)),\textit{delta}_2(j))$.
		%\State \textbf{goto} \emph{top}.
		\State\Return{$\close{\text{Output},\text{Grid}}$}
		\EndProcedure
	\end{algorithmic}
\end{algorithm}

In the ensuing section \color {black} we describe the estimation procedure of the parameter vector of interest, $\B{\beta}$, using the wavelet-based JPEG estimate of $\widehat{\mathcal{M}}_1(\cdot)$ obtained from \eqref{Optimization:Gradient:decent}.% \color{blue}  $\widehat{\mathcal{M}}_1(\cdot)$ obtained using our proposed algorithm which is wavelet-based JPEG. \color{black}% on irregular-spaced data.
%Next we extrapolate a new datum \eqref{Irregular:Lifting}

%The optimal weights to guarantee a zero detail coefficient of the extrapolated datum:
%\begin{IEEEeqnarray}{lCr}
\label{Optimal:Weights:Sample:Size}
\subsection{The instrument variable estimator: truncated sample}
Denote the truncated data by a sequence of observations $\curly{y_{1i},\B{x}_i,\B{w}_i,\B{z}_i}_{i=1}^n$, such that each observation is an independent realization of the conditional joint distribution function of the random variables $\curly{\mathrm{y}_1,\B{\mathrm{x}},\B{\mathrm{w}},\B{\mathrm{z}}}$ given that they are selected into the sample ($\mathrm{y_{2}}=1)$. The endogenous variable is denoted by $\mathrm{x}_1$ and is included in vector $\B{\mathrm{x}}$. There are two types of joint dependence between the covariate vector and the substantive equation's random disturbance. The first type is intrinsic in the model and is generated by a variation in $\mathrm{v}$ (the endogenous part of $\mathrm{x_1}$) leading to a comovement between  $\mathrm{x_1}$ and $\disturb_{1}$.  The second type is related to the sample selection and is generated by a variation in $\B{\mathrm{w}}$ leading to a comovement between the covariate vector $\B{\mathrm{x}}$ and $\disturb_{1}$. %The model's intrinsic endogeneity implies the joint dependence of $\mathrm{x_1}$ and $\disturb_{1}$ which is generated by a comovement between these two random variables due to a variation in $\mathrm{v}$ (the endogenous part of $\mathrm{x_1}$). 
This implies that there are two sources of endogeneity to be taken into consideration: the first source is related to the endogenous covariate, while the second source is due to the truncation environment of the data. 
%The second type of $\mathrm{x_1}$ depends on the conditional random variable $\disturb_{1}|\disturb_{2}\ge-\B{\mathrm{w}}^T\B{\mathrm{\gamma}}$ due to a comovement between these two random variable with respect to $\B{\mathrm{w}}^T\B{\mathrm{\gamma}}$

%\section{Kernel Density Estimators}\label{Section:Kernel}
%\vspace{1em}

%Let $X$ be a random variable with continuous distribution $F(x)$ with density $f(x)=\frac{1}{dx}F(x)$. If one is interested in estimating $f(x)$ from a random sample $\curly{X_1,...,X_n}$, it seems natural to obtain first an estimate for the distribution function:
%\begin{IEEEeqnarray}{lCr}

%\hat F(x)=n^{-1}\sum_{i=1}^{n}(X_i\le x),
%\end{IEEEeqnarray}

%and  then estimate the density $f(x)$ as the discrete derivative of $\hat F(x)$ for some small $h > 0$, which is defined as:

%\begin{IEEEeqnarray}{lCr}

%\hat f(x) = \frac{\hat F(x+h)-\hat F(x-h)}{2h}
%\end{IEEEeqnarray}

%One can write the derivative above as:

%\begin{IEEEeqnarray}{lCr}
\label{fhat}
%\hat f(x) = \frac{1}{2nh}\sum_{i=1}^{n}I(x-h\le X_i\le x+h) = \frac{1}{2nh}\sum_{i=1}^{n}I\close{\frac{\abs{X_i-x}}{h}\le 1}
%\end{IEEEeqnarray}

%
%The estimator $\hat f(x)$ represents the percentage of observations which are close to the point $x$.
%If many observations are near $x$, then $\hat f(x)$ is large and conversely. The bandwidth $h$ controls the degree of smoothing.

%In order to simplify \eqref{fhat}, we define the following function:
%\begin{IEEEeqnarray}{lCr}
\label{UniformKernel}
%k(u) = \begin{cases} \frac{1}{2}, & \mbox{if} \abs{u}\le 1\\ 0 , & \mbox{if} \abs{u}>1

%\end{cases}
%\end{IEEEeqnarray}

%where $k(u)$ is the uniform density function on $[-1,1]$. 

%We obtain the kernel density estimator from \eqref{fhat} and \eqref{UniformKernel}:
%\begin{IEEEeqnarray}{lCr}
\label{Density:Equation}
%\hat f(x) =  \frac{1}{nh}\sum_{i=1}^{n}k\close{\frac{x-X_i}{h}},
%\end{IEEEeqnarray}

%where $k(u)$ is referred to as a kernel function.\footnote{A kernel function $k(u):\mathbb{R}\to \mathbb{R}$ is any function which satisfies $\int_{\infty}^{\infty}k(u) du=1$.} 

%Next using \eqref{Density:Equation} define the standardize kernel function:
%\begin{IEEEeqnarray}{lCr}
\label{Standarized:kernel}
%k_h(x)=\frac{1}{h}k\close{\frac{x}{h}}
%\end{IEEEeqnarray}

%which will be employed to arrive an alternative IV estimator.
Next we discuss the two-step estimation procedure to be employed for the correction of both endogeneity and truncation bias propagated by truncation.
\subsection{The estimation procedure}
%\color{red} Tell an introduction about the two stages and why we are doing that \color{black}
In this section we introduce a two-step estimation procedure to eliminate the two sources of bias discussed. To eliminate the endogeneity bias term we adapt a similar approach to the two step procedure in \cite{zhou2016estimation} for a partially linear single index model estimation, in which the first stage is a regression of the endogenous covariate on all the exogenous covariates and the instrumental variable. In the second stage, the endogenous covariate is substituted with the fitted values obtained from the first stage. However, the estimation approach in \cite{zhou2016estimation} cannot be implemented in a truncated environment, because it treats the first stage regression as a linear population regression (as if the entire covariates distribution function is observed). We alleviate this by modeling both the first as well as the second stage equations as endogenously truncated equations. In order to eliminate the endogenous truncation bias, we control for this source of bias by including the truncation bias term as an additional covariate in the substantive equations, as depicted in \eqref{Truncated:Substatnive:System}. Thus, the partial linearity is applied to both the first as well as the second stage equations.  

In the first stage, we regress the endogenous covariate on the instrumental and exogenous variables, by minimizing the partially linear index model:
\begin{IEEEeqnarray}{lCr}
	\label{first:stage}
	\scaleobj{0.7}{(\B{\widehat{\delta},\widehat{\theta_{1f}}})=\underset{\B{(\delta,\theta_{1f})}\in\Theta\times \Delta_K}{\arg\min}\frac{1}{n}\sum_{i=1}^{n}\close{x_{1i}-\braket{\B{x}_{-1_i}^T,\B{z}_i^T}\B{\delta}-\widehat{\mathcal{M}}_2(\B{w_i};\B{\theta_{1f}})}^2}
\end{IEEEeqnarray}

In the second stage, the endogenous variable is replaced by its predicted value obtained from the first stage in \eqref{first:stage}, and we minimize the following function:

\begin{IEEEeqnarray}{lCr}
	\label{second:stage}
	\scaleobj{0.7}{(\B{\widehat{\beta},\widehat{\theta_{2f}}})=\underset{\B{(\beta,\theta_{2f})}\in\Theta\times \Delta_K}{\arg\min}\frac{1}{n}\sum_{i=1}^{n}\biggl(y_{1i}-\braket{\widehat{x_{1i}},\B{x}_{-1_i}^T}\B{\beta}} \scaleobj{0.7}{-\widehat{\mathcal{M}}_1(\B{w_i};\B{\theta_{2f}})\biggr)^2}
\end{IEEEeqnarray}

%Next we implement endogenously truncated Hausman exogeneity test.
%\color{blue} Say something here \color{black}

As can be seen in \eqref{second:stage} the two sources of endogeneity bias we deal with are: (i) the bias propagated by the endogenous covariate is alleviated by utilizing the covariate set $\braket{\widehat{x_{1i}},\B{x}_{-1_i}^T}$ consisting entirely of exogenous covariates and (ii) the bias propagated by the endogenous truncation is alleviated by controlling for the selection bias term $\widehat{\mathcal{M}}_1(\cdot)$.

Next we present Monte Carlo simulation to examine our semiparametric IV estimator's performance in a truncated environment.
\section{Simulation}\label{Section:Simulation}
In this section, we generate multiple random data sets to be used for the examination of our model's performance, using different sample sizes. %, in the presence of an endogenous variable in the substantive equation given a truncated data set.

First, we discuss the procedure for the data generation process (DGP).
\subsection{Data generation process}\label{Sec:DGP}
Denote the sample size by $N\in\bigl\{500$, $2000$, $3000$, $5000$, $8000$, $10000\bigr\}$. In order to not restrict the data generation process to the family of symmetric unimodal distribution functions, a mixture of distribution functions is utilized to generate each of the selection model's disturbances that are jointly dependent (as will be discussed in section \ref{Joint:Disturb} to follow). In order to verify that our proposed model performs well under different data generating processes (DGP), we construct a data set consisting of 2,000,000 distribution functions,\footnote{The estimates obtained given the various data distribution functions will be supplied upon request.} practically generating 100 millions realizations which are not i.i.d. By construction, each observation is randomly drawn from a unique mixture of distribution functions. %The motivation is to examine our model's performance in cases of a non-standard distribution functions for the disturbances, such as a mixture of distribution functions.  %the functions as the marginal density, in order to generate a non-normal distributed disturbances to be estimated.
% and the predicted class (since the classes are latent). 
%The participation shares depicted in \eqref{Participation:Share} are generated using arbitrarily chosen number in order to construct the conditional distribution function of the individual characteristics given the group's characteristics. 

%\vspace{0.5em}
\subsubsection{The disturbances' joint distribution function}\label{Joint:Disturb}
Each triple of disturbances $\curly{\disturb_{1i},\disturb_{2i},\mathrm{v}_i}$ is randomly and independently drawn from $F_{\disturb_1,\disturb_2,\mathrm{v}}$, which is the substantive and participation equations' disturbances joint distribution function. The aforementioned joint density function consists of two components: a Copula function,\footnote{Any continuous joint distribution function can be characterized by a set of marginal distribution functions and a joint distribution function determining the dependence structure which is referred to as a Copula function (Sklar's Theorem \cite{sklar1959fonctions}). } which characterizes the disturbances' dependence structure, and three marginal distribution functions $F_{\disturb_1}$, $F_{\disturb_2}$ and $F_{\mathrm{v}}$. In order to verify our model's performance in the presence of random disturbances' distribution functions that are not restricted to the family of symmetric and unimodal distribution functions, each one of the sample selection model's disturbances  $\disturb_1$ and $\disturb_2$ is marginally-distributed according to a mixture of three different distribution functions: (i) a normal distribution function with expectation and standard deviation parameters $(\mu, \sigma_a)$ denoted by $\mathcal{N}(\mu,\sigma_{a}^2)$; (ii) a  normal distribution function with expectation and standard deviation parameters $(-\mu, \sigma_b)$ denoted by $\mathcal{N}(-\mu,\sigma_{b}^2)$; (iii) a gamma distribution function with scale and shape parameters $(\mu\varphi,\varphi)$ denoted by $\Gamma_{\scaleto{\mathcal{\text{Gamma}}}{4pt}}\close{\mu\varphi, \varphi}$\footnote{The scale and shape parameters imply that the expectation and standard deviation parameters are $(\mu ,\sqrt{\mu/\varphi})$, respectively.}. This mixture distribution function is defined as: 
\begin{IEEEeqnarray}{lCr}
	\label{Marginal:Random:Dist}
	\scaleobj{0.7}{\begin{cases}
			\mathrm{v}\sim \mathcal{N}(0,\sigma_{\mathrm{v}}^2)\\
			\mathrm{\disturb_{j}}\sim 0.4\mathcal{N}(\mu,\sigma_{a}^2)+0.5\mathcal{N}(-\mu,\sigma_{b}^2)+0.1\Gamma_{\scaleto{\mathcal{\text{Gamma}}}{4pt}}\close{\mu\varphi, \varphi}, & j=1,2.
	\end{cases}}
\end{IEEEeqnarray}

where $\mathbb{E}\braket{\mathrm{\disturb_{j}}}=0$ and   $\mathbb{E}\braket{\mathrm{v}}=0$ .

The parameters set $(\mu,\sigma_a,\sigma_b,\varphi,\sigma_{\mathrm{v}})=(4,2.5,1.5,2,1)$ is arbitrarily chosen. Due to its simplicity, the Clayton Copula (as will be discussed in section \ref{Sec:Archimedean:Copula} to follow) with a degree of dependence parameter is set to equal $1$, assuring a mild correlation between the disturbances, is used for controlling the dependence structure. Choosing a mild correlation, is important in order to be conservative by examining the potential bias in the parameter estimates under conditions which are not extreme.

%\color{blue} Explain how this is connected to Archimedean Copula \color{black}
Next we employ a function characterizing the dependence properties of the Copula \cite{mcneil2009multivariate}, referred to as \textit{a generator function} to construct the joint dependence of the random disturbances in \eqref{Marginal:Random:Dist}.

\subsubsection{Archimedean Copula function}\label{Sec:Archimedean:Copula}
An Archimedean Copula is a Copula characterized by a non-increasing, continuous generator function $\psi$: $[0,\infty] \to [0, 1]$, which satisfies $\psi(0) = 1$, $\psi(\infty) = 0$ and is strictly decreasing on $[0, \inf\curly{t : \psi(t) = 0}]$.
%According to Sklar's Theorem \cite{sklar1959fonctions}, any continuous joint distribution function can be characterized by a set of marginal distribution functions and a joint distribution function determining the dependence structure which is referred to as a Copula function.
In particular, we are interested in the $d$ dimensional Archemdean Copula family ($3$ in the present case\footnote{$d=3$ representing the three-dimensional vector of random disturbances $(\mathrm{v_i},\disturb_{1i},\disturb_{2i})$.}) which has the simple algebraic form \cite{mcneil2009multivariate}:\footnote{Knowing the distribution corresponding to a generator $\psi$, \cite{marshall1988families} presented a sampling algorithm for exchangeable Archimedean copulas which does not require the knowledge of the copula density. This algorithm is therefore applicable to large dimensions \cite{hofert2008sampling}.}
\begin{IEEEeqnarray}{lCr}
	\scaleobj{0.7}{\mathcal{C}(u_1,..,u_d)=\psi(\psi^{-1}(u_1),...,\psi^{-1}(u_d)), \hspace{1em} (u_1,...,u_d)\in [0,1]^d}
\end{IEEEeqnarray}

where $\psi$ is a specific function known as the generator of $\mathcal{C}$. To generate the disturbances, the Clayton Copula's generator $\psi(t)=(1+t)^{-1/\theta}$ is chosen. %The reason for choosing an Archimedean Copula is due to a technical convenience.

The covariates vector of $i$'th observation is a realization of the random variables $\braket{\mathrm{z},\mathrm{x_2},\mathrm{w_1},\mathrm{w_2}}$ which are jointly distributed with their corresponding marginal distribution functions $\curly{\mathlcal{F}_{\mathrm{z}}^i, \mathlcal{F}_{\mathrm{x_2}}^i, \mathlcal{F}_{\mathrm{w_1}}^i,\mathlcal{F}_{\mathrm{w_2}}^i}$. The dependence structure is modeled by utilizing a Gaussian Copula which is a convenient way to generate high dimensional data. By construction, each datum is generated by utilizing a different sequence of marginal distribution functions (constructed as a finite mixture drawn from  a menu of $2,000,000$ continuous distribution functions). These random variables expectation vector $\B{\mathrm{\mu}}=[0,0,0,0]^T$ and a covariance matrix  $\B{\mathrm{\Sigma_{4\times 4}}}$. The arbitrarily chosen covariance matrix is:
\begin{IEEEeqnarray}{lCr}
	\scaleobj{0.7}{\B{\mathrm{\Sigma_{4\times 4}}}=\braket{\begin{matrix}
				\sigma_{\mathrm{z}}^2 & \sigma_{\mathrm{z,x_2}} & \sigma_{\mathrm{z,w_1}} & \sigma_{\mathrm{z,w_2}}\\
				\sigma_{\mathrm{z,x_2}} & \sigma_{\mathrm{x_2}}^2 & \sigma_{\mathrm{x_2,w_1}} & \sigma_{\mathrm{x_2,w_2}}\\
				\sigma_{\mathrm{z,w_1}} & \sigma_{\mathrm{x_2,w_1}} & \sigma_{\mathrm{w_1}}^2 &  \sigma_{\mathrm{w_1,w_2}}\\
				\sigma_{\mathrm{z,w_2}} & \sigma_{\mathrm{x_2,w_2}} & \sigma_{\mathrm{w_1,w_2}} & \sigma_{\mathrm{w_2}}^2
		\end{matrix}}=\braket{\begin{matrix}
				1 & 0.4 & 0.8 & -0.6\\ 0.4 & 1.264 & 0.36 & -0.48\\ 0.8 & 0.36  & 2 & -0.4\\ -0.6 & -0.48 & -0.4 & 2
	\end{matrix}}}\nonumber
\end{IEEEeqnarray}
We generate the data $y_{1i}, y_{2i}, x_{1i}$ according to the following data generation process (DGP) \cite{escanciano2017simple}:
\begin{IEEEeqnarray}{lCr}
	\text{DGP}1:\begin{cases}
		y_{1i}^*= \alpha_1+\beta_1 x_{1i}+\beta_2 x_{2i}+\mathrm{\disturb_{1i}}\\
		y_{2i}^*= \alpha_2+\gamma_1 w_{1i}+\gamma_2 w_{2i}+\mathrm{\disturb_{2i}}\\
		x_{1i}^* = \delta_1 z_i+\delta_2 x_{2i}+\mathrm{v_{i}}\\
	\end{cases}
\end{IEEEeqnarray}

where each $i$ element in the sequence $\curly{x_{2i},z_i,w_{1i},w_{2i}}_{i=1}^N$ is an independent realization of the random variables $(\mathrm{x_2,z,w_1,w_2})$. We choose the parameter setting $[\alpha_1,\alpha_2,\beta_1,\beta_2,\delta_1,\delta_2,\gamma_1,\gamma_2]=[2, 0.5, 1, 1.25, 0.5, 1 ,2, -1]$.
%\in\mathbb{R}^4

%The $d$ dimensional Copula can be constructed as follows \cite{mcneil2009multivariate}

The truncated data set is characterized by the following equations:
\begin{IEEEeqnarray}{lCr}
	\label{Truncated:Equation}
	\braket{\begin{matrix}
			y_{1i}\\
			x_{1i}\\
	\end{matrix}}= \begin{cases}
		\braket{\begin{matrix}
				\alpha_1+\beta_1 x_{1i}+\beta_2 x_{2i}+\mathrm{\disturb_{1i}} \\
				\delta_1 z_i+\delta_2 x_{2i}+\mathrm{v_{i}}\\
		\end{matrix}} & \text{if } y_{2i}^*\ge 0\\
		\text{Unobserved} & \text{if } y_{2i}^*< 0\\
	\end{cases},
\end{IEEEeqnarray}

where $x_{1i}$ is an endogenous variable included in vector $\B{x}_i\in\mathbb{R}^p$, in which all the elements (except for $x_i$) are exogenous variables and $\B{\beta}\in\mathbb{R}^p$ is a covariates vector. The substantive equation's random disturbance is denoted by $\disturb_{1i}$.

\subsection{Simulations result}
We have randomly generated for each sample size $N\in\curly{500,2000,3000,5000,8000,10000}$, a total of $10,000$ data sets using the data generation process elaborated on in \ref{Sec:DGP}. For a given number of observations $N$, different models are estimated: (i) an $OLS$ estimator utilizing a sample consisting of random realizations from the complete distribution function, without correcting for the endogeneity of $x_{1i}$ covariate; (ii) a conventional IV estimator, correcting for the endogeneity of $x_{1i}$ covariate  using the aforementioned entire distribution function; (iii) both $OLS$ as well as a conventional IV estimators %correcting for the endogeneity of $x_{1i}$ covariate 
are applied to a truncated portion of the data distribution function consisting of participants only (without correcting for the self-selection bias); (iv) truncated sample model's estimates using the developed wavelet-based JPEG IV estimator, correcting for both truncation as well as endogeneity biases.\footnote{For sake of brevity, we delegate results of the first stage to Appendix \ref{Appendix:First:Stage}.}

\newcommand\mse[2]{\close{\frac{1}{\Omega}\sum_{i=1}^{\Omega}\close{{{#1}-{#2}}}^2}^{1/2}}
\newcommand\rmse[2]{\close{\frac{1}{\Omega}\sum_{i=1}^{\Omega}\close{\frac{{{#1}-{#2}}}{{#2}}}^2}^{1/2}}
\newcommand\BBigger[1]{\mathlarger{\mathlarger{\mathlarger{#1}}}}	
Table \ref{tabB} presents summary statistics of estimates for models (i) and (ii), while Table \ref{tabC} presents summary statistics of estimates for models (iii) and (iv). In Table \ref{tabD}, different convergence measures of these estimates are presented.

\linespread{0.1}
\begin{table}%[H]
	
	\caption{\label{tabB}Monte Carlo Simulation - Non-truncated (complete) data set with (without) IV correction}
	
	\begin{center}\scalebox{1}[0.7]{\begin{adjustbox}{width=25em}
				\begin{tabular}{ccccccccccc|}
					\hline
					\multicolumn{3}{|l|}{\multirow{3}{*}{True Parameter$^{\mathrm{a}}$}} & \multicolumn{2}{l|}{\multirow{3}{*}{Estimate}} & \multicolumn{6}{c|}{Model Setup}                                                                                                                                       \\ \cline{6-11} 
					\multicolumn{3}{|l|}{}                              & \multicolumn{2}{l|}{}                          & \multicolumn{6}{c|}{Sample size}                                                                                                                                       \\ \cline{6-11} 
					\multicolumn{3}{|l|}{}                              & \multicolumn{2}{l|}{}                          & \multicolumn{1}{c|}{500} & \multicolumn{1}{c|}{2000} & \multicolumn{1}{c|}{3000} & \multicolumn{1}{c|}{5000} & \multicolumn{1}{c|}{8000} & \multicolumn{1}{c|}{10000} \\ \hline
					\hline\multicolumn{11}{|c|}{$\textbf{Full}$ sample $\textbf{OLS}$ estimates (\textbf{without IV correction})}\\\cline{1-11} 		
					\multicolumn{3}{|c|}{\multirow{3}{*}{$\beta_1=1$}} & \multicolumn{2}{c|}{Mean} & 2.6848 & 2.6798 & 2.6807 & 2.6810 & 2.6809 & 2.6805 \\\cline{4-11} \multicolumn{3}{|c|}{} & \multicolumn{2}{c|}{Median}   & 2.6857 & 2.6815 & 2.6811 & 2.6812 & 2.6819 & 2.6812 \\\cline{4-11} \multicolumn{3}{|c|}{} & \multicolumn{2}{c|}{Std}   & 0.1602 & 0.1138 & 0.0858 & 0.0605 & 0.0534 & 0.0507 \\\hline

					\multicolumn{3}{|c|}{\multirow{3}{*}{$\beta_2=1.25$}} & \multicolumn{2}{c|}{Mean} & -0.7009 & -0.6953 & -0.6988 & -0.6983 & -0.6963 & -0.6976 \\\cline{4-11} \multicolumn{3}{|c|}{} & \multicolumn{2}{c|}{Median}   & -0.7085 & -0.6983 & -0.6967 & -0.6994 & -0.6968 & -0.6976  \\\cline{4-11} \multicolumn{3}{|c|}{} & \multicolumn{2}{c|}{Std}   & 0.2420 & 0.1743 & 0.1275 & 0.0838 & 0.0710 & 0.0652 \\\hline

					\hline\multicolumn{11}{|c|}{$\textbf{Full}$ sample $\textbf{OLS}$ estimates (\textbf{with IV correction})}\\\cline{1-11} 	
					\multicolumn{3}{|c|}{\multirow{3}{*}{$\beta_1=1$}} & \multicolumn{2}{c|}{Mean} & 0.9826 & 0.9995 & 1.0025 & 1.0020 & 0.9986 & 0.9991 \\\cline{4-11} \multicolumn{3}{|c|}{} & \multicolumn{2}{c|}{Median}  & 1.0050 & 1.0057 & 1.0083 & 0.9996 & 0.9997 & 0.9998 \\\cline{4-11} \multicolumn{3}{|c|}{} & \multicolumn{2}{c|}{Std} & 0.4397 & 0.3072 & 0.2174 & 0.1393 & 0.1116 & 0.1006\\\hline
					
					\multicolumn{3}{|c|}{\multirow{3}{*}{$\beta_2=1.25$}} & \multicolumn{2}{c|}{Mean}& 1.2687 & 1.2493 & 1.2457 & 1.2469 & 1.2510 & 1.2517 \\\cline{4-11} \multicolumn{3}{|c|}{} & \multicolumn{2}{c|}{Median}   & 1.2365 & 1.2382 & 1.2438 & 1.2466 & 1.2497 & 1.2506 \\\cline{4-11} \multicolumn{3}{|c|}{} & \multicolumn{2}{c|}{Std}  & 0.5406 & 0.3771 & 0.2655 & 0.1692 & 0.1366 & 0.1211\\\hline
					
	\end{tabular}\end{adjustbox} }\end{center}  
	\scriptsize\begin{minipage}{8.75cm}%
		$\myspace$ \\ \textbf{Note:}   $^{\mathrm{a}}$ The parameters that are used in the data generation process.\\  We estimate by ordinary least squares (OLS) method the parameters for the full sample and truncated sample without correction for the selectivity bias. % and compute the standard deviation in every random sample consisting of N observations. 
		Then, we calculate for these estimates the mean, median and standard deviation (Std.) over all data sets. The standard deviations are obtained using the estimates from the Monte Carlo simulations.
	\end{minipage}		
	\footnotesize
	\renewcommand{\baselineskip}{11pt}		
\end{table} 
Entries in Table \ref{tabB} indicate that regardless of sample size, the means of the $OLS$ estimates are biased, such that e.g., for a sample size of $3000$ observations $\beta_1=2.6807$ and $\beta_2=-0.6988$, while the mean of the full sample IV's estimates are $\beta_1=1.0025$ and $\beta_2=1.2457$ for the same sample size. For a sample size of $500$ observations, the standard deviation obtained for $\beta_1$ (the endogenous covariate's coefficient) using the IV estimator is $2.73$ times larger than in the $OLS$ estimator and decreases from $0.4397$ to $0.1006$, when the sample size increases from $500$ to $10,000$ observations. 

In Table \ref{tabC} to follow, the estimates obtained by using the truncated data are presented. It is evident that the mean of the truncated sample $OLS$ estimates are $\beta_1=2.26$ and $\beta_2=-0.3084$ for a sample size of $500$ observations, where as the true parameter values are $\beta_1=1$ and $\beta_2=1.25$, respectively. This is a huge bias which is hardly improved as the sample size increases. Further, applying a conventional IV produces estimates which still represent a huge bias, particularly $\beta_1=0.1084$ and $\beta_2=2.0544$ for the same sample size (500 observations). %These results do not improve much when the sample size increases. The proposed JPEG IV introduced completely alleviates the aforementioned biases. Entries in Table \ref{tabC} attest to high precision of estimates generated by the newly introduced JPEG IV, as reflected by their relative small dispersions compared to estimates generated by the conventional IV.
\linespread{0.1}
\begin{table}[H]
	
	\caption{\label{tabC}Monte Carlo Simulation - Truncated data set with (without) truncation bias correction}
	
	\begin{center}\scalebox{1}[0.7]{\begin{adjustbox}{width=25em}\begin{tabular}{ccccccccccc|}
					\hline
					\multicolumn{3}{|l|}{\multirow{3}{*}{True Parameter$^{\mathrm{a}}$}} & \multicolumn{2}{l|}{\multirow{3}{*}{Estimate}} & \multicolumn{6}{c|}{Model Setup}                                                                                                                                       \\ \cline{6-11} 
					\multicolumn{3}{|l|}{}                              & \multicolumn{2}{l|}{}                          & \multicolumn{6}{c|}{Sample size}                                                                                                                                       \\ \cline{6-11} 
					\multicolumn{3}{|l|}{}                              & \multicolumn{2}{l|}{}                          & \multicolumn{1}{c|}{500} & \multicolumn{1}{c|}{2000} & \multicolumn{1}{c|}{3000} & \multicolumn{1}{c|}{5000} & \multicolumn{1}{c|}{8000} & \multicolumn{1}{c|}{10000} \\ \hline
					\hline\multicolumn{11}{|c|}{$\textbf{Truncated}$ sample $\textbf{OLS}$ estimates \tiny(without truncation bias correction) \normalsize}\\\cline{1-11} 	
					\multicolumn{3}{|c|}{\multirow{3}{*}{$\beta_1=1$}} & \multicolumn{2}{c|}{Mean} & 2.2600 & 2.2493 & 2.2504 & 2.2485 & 2.2488 & 2.2492 \\\cline{4-11} \multicolumn{3}{|c|}{} & \multicolumn{2}{c|}{Median}  & 2.2711 & 2.2514 & 2.2534 & 2.2510 & 2.2484 & 2.2493\\\cline{4-11} \multicolumn{3}{|c|}{} & \multicolumn{2}{c|}{Std}  & 0.2676 & 0.1878 & 0.1362 & 0.0865 & 0.0723 & 0.0675\\\hline		 			
					\multicolumn{3}{|c|}{\multirow{3}{*}{$\beta_2=1.25$}} & \multicolumn{2}{c|}{Mean} & -0.3084 & -0.2926 & -0.2974 & -0.2955 & -0.2939 & -0.2940 \\\cline{4-11} \multicolumn{3}{|c|}{} & \multicolumn{2}{c|}{Median}  & -0.3170 & -0.2884 & -0.3048 & -0.2974 & -0.2944 & -0.2942 \\\cline{4-11} \multicolumn{3}{|c|}{} & \multicolumn{2}{c|}{Std}  & 0.3693 & 0.2647 & 0.1909 & 0.1162 & 0.0940 & 0.0851 \\\hline

					\hline\multicolumn{11}{|c|}{$\textbf{Truncated}$ sample $\textbf{OLS}$ estimates -  IV correction \tiny(without truncation bias correction) \normalsize}\\\cline{1-11} 	
					\multicolumn{3}{|c|}{\multirow{3}{*}{$\beta_1=1$}} & \multicolumn{2}{c|}{Mean} & 0.1084 & 0.1805 & 0.1942 & 0.1910 & 0.1850 & 0.1847 \\\cline{4-11} \multicolumn{3}{|c|}{} & \multicolumn{2}{c|}{Median}  & 0.1410 & 0.2095 & 0.2074 & 0.2003 & 0.1901 & 0.1876\\\cline{4-11} \multicolumn{3}{|c|}{} & \multicolumn{2}{c|}{Std}  & 0.7595 & 0.5288 & 0.3588 & 0.2235 & 0.1777 & 0.1634\\\hline		 			
					\multicolumn{3}{|c|}{\multirow{3}{*}{$\beta_2=1.25$}} & \multicolumn{2}{c|}{Mean} & 2.0544 & 2.0148 & 1.9973 & 2.0018 & 2.0083 & 2.0088 \\\cline{4-11} \multicolumn{3}{|c|}{} & \multicolumn{2}{c|}{Median}  & 2.0022 & 1.9798 & 1.9833 & 1.9968 & 2.0058 & 2.0082 \\\cline{4-11} \multicolumn{3}{|c|}{} & \multicolumn{2}{c|}{Std}  & 0.8910 & 0.6219 & 0.4212 & 0.2601 & 0.2055 & 0.1878 \\\hline
					
					%			\hline\multicolumn{11}{|c|}{$\textbf{Truncated}$ sample \textbf{JPEG \textit{IV}} estimates (First stage)}\\\cline{1-11} 
					%\multicolumn{3}{|c|}{\multirow{3}{*}{$\delta_1=0.5$}} & \multicolumn{2}{c|}{Mean} & 0.4944 & 0.4975 & 0.4973 & 0.4996 & 0.5000 & 0.4999 \\\cline{4-11} \multicolumn{3}{|c|}{} & \multicolumn{2}{c|}{Median}  & 0.4951 & 0.4963 & 0.4971 & 0.4999 & 0.4998 & 0.5000 \\\cline{4-11} \multicolumn{3}{|c|}{} & \multicolumn{2}{c|}{Std}  & 0.1236 & 0.0675 & 0.0514 & 0.0461 & 0.0352 & 0.0298 \\\hline		
					
					%\multicolumn{3}{|c|}{\multirow{3}{*}{$\delta_2=1$}} & \multicolumn{2}{c|}{Mean} & 0.9976 & 0.9983 & 0.9994 & 1.0001 & 1.0005 & 0.9999\\\cline{4-11} \multicolumn{3}{|c|}{} & \multicolumn{2}{c|}{Median}  &0.9977 & 0.9989 & 0.9989 & 1.0001 & 1.0006 & 1.0003\\\cline{4-11} \multicolumn{3}{|c|}{} & \multicolumn{2}{c|}{Std}  & 0.0902 & 0.0444 & 0.0360 & 0.0279 & 0.0221 & 0.0197\\\hline
					
					\hline\multicolumn{11}{|c|}{$\textbf{Truncated}$ sample \textbf{JPEG \textit{IV}} estimates}\\\cline{1-11} 
					\multicolumn{3}{|c|}{\multirow{3}{*}{$\beta_1=1$}} & \multicolumn{2}{c|}{Mean} & 0.9455 & 0.9875 & 1.0058 & 1.0010 & 1.0004 & 1.0000 \\\cline{4-11} \multicolumn{3}{|c|}{} & \multicolumn{2}{c|}{Median}  & 0.9717 & 1.0085 & 1.0079 & 1.0047 & 1.0013 & 1.0009\\\cline{4-11} \multicolumn{3}{|c|}{} & \multicolumn{2}{c|}{Std}  &  0.7325 & 0.3641 & 0.2972 & 0.2178 & 0.1732 & 0.1577\\\hline
					
					\multicolumn{3}{|c|}{\multirow{3}{*}{$\beta_2=1.25$}} & \multicolumn{2}{c|}{Mean} & 1.3174 & 1.2830 & 1.2617 & 1.2465 & 1.2506 & 1.2499\\\cline{4-11} \multicolumn{3}{|c|}{} & \multicolumn{2}{c|}{Median}  & 1.2945 & 1.2623 & 1.2579 & 1.2456 & 1.2518 & 1.2469\\\cline{4-11} \multicolumn{3}{|c|}{} & \multicolumn{2}{c|}{Std}  & 0.8145 & 0.4081 & 0.3324 & 0.2433 & 0.1933 & 0.1754\\\hline
					
					%${\gamma_2=-1}$ & -2.0575 & -1.5680 & -1.3809 & -1.1803 & -1.0619 & -1.0347\\{$\text{median}$}  & -1.8598 & -1.4542 & -1.2444 & -1.0882 & -1.0181 & -1.0118 \\{$\text{std}$}  & (1.4035) & (0.9544) & (0.9047) & (0.5962) & (0.3876) & (0.3190) \\\hline
					
	\end{tabular}\end{adjustbox} }\end{center}
	\scriptsize\begin{minipage}{8.8cm}%
		$\myspace$ \\ \textbf{Note:}    $^{\mathrm{a}}$ The parameters that are used in the data generation process.\\ We estimate the parameters for the truncated sample without correction for the selectivity bias by employing the $OLS$ and the (conventional) IV methods. Additionally, we estimate the parameters using the same truncated sample by employing our JPEG IV estimator, correcting for both truncation as well as endogeneity bias. For brevity, we introduce here only the second stage results and delegate the first stage results to Table \ref{tabE} Appendix \ref{Appendix:First:Stage}. % In both estimation procedures (IV and JPEG IV), the standard deviations are obtained using the estimates from the Monte Carlo simulations. %computed in every random sample consisting of N observations. 
		In each of the estimation procedures ($OLS$, IV and JPEG IV), we calculate for these estimates the mean, median and standard deviation (Std.) over all data sets. %The standard deviations are obtained using the estimates from the Monte Carlo simulations.
	\end{minipage}		
	\footnotesize
	\renewcommand{\baselineskip}{11pt}			
\end{table}

Entries in Table \ref{tabC} indicate that regardless of sample size, the means of the truncated sample IV's estimates are biased (ranges from one-tenth to one-fifth of the estimate that would have emerged by employing the conventional IV method in the absence of truncation).\footnote{For sake of brevity, we have omitted the estimates of the nuisance parameters which can be furnished upon request as well as the parameter estimates of the first stage, which are delegated to Table \ref{tabE} Appendix \ref{Appendix:First:Stage}.} 	
Note that estimates' accuracy hardly improved as sample size increases. This is due to the presence of two sources of bias. The mean estimate of $\beta_1$ (the endogenous covariate's parameter) which is obtained from implementing our proposed methodology, basically mimics the results obtained using a random sample from the entire data distribution function for sample sizes, above $2,000$ observations. The standard deviations of this estimate for sample sizes of $500$ and $10,000$ observations are $0.7325$ and $0.1577$, respectively. For a sample size of $5,000$ observations (or above), the mean estimate of $\beta_2$ (the exogenous covariate's parameter) approximates the estimate obtained by employing the conventional IV, using a random sample from the entire data distribution function. However, the estimate of $\beta_2$ obtained by employing the conventional IV in a truncated sample is biased even for $10,000$ observations. %Further, entries in Table \ref{tabC} attest to high precision of estimates generated by the newly introduced JPEG IV, as reflected by their relative small dispersions compared to estimates generated by the conventional IV. 

We conduct sensitivity test to measure the influence of an increase in number of observations on the accuracy of the truncated sample's parameter estimates. 

The first accuracy measure we use is the standardized root mean square error, $RMSE_j$, measuring the bias in the truncated regression estimate relative to the true parameter value that would have been obtained in an non truncated distribution, defined as: 
\begin{IEEEeqnarray}{lCr}
	\label{rmse}
	RMSE_j\close{\Omega}=\rmse{\hat{\beta}_{i,j}^s}{\beta_j^s},
\end{IEEEeqnarray}

where $\hat\beta_{i,j}^s$ and $\beta_{j}^s$ stand for the  substantive (s) equation's $j$'th coefficient estimated in the $i$'th sample and the coefficient in the theoretical model that would have been obtained in the entire population, respectively. $\Omega$ is the number of data sets generated for the Monte Carlo simulations, which is $5000$ data sets (each one consists of $N$ observations). 

Another measure is based on a formula similar to the one described in \eqref{rmse}, and is intended to find the relative accuracy of the truncated sample's estimates, in comparison to full sample estimates, defined as:
\begin{IEEEeqnarray}{lCr}
	\label{eqprobrmse}
	R_j\close{\Omega}=\rmse{\hat{\beta}_{i,j}^{ts}}{\hat{\beta}_{i,j}^{s}},
\end{IEEEeqnarray}

where $\hat\beta_{i,j}^{ts}$ and $\hat\beta_{i,j}^s$ stand for the substantive (s) equation's $j$'th coefficient estimated using the truncated (t) sample and the full sample, respectively. This measure evaluates the relative model's performance in the truncated sample, with respect to the conventional IV using the full sample.

The last estimates' accuracy measure is the $\delta$ coefficient used for the calculation of the estimators' standard deviations convergence rate $n^\delta$ with respect to the sample size. It depicts the speed of standard deviation's shrinkage resulting from increasing the sample size. This coefficient is calculated based on the following ratio:
\begin{IEEEeqnarray}{lCr}
	\label{ndeltaConverence}
	\delta=\fracBig{\ln{\close{\nicefrac{\BBigger{\sigma_1}}{\BBigger{\sigma_2}}}}}{\ln{\close{\nicefrac{\BBigger{n_2}}{\BBigger{n_1}}}}},
	%\frac{\ln{\frac{\sigma_1}{\sigma_2}}}{\ln{\frac{n_2}{n_1}}}
\end{IEEEeqnarray}

\vspace{-2em}
\begin{IEEEeqnarray}{lCr}
	\nonumber
\end{IEEEeqnarray}

where $\sigma_1$ and $\sigma_2$ are the estimate's standard deviations that are calculated for data sets with $n_1$ and $n_2$ number observations, respectively (calculated for a given estimate).

% two steps JPIV estimator utilizing only a truncated data consisting of participants.
%the $k$'th data set ($k=1,...,5000$) is estimated using the full sample  
% using the complete data distribution function, such that $\beta_1=$ and $\beta_2=$, while the mean of the full sample IV's estimates are $\beta_1=0.996$ and $\beta_2=1.25$. The standard deviation obtained for $\beta_1$ (the endogenous covariate's coefficient) using the IV estimator is almost $3$ times larger than in the $OLS$ estimator and decreases from $0.4392$ to $0.0966$ when the sample size increases from $500$ to $10,000$ observations. 

\linespread{0.1}
\begin{table}[H]
	
	\caption{\label{tabD}Monte Carlo Simulation - Convergence measures}
	
	\begin{center}\scalebox{1}[0.7]{\begin{adjustbox}{width=25em}\begin{tabular}{|c|cccccc|}\hline\hline 
					\multicolumn{1}{|c|}{\multirow{1}{*}{}} & \multicolumn{6}{c|}{Model's estimates} \\ \cline{2-7} 
					\multicolumn{1}{|c|}{\multirow{1}{*}{Parameter}} & \multicolumn{6}{c|}{Number of observations} \\ \cline{2-7}
					%\multicolumn{1}{|c|}{\multirow{1}{*}{estimate}} &
					\multicolumn{1}{|c|}{} & \multicolumn{1}{c|}{500} & \multicolumn{1}{c|}{2000} & \multicolumn{1}{c|}{3000} & \multicolumn{1}{c|}{5000} & \multicolumn{1}{c|}{8000} & \multicolumn{1}{c|}{10000}  \\ \hline
					
					%	\hline\multicolumn{7}{|c|}{Truncated sample OLS estimates}\\\cline{1-7} 	
					\hline\multicolumn{7}{|c|}{\textbf{RMSE measure}}\\\cline{1-7} 				
					\hline\multicolumn{7}{|c|}{$\textbf{Full}$ sample $\textbf{OLS}$ estimates - IV correction}\\\cline{1-7} 	
					${\beta_1}$ & 0.4395 & 0.3069 & 0.2172 & 0.1393 & 0.1115 & 0.1005\\\hline
					${\beta_2}$ & 0.4321 & 0.3013 & 0.2123 & 0.1353 & 0.1092 & 0.0969\\\hline
					
					\hline\multicolumn{7}{|c|}{$\textbf{Truncated}$ sample $\textbf{OLS}$ estimates - IV correction}\\\cline{1-7} 	
					${\beta_1}$ & 1.1702 & 0.9749 & 0.8818 & 0.8392 & 0.8341 & 0.8315\\\hline
					${\beta_2}$ & 0.9590 & 0.7882 & 0.6861 & 0.6363 & 0.6284 & 0.6253\\\hline
					%			\hline\multicolumn{7}{|c|}{$\textbf{Truncated}$ sample \textbf{JPEG \textit{IV}} estimates (First stage)}\\\cline{1-7} 
					%${\delta_1}$ & 0.1785 & 0.1324 & 0.1087 & 0.0870 & 0.0670 & 0.0565\\\hline
					%${\delta_2}$ & 0.0654 & 0.0522 & 0.0451 & 0.0264 & 0.0209 & 0.0187\\\hline
					
					\hline\multicolumn{7}{|c|}{$\textbf{Truncated}$ sample \textbf{JPEG \textit{IV}} estimates}\\\cline{1-7} 
					${\beta_1}$ & 0.7328 & 0.3639 & 0.2968 & 0.2177 & 0.1731 & 0.1575\\\hline
					${\beta_2}$ & 0.6524 & 0.3271 & 0.2658 & 0.1946 & 0.1546 & 0.1402\\\hline
					%${\gamma_2}$ & 2.7116 & 1.0675 & 0.8544 & 0.5979 & 0.3866 & 0.3167\\\hline
					
					\hline\multicolumn{7}{|c|}{}\\\cline{1-7} 				
					\hline\multicolumn{7}{|c|}{$\B{R_j\close{n}}$ \textbf{measure - relative to full sample IV}}\\\cline{1-7} 				
					
					\hline\multicolumn{7}{|c|}{$\textbf{Truncated}$ sample $\textbf{OLS}$ estimates - IV correction}\\\cline{1-7} 	
					${\beta_1}$ & 5.5322 & 1.0441 & 0.9006 & 0.8434 & 0.8354 & 0.8328\\\hline
					${\beta_2}$ & 1.5804 & 0.7607 & 0.6797 & 0.6382 & 0.6266 & 0.6217\\\hline
					
					%\hline\multicolumn{7}{|c|}{Truncated sample model's estimates (First stage SPNLS)}\\\cline{1-7} 
					
					\hline\multicolumn{7}{|c|}{$\textbf{Truncated}$ sample \textbf{JPEG \textit{IV}} estimates}\\\cline{1-7} 
					${\beta_1}$ & 4.5652 & 0.5544 & 0.3176 & 0.2015 & 0.1588 & 0.1422\\\hline
					${\beta_2}$ & 1.4772 & 0.4244 & 0.2675 & 0.1721 & 0.1349 & 0.1210\\\hline
					
					\hline\multicolumn{7}{|c|}{}\\\cline{1-7} 		
					\multicolumn{7}{|c|}{$\B{\delta}$ \textbf{consistency measure} ($\B{n^{\delta}}\equiv$ \textbf{ the convergence rate})}\\         \cline{1-7} 
					\hline\multicolumn{7}{|c|}{$\textbf{Truncated}$ sample $\textbf{OLS}$ estimates - IV correction}\\\cline{1-7} 	
					${\beta_1}$ & - & 0.7627 & 0.3538 & 0.2949 & 0.2254 & 0.1793 \\\hline 	
					
					\hline\multicolumn{7}{|c|}{$\textbf{Truncated}$ sample \textbf{JPEG \textit{IV}} estimates}\\\cline{1-7}  
					${\beta_1}$ & - & 0.5040 & 0.5017 & 0.6072 & 0.4875 & 0.4201
					\\\hline

	\end{tabular}\end{adjustbox} }\end{center} 
	\scriptsize\begin{minipage}{8.8cm}%
		$\myspace$ \\ \textbf{Note:}  The parameters that are used in the data generation process.\\ We examine three different measures for the parameters presented in Table \eqref{tabD}. First, the standardized root mean square error $RMSE$ between each model's estimates and the true parameters is calculated based on equation \eqref{rmse}. Second, the $R_j$ measure is calculated based on equation \eqref{eqprobrmse}. Third, the convergence rate which is measured $n^\delta$ is calculated for both the conventional IV as well as the JPEG IV estimates. This convergence rate measure implies that multiplying the sample size by 2 shrinks the estimators' standard deviations by $2^{\delta}$.
	\end{minipage}		
	\footnotesize
	\renewcommand{\baselineskip}{11pt}				
\end{table}

Table \ref{tabD} entries indicate that the root mean squares error (RMSE) measure of the estimates obtained by employing the conventional IV estimator, using a random sample from the entire data distribution function, gets smaller as the sample size increases, as can be expected. However, applying the same procedure to the truncated data set leads to RMSE measures, which are in the range of 2 to 8-fold larger, given a sample size of $2,000$ to $10,0000$ observations, respectively. This is indeed a huge bias generated by the conventional IV, which is not immune to truncation bias. Additionally, the RMSE measures show negligible improvements as a function of number of observations for the conventional IV, whereas there is a huge improvement of the RMSE, as a function of the number of observations for the JPEG IV estimator provided by our model.
The proximity between the JPEG IV and the full sample IV estimates increases with the sample size, as reflected by the $R_j$ proximity measure. Using the same measure, we find that there is a much smaller improvement in the proximity between the truncated sample conventional IV  and the full sample IV, relative to the improvement in the proximity between the JPEG IV and the full sample IV estimates.%, which is hardly improved as the sample size increases.  
% Similar results are obtained by using the $R_j$ measures, which measures the relative proximity of the JPEG IV and the conventional IV in the truncated sample to the full sample IV estimates.  

It is evident that JPEG IV  is a $\sqrt{n}$ consistent estimator, as depicted by the $\delta$ consistency measure, which is about $0.5$, implying that multiplying the sample size by 2 shrinks the estimators' standard deviations by $2^{\delta}=\sqrt{2}$. It is also evident that the truncated data conventional IV is poorly functioning in terms of consistency, as is shown by entries in Table  \ref{tabD}. %In other words, the consistency is not asymptotically improved.

\section{Conclusion}\label{Section:Conclusion}
%This paper extends the literature on instrumental variables for endogenously truncated data. We introduce a two stage estimation procedure, utilizing a wavelet-based JPEG IV estimator to capture the bias term. 
We provide an analytical proof showing that in an endogenously truncated data the conventional IV estimator does not perform the task it was intended to, but rather introduces an additional unintended bias into the parameter estimates of the substantive equation. The instrumental variable is endogenous by itself in the context of endogenously truncated data due to a comovement between the instrumental variable and the substantive equation's random disturbance, generated by mediating covariates. We offer a truncation-proof JPEG IV, shown to be a proper estimator under endogenous truncation. Employing Monte Carlo simulations attests to the JPEG IV estimator's high accuracy and its $\sqrt{n}$ consistency. These results have been verified by utilizing 2,000,000 different distribution functions (not restricted to the unimodal symmetric family), generating 100 million realizations to construct the covariates' data sets which are not imposed to be i.i.d. The various distribution functions  attest to  a very high accuracy of the model as depicted by the parameter estimates that closely mimic the true parameters.

\vspace{2em}

\parskip = 0pt
%\bibliography{reference2}		
%	\vspace{-5em}

%\bibliographystyle{apalike}

%\pagebreak
\appendix
\section{Appendices}
\vspace{1em}
\subsection{An element-wise thresholding}\label{AppendixA}

In this appendix we show that for any $\gamma\in(1,\infty)$, $S_{\alpha}(\cdot)$ in \eqref{S:Thresholding:Solution} is the solution to the min-max concave penalty function in \eqref{MCP} $\forall\alpha\in(1/\gamma,\infty)$.

%\section*{Appendix A}

\begin{proof}
	Let $\B{\tilde{\delta}}\equiv\B{\delta}^{(iter)}-1/(\alpha n)\B{\Psi}_{_I}^T(\B{u}-\B{\Psi}_{_I}\B{\delta}^{(iter)})$	
	\begin{IEEEeqnarray}{lCr}
		\scaleobj{0.7}{\B{\delta}=\underset{\B{\delta}}{\arg\min}\frac{1}{2}\normsq{\B{\delta}-\B{\tilde{\delta}}}{2}+\frac{1}{\alpha}\Penalty{\B{\delta}} }
	\end{IEEEeqnarray}

	As $\Penalty{\B{\delta}}$ is a separable function of $\B{\delta}$, the minimization problem can be implemented in an element-wise fashion. Let $\delta_k$ be the solution to the following univariate regularized least squares problem:
	\begin{IEEEeqnarray}{lCr}
		\label{Element:wise}
		\scaleobj{0.7}{\delta_k=\begin{cases}
				\underset{\delta}{\arg\min}\frac{1}{2}\normsq{\delta-\tilde{\delta}_k}{2}+	\lambda \abs{\delta} - \frac{\delta^2}{2\gamma} & \text{if }\abs{\delta_k} \le \lambda\gamma\\
				\underset{\delta}{\arg\min}\frac{1}{2}\normsq{\delta-\tilde{\delta}_k}{2}+	\frac{1}{2}\lambda^2\gamma & \text{if }\abs{\delta_k} > \lambda\gamma	
		\end{cases}}
	\end{IEEEeqnarray}

	We obtain the first order condition of \eqref{Element:wise} with respect to $\delta$:
	\begin{IEEEeqnarray}{lCr}
		\label{Univariate:Solution}
		\scaleobj{0.7}{\delta_k = \begin{cases}
				\frac{1}{1-1/(\alpha\gamma)}\braket{\tilde{\delta}_k - \text{sign}(\delta_k)\frac{\lambda}{\alpha}}  & \text{if } \abs{\delta_k} \le \lambda\gamma\\
				\tilde{\delta}_k & \text{if } \abs{\delta_k} > \lambda\gamma
		\end{cases}}
	\end{IEEEeqnarray}
	
\newcommand\Smallabs[1]{\big|{#1} \big|}
Note that if $\Smallabs{\tilde{\delta}_k}>\lambda\gamma$ and $\abs{\delta_k}\le\lambda\gamma$ it implies that either $\text{sign}(\delta_k)=-1$ and $-\lambda\gamma\le\tilde{\delta_k}\le \lambda\gamma-2\lambda/\alpha$ or $\text{sign}(\delta_k)=1$ and $-\lambda\gamma+2\lambda/\alpha\le\tilde{\delta_k}\le \lambda\gamma$. Both cases contradict the fact that $\Smallabs{\tilde{\delta}_k}>\lambda\gamma$. Similarly, if $\Smallabs{\tilde{\delta}_k}\le\lambda\gamma$ and $\abs{\delta_k}>\lambda\gamma$, it contradicts the fact that $\delta_k=\tilde{\delta_k}$ which follows directly from \eqref{Univariate:Solution}. Consequently, $\Smallabs{\tilde{\delta}_k}>\lambda\gamma$ $\Leftrightarrow$ $\abs{\delta_k}>\lambda\gamma$. We get:
	\begin{IEEEeqnarray}{lCr}
		\label{Delta:Cutoffs}
		\delta_k=\begin{cases}
			\frac{\tilde{\delta}_k-\frac{\lambda}{\alpha}}{1-1/(\alpha\gamma)} \hspace{1em} & \text{if }0 < \frac{\tilde{\delta}_k-\frac{\lambda}{\alpha}}{1-1/(\alpha\gamma)} \le \lambda\gamma\\
			\frac{\tilde{\delta}_k+\frac{\lambda}{\alpha}}{1-1/(\alpha\gamma)} \hspace{1em} &\text{if }-\lambda\gamma \le \frac{\tilde{\delta}_k+\frac{\lambda}{\alpha}}{1-1/(\alpha\gamma)} <0
		\end{cases}
	\end{IEEEeqnarray}
	
	%\delta_k  =  \tilde{\delta} & \text{if } \tilde{\delta} > \lambda\gamma
	
	and provided that  $\forall \alpha \in (1/\gamma,\infty)$ \eqref{Delta:Cutoffs} is simplified to:
	\begin{IEEEeqnarray}{lCr}
		\delta_k=\begin{cases}
			\frac{\tilde{\delta}_k-\frac{\lambda}{\alpha}}{1-1/(\alpha\gamma)}  & \text{if }\frac{\lambda}{\alpha} < \tilde{\delta}_k \le  \lambda\gamma\\
			\frac{\tilde{\delta}_k+\frac{\lambda}{\alpha}}{1-1/(\alpha\gamma)}  &\text{if }-\lambda\gamma \le \tilde{\delta}_k < -\frac{\lambda}{\alpha}
		\end{cases}
	\end{IEEEeqnarray}
	
	%Second, we consider the cases where $\abs{\delta_k} > \lambda\gamma$ taking the first order condition of \eqref{Element:wise} with respect to $\delta$:
	%\begin{IEEEeqnarray}{lCr}
	
	%\delta_k = \tilde{\delta}_k
	%\end{IEEEeqnarray}

	After some algebraic manipulation we obtain:
	\begin{IEEEeqnarray}{lCr}
		\scaleobj{0.7}{\delta_k = \begin{cases}
				\frac{1}{1-1/(\alpha\gamma)}\text{sign}(\tilde{\delta}_k)(\abs{\tilde{\delta}_k} - \lambda/\alpha) & \text{if }  \lambda/\alpha < \abs{\tilde{\delta}_k} \le \lambda\gamma \\
				\tilde{\delta}_k & \text{if } \abs{\tilde{\delta}_k} > \lambda\gamma \\ 
		\end{cases}}
	\end{IEEEeqnarray}
\end{proof}	
	
	%Provided that $\abs{\tilde{\delta}_k}\le\lambda\gamma$ and $\alpha \in (1/\gamma,\infty)$, we want to show that  $\tilde{\delta}_k > \lambda/\alpha$ implies $\delta_k>0$ as well as that  $\tilde{\delta}_k<-\lambda/\alpha$ implies $\delta_k < 0$.
	%We use \eqref{Univariate:Solution} to show that if $\tilde{\delta}_k-\lambda/\alpha >0$, after some algebraic manipulation, it follows that $\delta_k > 0$ as both $\lambda$ as well as $\alpha$ are positive. A similar argument is applied to $\tilde{\delta}_k+\lambda/\alpha <0$ which implies  $\delta_k<0$.  Consequently, 
	%\begin{IEEEeqnarray}{lCr}
	
	%\delta_k = S_{\alpha}(\tilde{\delta}_k; \lambda, \gamma) = \begin{cases}
	%\begin{cases}
	%\frac{1}{1-1/(\alpha\gamma)}\text{sign}(\tilde{\delta}_k)(\abs{\tilde{\delta}_k} - \lambda/\alpha) & \text{if }\abs{\tilde{\delta}_k} > \lambda/\alpha\\
	%0 & \text{if }\abs{\tilde{\delta}_k} \le \lambda/\alpha
	%\end{cases} & \text{if } \abs{\tilde{\delta}_k}  \le \lambda\gamma \\
	%\tilde{\delta}_k & \text{if } \abs{\tilde{\delta}_k}  >  \lambda\gamma
	%\end{cases}
	%\end{IEEEeqnarray}
	
	Next we show how to obtain an analytic representation of the transpose of the wavelet inverse transform. For doing so, we based our arguments on the equivalence between the matrices-product and the lifting scheme representations of the wavelet transform. Then using this equivalence, we generate a filter to render the costly matrices building useless. Our objective is to reduce computational complexity.

Next we present the first stage estimates.
	\subsection{First stage estimates}\label{Appendix:First:Stage}
\raggedbottom
\linespread{0.1}
\begin{table}[H]
	
	\caption{\label{tabE}Monte Carlo Simulation - Truncated data set with truncation bias correction}
	
	\begin{center}\scalebox{1}[0.7]{\begin{adjustbox}{width=25em}\begin{tabular}{ccccccccccc|}
					\hline
					\multicolumn{3}{|l|}{\multirow{3}{*}{True Parameter$^{\mathrm{a}}$}} & \multicolumn{2}{l|}{\multirow{3}{*}{Estimate}} & \multicolumn{6}{c|}{Model Setup}                                                                                                                                       \\ \cline{6-11} 
					\multicolumn{3}{|l|}{}                              & \multicolumn{2}{l|}{}                          & \multicolumn{6}{c|}{Sample size}                                                                                                                                       \\ \cline{6-11} 
					\multicolumn{3}{|l|}{}                              & \multicolumn{2}{l|}{}                          & \multicolumn{1}{c|}{500} & \multicolumn{1}{c|}{2000} & \multicolumn{1}{c|}{3000} & \multicolumn{1}{c|}{5000} & \multicolumn{1}{c|}{8000} & \multicolumn{1}{c|}{10000} \\ \hline

					\hline\multicolumn{11}{|c|}{$\textbf{Truncated}$ sample \textbf{JPEG \textit{IV}} estimates (First stage)}\\\cline{1-11} 
					\multicolumn{3}{|c|}{\multirow{3}{*}{$\delta_1=0.5$}} & \multicolumn{2}{c|}{Mean} & 0.4944 & 0.4975 & 0.4973 & 0.4996 & 0.5000 & 0.4999 \\\cline{4-11} \multicolumn{3}{|c|}{} & \multicolumn{2}{c|}{Median}  & 0.4951 & 0.4963 & 0.4971 & 0.4999 & 0.4998 & 0.5000 \\\cline{4-11} \multicolumn{3}{|c|}{} & \multicolumn{2}{c|}{Std}  & 0.1236 & 0.0675 & 0.0514 & 0.0461 & 0.0352 & 0.0298 \\\hline		
					
					\multicolumn{3}{|c|}{\multirow{3}{*}{$\delta_2=1$}} & \multicolumn{2}{c|}{Mean} & 0.9976 & 0.9983 & 0.9994 & 1.0001 & 1.0005 & 0.9999\\\cline{4-11} \multicolumn{3}{|c|}{} & \multicolumn{2}{c|}{Median}  &0.9977 & 0.9989 & 0.9989 & 1.0001 & 1.0006 & 1.0003\\\cline{4-11} \multicolumn{3}{|c|}{} & \multicolumn{2}{c|}{Std}  & 0.0902 & 0.0444 & 0.0360 & 0.0279 & 0.0221 & 0.0197\\\hline
					
					%${\gamma_2=-1}$ & -1.0169 & -1.0332 & -1.0363 & -1.0311 & -1.0111 & -1.0119\\{$\text{median}$}  & -0.9904 & -1.0026 & -1.0057 & -1.0010 & -0.9980 & -0.9960\\{$\text{std}$}  & (0.4276) & (0.4301) & (0.4194) & (0.3847) & (0.3145) & (0.2612)\\\hline

	\end{tabular}\end{adjustbox} }\end{center}  
	\scriptsize\begin{minipage}{8.7cm}%
		$\myspace$ \\ \textbf{Note:}    $^{\mathrm{a}}$ The parameters that are used in the data generation process.\\ We estimate the parameters using the  truncated sample by employing our JPEG IV estimator, correcting for both truncation as well as endogeneity bias. We calculate for these estimates the mean, median and standard deviation (Std.) over all data sets. The standard deviations are obtained using the estimates from the Monte Carlo simulations.
	\end{minipage}		
	\footnotesize
	\renewcommand{\baselineskip}{11pt}			
\end{table} 

%\vspace{1em}

\begin{algorithm}[H]
	\caption{The wavelet-based JPEG penalized regression}\label{PenalizedGradientDecent}
	\begin{algorithmic}[1]
		\tiny
		\Procedure{PenalizedLinearRegression}{u, Grid, Level, $\lambda$, $\gamma$}
		\Require (i) Two real-number vectors \textit{u} and Grid of size $n\times 1$; (ii) Level$\in\curly{1,...,\log_2(n)}$;\\ (iii) $\lambda\in(0,\infty)$ and $\gamma\in(1,\infty)$.	
		\Ensure $\text{Output}\gets$ a real-number coefficient vector $\widehat{\delta}$ of size $n\times 1$;
		
		%\Procedure{Euclid}
		
		\State $\textit{n} \gets \text{length of }\textit{u}$
		
		\State \emph{top}:
		\State $\delta_{\text{old}}$ $\gets$ \Call{ForwardTransform}{u, Grid, Level-1}
		
		\State $\widehat{u}_{\text{old}}$ $\gets$ u %\Call{InverseTransform}{$\delta_{\text{old}}$, Grid, 1}
		\State residual$_{\text{old}}$ $\gets$ u - $\widehat{u}_{\text{old}}$  
		\State Obj$_{\text{old}}$ $\gets$ $\frac{1}{2n}\sum_{i=1}^n$ residual$_{\text{old}}^2$[i]  + $\sum_{i=1}^n$ \Call{MCP.penalty}{$\abs{\delta_{\text{old}}[i]}$,$\lambda$, $\alpha$, $\gamma$)}
		\State v$_{\text{old}} \gets$ \Call{TransInverseTransform}{residual$_{\text{old}}$, Grid, Level-1} 
		\State flag $\gets$ TRUE	 
		\State gap $\gets$ infinite
		\State tolerance $\gets$ $10^{-9}$
		\State maxiter $\gets$ 1000
		\State $\eta$ $\gets$ 1.2
		\While{flag=TRUE and gap $>$ tolerance}
		\State $\alpha$ $\gets$ 1
		\State Iter $\gets$ 1	 
		\Repeat
		\State Update $\gets$ $\delta_{\text{old}}$ + $\frac{1}{\alpha n}$ v$_{\text{old}}$ 
		\State $\delta_{\text{new}}$ $\gets$ $S_{\alpha}$(Update, $\lambda$, $\alpha$, $\gamma$)
		\State $\widehat{u}_{\text{new}}$ $\gets$ \Call{InverseTransform}{$\delta_{\text{new}}$, Grid, Level-1}
		\State residual$_{\text{new}}$ $\gets$ u - $\widehat{u}_{\text{new}}$
		
		\State Obj$_{\text{new}}$ $\gets$ $\frac{1}{2n}\sum_{i=1}^n$ residual$_{\text{new}}^2$[i] + $\sum_{i=1}^n$ \Call{MCP.penalty}{$\abs{\delta_{\text{new}}[i]}$,$\lambda$, $\alpha$, $\gamma$)}
		\State $\alpha \gets \alpha \eta$
		\State flag $\gets$ Obj$_{\text{new}}$ $<$ Obj$_{\text{old}}$
		\Until{flag=TRUE or $iter > maxiter$}
		
		\If {flag = TRUE} 
		\State $\delta_{\text{old}}$ $\gets$ $\delta_{\text{new}}$
		\State Obj$_{\text{old}}$ $\gets$ Obj$_{\text{new}}$
		\State residual$_{\text{old}}$ $\gets$ residual$_{\text{new}}$
		\State v$_{\text{old}} \gets$ \Call{TransInverseTransform}{residual$_{\text{new}}$, Grid, Level-1} 
		\EndIf
		
		\State gap $\gets \normsq{\delta_{\text{new}}-\delta_{\text{old}}}{2}/\normsq{\delta_{\text{old}}}{2}$
		\State Iter $\gets$ Iter + 1		
		\EndWhile
		%$\close{\widehat{u}_{\text{old}},
		\State\Return{$\close{\delta_{\text{old}}}$}
		\EndProcedure
	\end{algorithmic}
\end{algorithm}

\subsection{Algorithms}
\begin{algorithm}[H]
	\tiny
	\caption{The Transposed-Inverse filter}\label{euclid}
	\begin{algorithmic}[1]
		\Procedure{TransInvFilter}{Series, Even, $\pi$, Grid}
		\Require (i) Two real-number vectors: Series of size $n\times 1$ and Grid of size $2n\times 1$; (ii) two scalars: $\text{Even}\in\curly{0,1}$ and $\pi\in\mathbb{R}$.	
		\Ensure $\text{Output}\gets$ Filter, a real-number vector of size $n\times 1$ consisting of the predicted series;			
		\State n $\gets$ length of Series
		\State  O $\gets$ copy the odd elements of Grid 
		\State E $\gets$ copy the even elements of Grid 			
		%        \State OSeries $\gets\curly{\text{OddGrid}, \text{OddGrid}}$
		\If {Even $=$ 1}
		%	\State $(new.u, new.Grid)\gets \text{ForwardTransform}\close{u[1:m], Grid[1:m], Level-1}$
		\State Low $\gets$ O[1:n-1]
		\State High $\gets$ O[2:n]
		\State weights $\gets$ (High-E[1:(n-1)])/(High-Low)
		\State $\varpi_l\gets$ [0, 1-weights]
		\State $\varpi_h\gets$ $\gets$ [weights, 1]
		\State S $\gets$ [Series[1], Series]		
		\State Filter $\gets$ $2\pi$ $\varpi_l\gets$*S[1:n] + $2\pi$ $\varpi_h\gets$*S[2:(n+1)]
		%		\State $j \gets j-1$.
		%		\State $i \gets i-1$.
		%		\State \textbf{goto} \emph{loop}.
		%		\State \textbf{close};
		\Else
		\State Low $\gets$ E[1:(n-1)]
		\State High $\gets$ E[2:n]
		\State weights $\gets$ (High-O[2:n])/(High-Low)
		\State $\varpi_l\gets$ [1, 1-weights]
		\State $\varpi_h\gets$ [weights, 0]
		\State S $\gets$ [Series, Series[n]]		
		\State Filter $\gets$ $2\pi$ $\varpi_l$*S[1:n] + $2\pi$ $\varpi_h$*S[2:(n+1)]		
		\EndIf
		%\State $i \gets i+\max(\textit{delta}_1(\textit{string}(i)),\textit{delta}_2(j))$.
		%\State \textbf{goto} \emph{top}.
		\State\Return{$\close{\text{Filter}}$}
		\EndProcedure
	\end{algorithmic}
\end{algorithm}

%\subsection{Algorithms}

\section*{Acknowledgment}
We would like to thank Boaz Nadler, Yaniv Tenzer and seminar participants in the Weizmann  Institute of Science and Tel-Aviv university for very constructive comments.
%We thank Larry Manevitz for very constructive comments and Omiros Papaspiliopoulos for very constructive conversations. 

	\bibliographystyle{ieeetran}
\linespread{1.1}
	\bibliography{reference2}
\linespread{0.7}	
	\begin{IEEEbiography}[{\includegraphics[width=2in,height=1in,clip,keepaspectratio]{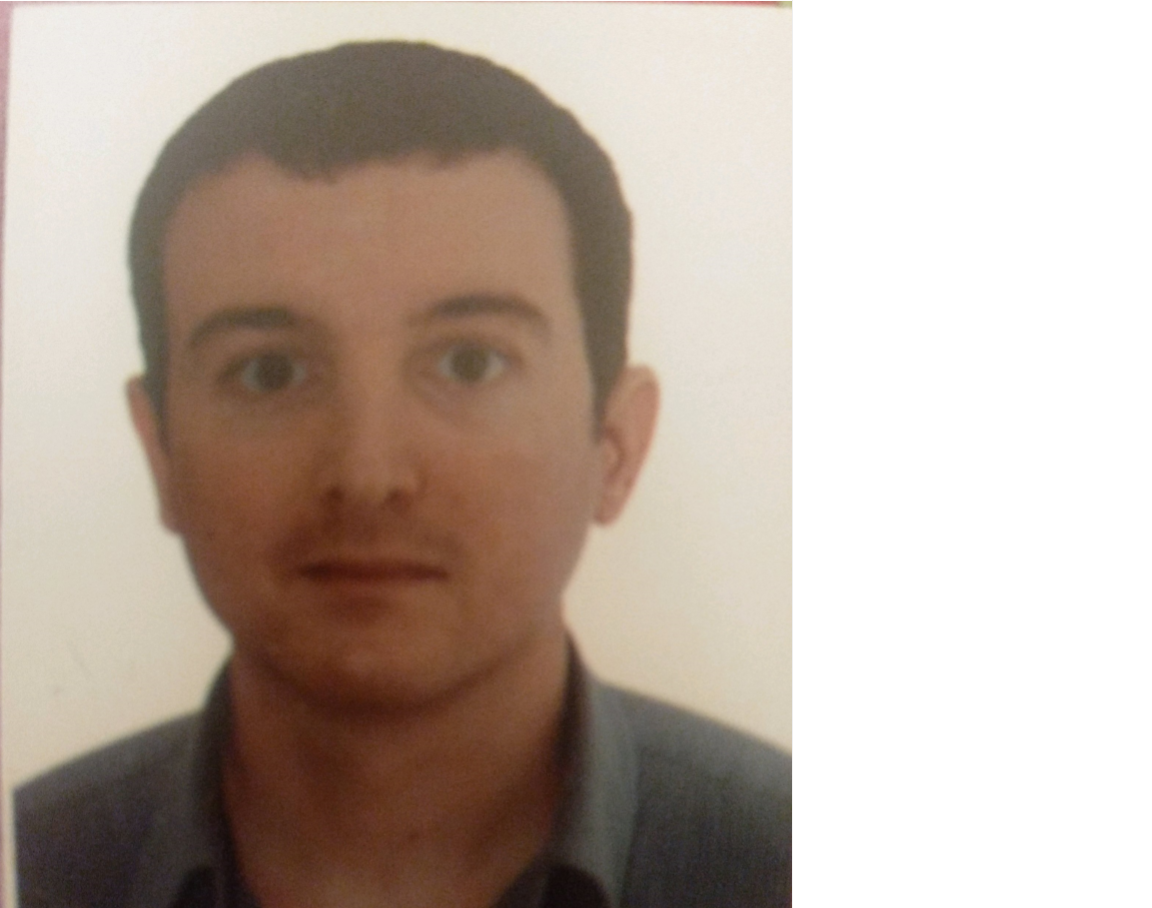}}]{D\lowercase{r}. Nir Billfeld} is a researcher at the university of Haifa, Israel. He received the B.A. in economics and statistics from the university of Haifa (2006), Israel, M.A. in economics from Tel-Aviv university (2010). He received his Ph.D. from the university of Haifa (2019). His work has appeared in IEEE. 
	\end{IEEEbiography}

	\begin{IEEEbiography}[{\includegraphics[width=1in,height=1.25in,clip,keepaspectratio]{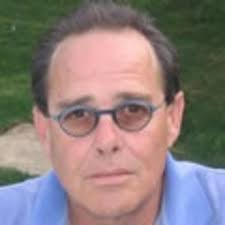}}]{P\lowercase{rof.} Moshe Kim} is professor of economics at the University of Haifa, Israel. He is the founder and former director of Barcelona Banking Summer School at the Universitat Pompeu Fabra, Barcelona, former director of the endowed chair of banking at Humboldt University of Berlin, Senior Distinguished Fellow at the Swedish School of Economics in Helsinki (Hanken), institute research professor at the German Institute for Economic Research (DIW), consultant at the Central Bank of Norway and recently visited NYU Shanghai. He was recently declared high end foreign expert by the Chinese foreign ministry and is a recent recipient of the Outstanding Tutor Award from the Chinese Ministry of Education. He holds a PhD from the University of Toronto. Kim's research interests are econometrics, banking, financial markets, and industrial organization. His books include Microeconometrics of Banking: Methods, Applications, and Results (Oxford University Press, 2009). His work has appeared in IEEE, the Journal of Finance, the Journal of Monetary Economics, the Journal of Financial Intermediation, the Journal of Business and Economics Statistics, the Journal of Money Credit and Banking, International Economic Review, the Journal of Accounting and Economics, the Journal of Public Economics, the Journal of Urban Economics, the Journal of Law and Economics, the Journal of Banking and Finance, the Journal of Industrial Economics, the International Journal of Industrial organization. 
	\end{IEEEbiography}

	\EOD
	
\end{document}